\documentclass[11pt]{article}
\usepackage[letterpaper,hmargin=1in,vmargin=1.25in]{geometry}

\usepackage{url}

\usepackage{amsmath}
\usepackage{amssymb}

\usepackage{amsthm}

\theoremstyle{plain}
\newtheorem{theorem}{Theorem}[section]
\newtheorem{lemma}[theorem]{Lemma}
\newtheorem{corollary}[theorem]{Corollary}

\theoremstyle{definition}
\newtheorem{definition}[theorem]{Definition}

\newtheorem{protocol}[theorem]{Protocol}
\newtheorem{postulate}{Postulate}

\renewcommand{\labelenumi}{(\roman{enumi})}

\newcommand{\abs}[1]{\left\lvert#1\right\rvert}

\newcommand{\rest}[2]{#1\!\!\restriction_{#2}}

\newcommand{\osg}[1]{\left[#1\right]^{\prec}}

\newcommand{\N}{\mathbb{N}}%
\newcommand{\R}{\mathbb{R}}%
\newcommand{\C}{\mathbb{C}}%
\newcommand{\X}{\{0,1\}^*}%
\newcommand{\XI}{\{0,1\}^\infty}%

\newcommand{\Bm}[2]{\lambda_{#1}\left(#2\right)}

\newcommand{\ket}[1]{| #1 \rangle}
\newcommand{\bra}[1]{\langle #1 |}
\newcommand{\braket}[2]{\langle #1 | #2 \rangle}

\newcommand{\PS}{\mathbb{P}}%
\newcommand{\charaps}[2]{C\!\left(#1,#2\right)}
\newcommand{\chara}[2]{\mathrm{C}_{#1}\left(#2\right)}
\newcommand{\cond}[2]{\mathrm{Filtered}_{#1}\left(#2\right)}

\DeclareMathOperator{\tr}{tr}

\newcommand{\ssoa}{\overline{\mathcal{H}}}
\newcommand{\ancilla}{\mathcal{S}^{\mathrm{a}}}

\DeclareMathOperator{\Prob}{Pr}

\newcommand{\noi}{\noindent}

\begin{document}


\renewcommand{\thefootnote}{\fnsymbol{footnote}}
\begin{center}
{\Large \textbf{A refinement of quantum mechanics by algorithmic randomness}}
\end{center}
\renewcommand{\thefootnote}{\arabic{footnote}}
\setcounter{footnote}{0}

\vspace{0mm}

\begin{center}
Kohtaro Tadaki
\end{center}

\vspace{-5mm}

\begin{center}
Department of Computer Science, College of Engineering, Chubu University\\
1200 Matsumoto-cho, Kasugai-shi, Aichi 487-8501, Japan\\
E-mail: \textsf{tadaki@isc.chubu.ac.jp}\\
\url{http://www2.odn.ne.jp/tadaki/}
\end{center}

\vspace{-2mm}

\begin{quotation}
\noi\textbf{Abstract.}
The notion of probability plays a crucial role in quantum mechanics.
It appears in quantum mechanics as the Born rule.
In modern mathematics which describes quantum mechanics, however,
probability theory means nothing other than measure theory,
and therefore any operational characterization of the notion of probability is
still missing in quantum mechanics.
In this paper,
based on the toolkit of \emph{algorithmic randomness},
we present
a refinement of the Born rule, as an alternative rule to it,
for specifying the property of the results of quantum measurements \emph{in an operational way}.
Algorithmic randomness is a field of mathematics which
enables us to consider the randomness of an individual infinite sequence.
We then present
an operational refinement of the Born rule for mixed states, as an alternative rule to it,
based on algorithmic randomness.
In particular, we give a precise definition for the notion of mixed state.
We then show that all of
the refined rules of the Born rule for both pure states and mixed states
can be derived
from a single postulate, called the \emph{principle of typicality}, in a unified manner.
We do this from the point of view of the \emph{many-worlds interpretation of quantum mechanics}.
Finally, we make an application of our framework to the BB84 quantum key distribution protocol
in order to demonstrate how properly our framework works in practical problems
in quantum mechanics, based on the principle of typicality.
\end{quotation}

\begin{quotation}
\noi\textit{Key words\/}:
quantum mechanics,
Born rule,
probability interpretation,
algorithmic randomness,
operational characterization,
Martin-L\"of randomness,
many-worlds interpretation,
the principle of typicality,
quantum cryptography
\end{quotation}


\newpage
\tableofcontents 
\newpage

\section{Introduction}

The notion of probability plays a crucial role in quantum mechanics.
It appears in quantum mechanics as the so-called \emph{Born rule}, i.e.,
\emph{the probability interpretation of the wave function} \cite{D58,NC00}.
In modern mathematics which describes quantum mechanics, however,
probability theory means nothing other than \emph{measure theory},
and therefore any \emph{operational characterization of the notion of probability} is still missing
in quantum mechanics.
In this sense, the current form of quantum mechanics is considered to be \emph{imperfect}
as a physical theory which must stand on operational means.

In this paper, based on the toolkit of \emph{algorithmic randomness},
we present
a refinement of the Born rule as an alternative rule to it,
for aiming at
making quantum mechanics \emph{operationally perfect}.
Algorithmic randomness, also known as \emph{algorithmic information theory},
is a field of mathematics which enables us to consider the randomness of an \emph{individual}
infinite sequence \cite{Solom64,Kol65,C66,M66,Sch73,C75,N09,DH10}.
We use the notion of
\emph{Martin-L\"of randomness with respect to Bernoulli measure} \cite{M66} to present
the refinement of the Born rule.

\subsection{Operational characterization of the notion of probability}

In a series of works \cite{T14,T15,T16arXiv},
we presented
an
operational characterization of the notion of probability,
based on the notion of Martin-L\"of randomness with respect to Bernoulli measure.

To clarify our motivation and standpoint,
and the meaning of the \emph{operational characterization},
let us consider
a familiar example of a probabilistic phenomenon
\emph{in a completely classical setting}.
We
here
consider the repeated throwings of a fair die.
In this probabilistic phenomenon,
as throwings progressed,
a specific infinite sequence such as
\begin{equation*}
  3,5,6,3,4,2,2,3,6,1,5,3,5,4,1,\dotsc\dotsc\dotsc
\end{equation*}
is being generated,
where each number is the outcome of the corresponding throwing of the die.
Then the following
naive
question
may
arise naturally. 

\begin{quote}
\textbf{Question}:
What property should this infinite sequence
satisfy as a
probabilistic
phenomenon?
\end{quote}

In the series of works \cite{T14,T15,T16arXiv},
we tried to answer this question.
We there characterized the notion of probability as
an infinite sequence of outcomes in a probabilistic phenomenon
which has
a \emph{specific mathematical property}.
We called such an infinite sequence of outcomes
the \emph{operational characterization of the notion of probability}.
As the specific mathematical property,
in the works \cite{T14,T15,T16arXiv}
we adopted the notion of \emph{Martin-L\"of randomness with respect to Bernoulli measure},
a notion in algorithmic randomness.

In the works \cite{T15,T16arXiv}
we put forward this proposal as a thesis.
We then checked the validity of the thesis based on our intuitive understanding of the notion of probability.
Furthermore, we characterized \emph{equivalently}
the basic notions in probability theory in terms of the operational characterization.
Namely, we equivalently characterized the notion of the \emph{independence} of random variables/events
in terms of the operational characterization,
and represented the notion of \emph{conditional probability}
in terms of the operational characterization in a natural way.
The existence of these equivalent characterizations confirms further the validity of the thesis.

In the works \cite{T14,T15,T16arXiv},
we then made
applications of our
framework
to information theory and cryptography
as examples of the fields for the applications,
in order to demonstrate the \emph{wide applicability} of our
framework
to the \emph{general areas of science and technology}.
See Tadaki~\cite{T16arXiv}
for the detail of
our framework \cite{T14,T15,T16arXiv}
and its applications,
as well as its historical origins.

Modern probability theory originated from the \emph{axiomatic approach} to probability theory,
introduced by Kolmogorov~\cite{K50} in 1933,
where the probability theory is precisely \emph{measure theory}.
Since then,
it
has come a long way to become one of the most
active
fields in modern mathematics.
Such
a drastic development of
modern probability theory is partially
due to
the abandonment of the answer to the question above
about the operational meaning of the notion of probability, and
the
blind
identification of probability theory with measure theory.
One of the important roles of modern probability theory is, of course, in its \emph{applications} to
the general areas of science and technology.
As we have already pointed out,
however,
an operational characterization of the notion of probability is still missing in modern probability theory.
Thus, when we apply the results of modern probability theory,
we have no choice but to
make such applications \emph{thoroughly based on our intuition without formal means}.

The aim of
our framework \cite{T14,T15,T16arXiv}
is to try to fill in this gap between modern probability theory and its applications.
We there
proposed
the operational characterization of the notion of probability as
a \emph{rigorous interface} between theory and practice,
without appealing to our intuition for filling in the gap.
Anyway,
in the works \cite{T14,T15,T16arXiv}
we \emph{keep} modern probability theory \emph{in its original form without any modifications},
and propose the operational characterization of the notion of probability
as an \emph{additional mathematical structure} to it, which provides modern probability theory with
more comprehensive and rigorous opportunities for applications.
Of course, such applications of modern probability theory via our framework
include an application to \emph{quantum mechanics}.

\subsection{Algorithmic randomness}
\label{IntroAR}

\emph{Algorithmic randomness} is a field of mathematical logic.
It originated in the groundbreaking works of Solomonoff~\cite{Solom64}, Kolmogorov~\cite{Kol65}, and Chaitin~\cite{C66} in the mid-1960s.
They independently introduced the notion of \emph{program-size complexity}, also known as \emph{Kolmogorov complexity},
in order to quantify the randomness of an individual object.
Around the same time, Martin-L\"of \cite{M66} introduced a measure theoretic approach to characterize the randomness of an individual infinite binary sequence.
His
approach, called \emph{Martin-L\"of randomness} nowadays,
is one of the major notions in algorithmic randomness,
as well as
the notion of
program-size complexity.
Later on, in the 1970s
Schnorr~\cite{Sch73} and Chaitin~\cite{C75} showed that Martin-L\"of randomness is equivalent to the randomness defined by program-size complexity
in characterizing
random infinite binary sequences.
In the 21st century, algorithmic randomness makes remarkable progress through close interaction with recursion theory.
See \cite{N09,DH10} for the recent development as well as the historical detail of algorithmic randomness.
In this paper,
we use the notion of \emph{Martin-L\"of randomness with respect to Bernoulli measure},
a generalization of
Martin-L\"of randomness,
in order
to state the refined rule of the Born rule.

\subsection{Contribution of the paper: An operational refinement of the Born rule}

In this paper,
as a major application of
our framework
\cite{T14,T15,T16arXiv}
to
\emph{basic
science},
we present
the refined rule of the Born rule
based on our operational characterization of the notion of probability,
for the purpose of making quantum mechanics \emph{operationally perfect}.
Namely,
we use the notion of Martin-L\"of randomness with respect to Bernoulli measure \cite{M66} to state
the refined rule of the Born rule,
for specifying the property of the results of
quantum measurements \emph{in an operational way}.

In this paper,
as the first step of the research of this line,
we only consider, for simplicity,
the case of finite-dimensional quantum systems and measurements over them.
Note, however, that \emph{such a case is typical in
the field of
quantum information and quantum computation}
\cite{NC00}.

The contribution of the paper starts
with a reconsideration of
the form of the postulate of quantum measurements as it ought to be.
Consider
the nature of quantum measurements
from a general point of view.
We notice that
all that the experimenter of quantum measurements can obtain through the measurements
about quantum system is a \emph{specific} infinite sequence of outcomes of the measurements
which are being generated by
infinitely repeated measurements.
Thus, from an operational point of view,
the object about which the postulate of quantum measurements makes a statement should be
the \emph{properties of a
specific infinite sequence of outcomes of the measurements}.
Suggested by this consideration,
we introduce
an operational refinement of the Born rule,
Postulate~\ref{Tadaki-rule} in Section~\ref{SecTadaki-rule},
as an alternative rule to it,
based on the notion of Martin-L\"of randomness with respect to Bernoulli measure.

We then
check the validity of
the refined rule of the Born rule, Postulate~\ref{Tadaki-rule},
in Section~\ref{Validity}.
Based on
the results of the work \cite{T14,T15,T16arXiv},
we can see that
Postulate~\ref{Tadaki-rule}
is certainly a refinement of the Born rule,
from the point of view of our intuitive understanding of the notion of probability.
In the first place,
what is ``probability''?
In particular, what is ``probability'' in quantum mechanics?
It would seem very difficult to answer this question
\emph{completely} and \emph{sufficiently}.
However, we may enumerate the \emph{necessary} conditions which the notion of probability is
considered to have to satisfy
\emph{according to
our intuitive understanding of the notion of probability}.
We show that
the refined rule,
Postulate~\ref{Tadaki-rule}, satisfies these necessary conditions.

The refined rule of the Born rule, mentioned above,
is an operational refinement of the Born rule for \emph{pure states}.
Next, we consider an operational refinement of the Born rule for \emph{mixed states}
by algorithmic randomness
in Section~\ref{Mixed_States}.
We first note that, according to
the refined rule of the Born rule for pure states, Postulate~\ref{Tadaki-rule},
the result of the quantum measurements forms an infinite sequence of pure states which
is Martin-L\"of random with respect to Bernoulli measure.
On the other hand,
in the conventional quantum mechanics this measurement result is described as a mixed state.
Suggested by
these facts,
we propose
a \emph{mathematical definition of the notion of a mixed state}
in terms of the notion of Martin-L\"of randomness with respect to Bernoulli measure.
Then, using this rigorous definition of mixed state,
we introduce
an operational refinement of the Born rule for \emph{mixed states},
Postulate~\ref{Tadaki-rule2} in Section~\ref{Mixed_States},
as an alternative rule to it
in terms of
the notion of Martin-L\"of randomness with respect to Bernoulli measure.

\subsection{Contribution of the paper: The principle of typicality as a unifying principle}

In this paper,
we then consider the validity of our new rules,
the refined rule
of the Born rule for \emph{pure states} (i.e., Postulate~\ref{Tadaki-rule}) and
the refined rule
of the Born rule for \emph{mixed states} (i.e., Postulate~\ref{Tadaki-rule2}),
from the point of view of
the \emph{many-worlds interpretation of quantum mechanics} (\emph{MWI}, for short)
introduced by Everett~\cite{E57} in 1957.
More specifically, we
\emph{refine}  the argument of MWI
by adding to it a postulate, called the \emph{principle of typicality},
and then we derive
our refined rules of the Born rule for both pure states and mixed states
in the framework of
the \emph{refinement} of MWI based on the principle of typicality.

To begin with, we review the
original
framework of MWI,
introduced by Everett~\cite{E57}.
Actually, in this paper \emph{we reformulate
the original framework of MWI
in a form of mathematical rigor from a modern point of view}.
The point of our rigorous treatment of the original framework of MWI is
the use of the notion of
\emph{probability measure representation} and its \emph{induction of probability measure},
which are
presented in Subsection~\ref{MR}.

We stress that MWI is more than just an interpretation of quantum mechanics.
It aims to \emph{derive} the Born rule from the remaining postulates of quantum mechanics, i.e.,
Postulates~\ref{state_space}, \ref{composition}, and \ref{evolution} presented in Section~\ref{QM}.
In this sense, Everett~\cite{E57} proposed MWI as a ``metatheory'' of quantum mechanics.
The point is that in MWI the measurement process is fully treated
as the interaction between a system being measured and an apparatus measuring it, based only on
Postulates~\ref{state_space}, \ref{composition}, and \ref{evolution}, without reference to the Born rule.
Then MWI tries to derive the Born rule in such a setting.

As we already pointed out, however,
there is no operational characterization of the notion of probability in the Born rule,
while it makes a statement about the ``probability'' of measurement outcomes.
Therefore, what MWI has to show for deriving the Born rule is \emph{unclear}
(although the argument in MWI itself is rather operational).
Thus, we have no adequate criterion to confirm
that we have certainly accomplished the derivation of the Born rule based on MWI,
since the Born rule is \emph{vague} from an operational point of view.
By contrast,
the \emph{replacement} of the Born rule by
our refined rule,
Postulate~\ref{Tadaki-rule}, makes this \emph{clear}
since there is no ambiguity in
Postulate~\ref{Tadaki-rule}
from an operational point of view.

In this paper we clear up \emph{other} questionable points of
the original MWI
of
the form just proposed by Everett~\cite{E57}.
All of them come from a \emph{mathematical} deficiency of the arguments in it.
The arguments and results of the original MWI~\cite{E57} are
\emph{insufficient} from a mathematical point of view.
In particular, for deriving the Born rule,
the original MWI would seem to have wanted
to assume that our world is ``typical'' or ``random''  among many coexisting worlds.
However, the proposal of the MWI by Everett was nearly a decade earlier than
the advent of algorithmic randomness.
Actually, Everett \cite{E57} proposed MWI in 1957
while the notion of Martin-L\"of randomness was introduced by Martin-L\"of \cite{M66} in 1966.
Thus, the assumption of ``typicality'' by Everett in the original MWI is not
rigorous from a mathematical point of view.

The notion of ``typicality'' or ``randomness'' is just the
\emph{research object} of algorithmic randomness.
In Section~\ref{SecPOT},
based on the notion of Martin-L\"of randomness with respect to a probability measure,
we introduce a postulate, called the \emph{principle of typicality},
in order to overcome the deficiency of the original MWI.
The principle of typicality, Postulate~\ref{POT} in Section~\ref{SecPOT}, is
\emph{naturally formed} by applying
the \emph{basic idea} of Martin-L\"of randomness
into
the framework of MWI, and is
thought
to be a \emph{refinement} and therefore a \emph{clarification}
of the obscure assumption of
``typicality''
by Everett~\cite{E57}.
In this paper we make the \emph{whole} arguments by Everett~\cite{E57} \emph{clear}
and \emph{feasible},
based on the principle of typicality.
Actually, we can derive
our refined rule of the Born rule for pure states
(i.e., Postulate~\ref{Tadaki-rule}),
our refined rule of the Born rule for mixed states
(i.e., Postulate~\ref{Tadaki-rule2}), and
other postulates of the conventional quantum mechanics
\emph{regarding mixed states and density matrices}
(i.e., Postulate~\ref{COSMS} in Section~\ref{CPFCOMS} and
Postulate~\ref{density-matrices-probability} in Section~\ref{CPFPM})
from the principle of typicality in the framework of MWI \emph{in a unified manner}.
In the derivation, we make a comprehensive use of the results of the work \cite{T14,T15,T16arXiv}.

To begin with,
in Section~\ref{POT-pure-states}
we derive Postulate~\ref{Tadaki-rule},
our refined rule of the Born rule for pure states,
from the principle of typicality.
The setting
considered
in Postulate~\ref{Tadaki-rule} is just
the setting of the original framework of MWI, which is based on
the infinite repetition of the measurements of a single observable.
Therefore,
we see that the principle of typicality, Postulate~\ref{POT}, is precisely Postulate~\ref{Tadaki-rule}.
In this way,
we make clear the argument by Everett~\cite{E57},
based on
the principle of typicality, Postulate~\ref{POT}.

Next,
we derive Postulate~\ref{Tadaki-rule2},
our refined rule of the Born rule for mixed states,
from
the principle of typicality
in several scenarios of the setting of measurements.
We can do this by considering more complicated interactions between systems and
apparatuses as measurement processes than ones used for deriving Postulate~\ref{Tadaki-rule}.
In our framework
a mixed state, on which measurements are performed in the setting of Postulate~\ref{Tadaki-rule2},
is an infinite sequence of pure states.
This
implies
that
we have to perform measurements on a mixed state while generating it.
We
have investigated several scenarios which implement this setting.
In all the scenarios which we have considered so far,
Postulate~\ref{Tadaki-rule2} can be derived from
the principle of typicality.
In Section~\ref{POT-mixed-states}
we describe the detail of the derivation of Postulate~\ref{Tadaki-rule2} from
the principle of typicality
in the simplest scenario,
where a mixed state being measured is
an
infinite sequence over \emph{mutually orthogonal} pure states.
In Section~\ref{POT-mixed-states2} we describe the detail of the same derivation
in a \emph{more general} scenario, where a mixed state being measured is
an
infinite sequence over general \emph{mutually non-orthogonal} pure states.

Furthermore,
we investigate the validity of \emph{other} postulates
of the conventional quantum mechanics
\emph{regarding mixed states and density matrices},
in terms of our framework based on the principle of typicality.
We consider two of such postulates.
One of them is Postulate~\ref{COSMS} in Section~\ref{CPFCOMS},
which is ``Postulate 4'' described in Nielsen and Chuang \cite[Section 2.4.2]{NC00}.
It
states
how the density matrix of a mixed state of a composite system is
calculated from the density matrices of the mixed states of the component systems.
The other is Postulate~\ref{density-matrices-probability} in Section~\ref{CPFPM},
which is the last part of ``Postulate 1'' described in Nielsen and Chuang \cite[Section 2.4.2]{NC00}.
It treats the ``probabilistic mixture'' of mixed states.

In Section~\ref{CPFCOMS} we
point out the failure of Postulate~\ref{COSMS} first.
We then present a \emph{necessary and sufficient condition}
for the
statement of Postulate~\ref{COSMS} to hold under a certain natural restriction on
the forms
of the mixed states of the component systems,
in terms of the framework of the works~\cite{T14,T15,T16arXiv}.
After that, we
describe
a natural and simple scenario regarding the setting of measurements
in which the
statement of Postulate~\ref{COSMS} holds,
based on
the principle of typicality.

In Section~\ref{CPFPM} we
investigate
the validity of
Postulate~\ref{density-matrices-probability}
in terms of our framework based on
the principle of typicality.
First of all,
Postulate~\ref{density-matrices-probability}
seems
\emph{very vague} in its original form.
What does the ``probabilistic mixture'' of mixed states mean?
What does the word ``probability'' mean here?
Based on the principle of typicality,
we give a certain \emph{precise meaning} to Postulate~\ref{density-matrices-probability}
by
means of
giving an appropriate scenario
in which Postulate~\ref{density-matrices-probability} clearly holds.
In other words,
we \emph{derive} Postulate~\ref{density-matrices-probability}
from the principle of typicality
in a certain natural and simple scenario regarding the setting of measurements.

In this way, we see in the framework of MWI that
the refined rule of the Born rule for pure states (i.e., Postulate~\ref{Tadaki-rule}),
the refined rule of the Born rule for mixed states (i.e., Postulate~\ref{Tadaki-rule2}), and
other postulates of the conventional quantum mechanics regarding mixed states and density matrices
(i.e., Postulates~\ref{COSMS} and \ref{density-matrices-probability})
can \emph{all} be derived
from a \emph{single} postulate, the principle of typicality, \emph{in a unified manner}.

Finally, in Section~\ref{BB84QKD}
we make an application of our framework based on the principle of typicality,
to the BB84 quantum key distribution protocol \cite{BB84}.
Thereby we demonstrate how \emph{properly} our framework
works in \emph{practical problems in quantum mechanics}, based on the principle of typicality.

\subsection{Organization of the paper}

The paper is organized as follows.
We begin in Section~\ref{preliminaries} with
some mathematical preliminaries, in particular, about measure theory and Martin-L\"of randomness.
In Section~\ref{QM} we review the central postulates of the conventional quantum mechanics
according to Nielsen and Chuang~\cite{NC00}.
In Section~\ref{SecTadaki-rule}
the notion of Martin-L\"of randomness with respect to Bernoulli measure is introduced,
and then, based on this notion,
we present
an operational refinement of the Born rule for \emph{pure} states, Postulate~\ref{Tadaki-rule}.
We then check the validity of
Postulate~\ref{Tadaki-rule}
in Section~\ref{Validity},
based on the results of Tadaki \cite{T14,T15,T16arXiv}.
In Section~\ref{Mixed_States}
we present
an operational refinement of the Born rule for \emph{mixed} states, Postulate~\ref{Tadaki-rule2},
after introducing a \emph{mathematical definition} of the notion of a mixed state.

In Section~\ref{MWI} we reformulate the original framework of MWI~\cite{E57}
in a form of mathematical rigor from a modern point of view.
We then point out the deficiency of the original MWI~\cite{E57}.
In Section~\ref{SecPOT},
based on the notion of Martin-L\"of randomness with respect to a probability measure,
we introduce
the \emph{principle of typicality}, Postulate~\ref{POT},
in order to overcome the deficiency of the original MWI.
In order to show the results in the rest of the paper,
we need
several theorems
on Martin-L\"of randomness with respect to Bernoulli measure
from
Tadaki \cite{T14,T15,T16arXiv}.
We enumerate them and their corollaries in Section~\ref{FMP}.
In Section~\ref{POT-pure-states}
we derive Postulate~\ref{Tadaki-rule},
the refined rule of the Born rule for \emph{pure states},
from the principle of typicality.
On the one hand, in Section~\ref{POT-mixed-states}
we derive Postulate~\ref{Tadaki-rule2},
the refined rule of the Born rule for \emph{mixed states},
from the principle of typicality in the simplest scenario,
where a mixed state being measured is an infinite sequence over mutually orthogonal pure states.
On the other hand, in Section~\ref{POT-mixed-states2}
we derive Postulate~\ref{Tadaki-rule2} from the principle of typicality
in a more general scenario, where a mixed state being measured is
an infinite sequence over general mutually non-orthogonal pure states.

Furthermore,
we investigate the validity of \emph{other} postulates
of the conventional quantum mechanics
\emph{regarding mixed states and density matrices},
in terms of our framework based on the principle of typicality.
We consider two of such postulates.
One of them is Postulate~\ref{COSMS} presented in Section~\ref{CPFCOMS},
which describes how the density matrix of a mixed state of a composite system is
calculated from the density matrices of the mixed states of the component systems.
In Section~\ref{CPFCOMS}
we investigate the validity of Postulate~\ref{COSMS}.
In particular, we give a natural and simple scenario regarding the setting of measurements
in which the statement of Postulate~\ref{COSMS} holds, based on the principle of typicality.
The other is Postulate~\ref{density-matrices-probability} presented in Section~\ref{CPFPM},
which treats the ``probabilistic mixture'' of mixed states.
In Section~\ref{CPFPM} we investigate the validity of Postulate~\ref{density-matrices-probability}.
Postulate~\ref{density-matrices-probability} seems very vague in its original form.
We give a certain precise meaning to Postulate~\ref{density-matrices-probability}
by means of giving an appropriate scenario
in which Postulate~\ref{density-matrices-probability} clearly holds,
based on the principle of typicality.

In Section~\ref{BB84QKD}, we make an application of our framework based on the principle of typicality,
to the BB84 quantum key distribution protocol~\cite{BB84}
in order to demonstrate how properly our framework works in practical problems
in quantum mechanics.
We conclude this paper with remarks and
a mention of the future direction of this work in Section~\ref{Concluding}.

\section{Mathematical preliminaries}
\label{preliminaries}

\subsection{Basic notation and definitions}
\label{basic notation}

We start with some notation about numbers and strings which will be used in this paper.
$\#S$ is the cardinality of $S$ for any set $S$.
$\N=\left\{0,1,2,3,\dotsc\right\}$ is the set of \emph{natural numbers},
and $\N^+$ is the set of \emph{positive integers}.
$\R$ is the set of \emph{reals}, and $\C$ is the set of \emph{complex numbers}.

An \emph{alphabet} is a non-empty finite set.
Let $\Omega$ be an arbitrary alphabet throughout the rest of this subsection.
A \emph{finite string over $\Omega$} is a finite sequence of elements from the alphabet $\Omega$.
We use $\Omega^*$ to denote the set of all finite strings over $\Omega$,
which contains the \emph{empty string} denoted by $\lambda$.
For any $\sigma\in\Omega^*$, $\abs{\sigma}$ is the \emph{length} of $\sigma$.
Therefore $\abs{\lambda}=0$.
A subset $S$ of $\Omega^*$ is called
\emph{prefix-free}
if no string in $S$ is a prefix of another string in $S$.

An \emph{infinite sequence over $\Omega$} is an infinite sequence of elements from the alphabet $\Omega$,
where the sequence is infinite to the right but finite to the left.
We use $\Omega^\infty$ to denote the set of all infinite sequences over $\Omega$.
Let $\alpha\in\Omega^\infty$.
For any $n\in\N$, we denote by $\rest{\alpha}{n}\in\Omega^*$ the first $n$ elements
in the infinite sequence $\alpha$ and by $\alpha(n)$ the $n$th element in $\alpha$.
Thus, for example, $\rest{\alpha}{4}=\alpha(1)\alpha(2)\alpha(3)\alpha(4)$, and $\rest{\alpha}{0}=\lambda$.
For any $S\subset\Omega^*$, the set
$\{\alpha\in\Omega^\infty\mid\exists\,n\in\N\;\rest{\alpha}{n}\in S\}$
is denoted by $\osg{S}$.
Note that (i) $\osg{S}\subset\osg{T}$ for every $S\subset T\subset\Omega^*$, and
(ii) for every set $S\subset\Omega^*$ there exists a prefix-free set $P\subset\Omega^*$ such that
$\osg{S}=\osg{P}$.
For any $\sigma\in\Omega^*$, we denote by $\osg{\sigma}$ the set $\osg{\{\sigma\}}$, i.e.,
the set of all infinite sequences over $\Omega$ extending $\sigma$.
Therefore $\osg{\lambda}=\Omega^\infty$.

\subsection{Measure theory}
\label{MR}

We briefly review measure theory according to Nies~\cite[Section 1.9]{N09}.
See also Billingsley~\cite{B95} for measure theory in general.

Let $\Omega$ be an arbitrary alphabet.
A real-valued function $\mu$ defined on the class of all subsets of $\Omega^\infty$ is called
an \emph{outer measure on $\Omega^\infty$} if the following conditions hold:
\begin{enumerate}
  \item $\mu\left(\emptyset\right)=0$;
  \item $\mu\left(\mathcal{C}\right)\le\mu\left(\mathcal{D}\right)$
    for every subsets $\mathcal{C}$ and $\mathcal{D}$ of $\Omega^\infty$
    with $\mathcal{C}\subset\mathcal{D}$;
  \item $\mu\left(\bigcup_{i}\mathcal{C}_i\right)\le\sum_{i}\mu\left(\mathcal{C}_i\right)$
    for every sequence $\{\mathcal{C}_i\}_{i\in\N}$ of subsets of $\Omega^\infty$.
\end{enumerate}
A \emph{probability measure representation over $\Omega$} is
a function $r\colon\Omega^*\to[0,1]$ such that
\begin{enumerate}
  \item $r(\lambda)=1$ and
  \item for every $\sigma\in\Omega^*$ it holds that
    \begin{equation}\label{pmr}
       r(\sigma)=\sum_{a\in\Omega}r(\sigma a).
    \end{equation}
\end{enumerate}
A probability measure representation $r$ over $\Omega$ \emph{induces}
an outer measure $\mu_r$ on $\Omega^\infty$ in the following manner:
A subset $\mathcal{R}$ of $\Omega^\infty$ is called \emph{open} if
$\mathcal{R}=\osg{S}$ for some $S\subset\Omega^*$.
Let $r$ be an arbitrary probability measure representation over $\Omega$. 
For each open subset $\mathcal{A}$ of $\Omega^\infty$, we define $\mu_r(\mathcal{A})$ by
$$\mu_r(\mathcal{A}):=\sum_{\sigma\in E}r(\sigma),$$
where $E$ is a prefix-free subset of $\Omega^*$ with $\osg{E}=\mathcal{A}$.
Due to the equality \eqref{pmr} the sum is independent of the choice of  the prefix-free set $E$,
and therefore the value $\mu_r(\mathcal{A})$ is well-defined.
Then, for any subset $\mathcal{C}$ of $\Omega^\infty$, we define $\mu_r(\mathcal{C})$ by
$$\mu_r(\mathcal{C}):=
\inf\{\mu_r(\mathcal{A})\mid
\mathcal{C}\subset\mathcal{A}\text{ \& $\mathcal{A}$ is an open subset of $\Omega^\infty$}\}.$$
We can then show that $\mu_r$ is an \emph{outer measure} on $\Omega^\infty$  such that
$\mu_r(\Omega^\infty)=1$.

A class $\mathcal{F}$ of subsets of $\Omega^\infty$ is called
a \emph{$\sigma$-field on $\Omega^\infty$}
if  $\mathcal{F}$ includes $\Omega^\infty$, is closed under complements,
and is closed under the formation of countable unions.
The \emph{Borel class} $\mathcal{B}_{\Omega}$ is the $\sigma$-field \emph{generated by}
all open sets on $\Omega^\infty$.
Namely, the Borel class $\mathcal{B}_{\Omega}$ is defined
as the intersection of all the $\sigma$-fields on $\Omega^\infty$ containing
all open sets on $\Omega^\infty$.
A real-valued function $\mu$ defined on the Borel class $\mathcal{B}_{\Omega}$ is called
a \emph{probability measure on $\Omega^\infty$} if the following conditions hold:
\begin{enumerate}
  \item $\mu\left(\emptyset\right)=0$ and $\mu\left(\Omega^\infty\right)=1$;
  \item $\mu\left(\bigcup_{i}\mathcal{D}_i\right)=\sum_{i}\mu\left(\mathcal{D}_i\right)$
    for every sequence $\{\mathcal{D}_i\}_{i\in\N}$ of sets in $\mathcal{B}_{\Omega}$ such that
    $\mathcal{D}_i\cap\mathcal{D}_i=\emptyset$ for all $i\neq j$.
\end{enumerate}
Then, for every probability measure representation $r$ over $\Omega$,
we can show that the restriction of the outer measure $\mu_r$ on $\Omega^\infty$
to the Borel class $\mathcal{B}_{\Omega}$ is
a probability measure on $\Omega^\infty$.
We denote the restriction of $\mu_r$ to $\mathcal{B}_{\Omega}$ by
$\mu_r$
just the same.

Then it is easy to see that
\begin{equation}\label{mr}
  \mu_r\left(\osg{\sigma}\right)=r(\sigma)
\end{equation}
for every probability measure representation $r$ over $\Omega$ and every $\sigma\in \Omega^*$.
The probability measure $\mu_r$ is called a
\emph{probability measure induced by the probability measure representation $r$}.

\subsection{Computability}
\label{Computability}

A function $f\colon\N^+\to\N^+$ is called \emph{computable}
if there exists a deterministic Turing machine $\mathcal{M}$ such that, for each $n\in\N^+$,
when executing $\mathcal{M}$ with the input $n$,
the computation of $\mathcal{M}$ eventually terminates and then $\mathcal{M}$ outputs $f(n)$.
A computable function is also called a \emph{total recursive function}.

A subset $\mathcal{C}$ of $\N^+\times\Omega^*$ is called \emph{recursively enumerable}
if there exists a deterministic Turing machine $\mathcal{M}$ such that,
for each $x\in\N^+\times\Omega^*$,
when executing $\mathcal{M}$ with the input $x$,
\begin{enumerate}
\item if $x\in\mathcal{C}$ then the computation of $\mathcal{M}$ eventually terminates;
\item if $x\notin\mathcal{C}$ then the computation of $\mathcal{M}$ does not terminate.
\end{enumerate}

See Sipser~\cite{Sip13} for the
basic
definitions and results of the theory of computation.

\subsection{Martin-L\"of randomness with respect to an arbitrary probability measure}
\label{MLRam}

In this subsection,
we introduce the notion of \emph{Martin-L\"of randomness} \cite{M66}
in a general setting.

Let $\Omega$ be an arbitrary alphabet,
and $\mu$ be an arbitrary probability measure on $\Omega^\infty$.
The basic idea of Martin-L\"of randomness
(with respect to the probability measure $\mu$)
is as follows.
\begin{quote}
\textbf{Basic idea of Martin-L\"of randomness}:
The \emph{random} infinite sequences over $\Omega$ are precisely sequences
which are not contained in any \emph{effective null set} on $\Omega^\infty$.
\end{quote}
Here, an \emph{effective null set} on $\Omega^\infty$ is
a set $\mathcal{S}\in\mathcal{B}_\Omega$ such that $\mu(\mathcal{S})=0$
and moreover $\mathcal{S}$ has some type of \emph{effective} property.
As a specific implementation of the idea of effective null set, we introduce the following notion.

\begin{definition}[Martin-L\"{o}f test with respect to a probability measure]
\label{ML-testM}
Let $\Omega$ be an alphabet, and let $\mu$ be a probability measure on $\Omega^\infty$.
A subset $\mathcal{C}$ of $\N^+\times\Omega^*$ is called a
\emph{Martin-L\"{o}f test with respect to $\mu$} if
$\mathcal{C}$ is
a recursively enumerable set,
and
\begin{equation}\label{muocn<2n}
  \mu\left(\osg{\mathcal{C}_n}\right)<2^{-n}
\end{equation}
for every $n\in\N^+$,
where $\mathcal{C}_n$ denotes the set
$\left\{\,
    \sigma\mid (n,\sigma)\in\mathcal{C}
\,\right\}$.
\qed
\end{definition}

Let $\mathcal{C}$ be a Martin-L\"{o}f test with respect to $\mu$.
Then,
it follows from \eqref{muocn<2n}
that $\mu\left(\bigcap_{n=1}^{\infty}\osg{\mathcal{C}_n}\right)=0$.
Therefore,
the set $\bigcap_{n=1}^{\infty}\osg{\mathcal{C}_n}$ serves as an effective null set.
In this manner, the notion of an effective null set is implemented
as
a Martin-L\"{o}f test with respect a probability measure
in Definition~\ref{ML-testM}.

Then,
the notion of \emph{Martin-L\"of randomness with respect to a probability measure} is
defined as follows, according to the basic idea of Martin-L\"{o}f randomness stated above.

\begin{definition}[Martin-L\"{o}f randomness with respect to a probability measure]
\label{ML-randomness-wrtm}
Let $\Omega$ be an alphabet, and let $\mu$ be a probability measure on $\Omega^\infty$.
For any $\alpha\in\Omega^\infty$, we say that $\alpha$ is
\emph{Martin-L\"{o}f random with respect to $\mu$} if
$$\alpha\notin\bigcap_{n=1}^{\infty}\osg{\mathcal{C}_n}$$
for every Martin-L\"{o}f test $\mathcal{C}$ with respect to $\mu$.\qed
\end{definition}

Since there are only countably infinitely many algorithms and
every Martin-L\"of test with respect to $\mu$ induces an effective null set,
it is easy to show the following theorem.

\begin{theorem}\label{MLmae}
Let $\Omega$ be an alphabet, and let $\mu$ be a probability measure on $\Omega^\infty$.
Let $\mathrm{ML}_\mu$ be the set of all $\alpha\in\Omega^\infty$ such that
$\alpha$ is Martin-L\"of random with respect to $\mu$.
Then $\mathrm{ML}_\mu\in\mathcal{B}_{\Omega}$ and
$\mu\left(\mathrm{ML}_\mu\right)=1$.
\qed
\end{theorem}

\section{Postulates of quantum mechanics}
\label{QM}

In this section, we
review
the central postulates of
(the conventional)
quantum mechanics.
For simplicity, in this paper we consider the postulates of quantum mechanics
for a \emph{finite-dimensional} quantum system, i.e.,
a quantum system whose state space is a \emph{finite-dimensional} Hilbert space.
Nielsen and Chuang
\cite{NC00}
treat thoroughly the postulates of
(the conventional)
quantum mechanics in the finite-dimensional case,
as a
textbook of
the field of
quantum computation and quantum information
in which
such a case
is typical.
Throughout this paper
we refer to
the postulates of
the \emph{conventional}
quantum mechanics in the form presented in
Nielsen and Chuang
\cite[Chapter 2]{NC00}.
Note that
the postulates reviewed in this section
are about \emph{pure states}.
We will consider the postulates of quantum mechanics about \emph{mixed states} and
their refinements later.

The first postulate of quantum mechanics is about \emph{state space} and \emph{state vector}.

\begin{postulate}[State space and state vector]\label{state_space}
Associated to any isolated physical system is a complex vector space with inner product
(that is, a Hilbert space)
known as the \emph{state space} of the system.
The system is completely described by its \emph{state vector},
which is a unit vector in the system's state space.
\qed
\end{postulate}

The second postulate of quantum mechanics is about the \emph{composition} of systems.

\begin{postulate}[Composition of systems]\label{composition}
The state space of a composite physical system is the tensor product of the state spaces of the component physical systems.
Moreover, if we have systems numbered $1$ through $n$, and system number $i$ is
prepared
in the state $\ket{\Psi_i}$,
then the joint state of the total system is
$$\ket{\Psi_1}\otimes\ket{\Psi_2}\otimes\dots\otimes\ket{\Psi_n}.$$
\qed
\end{postulate}

The third postulate of quantum mechanics is about the \emph{time-evolution} of
\emph{closed} quantum systems.

\begin{postulate}[Unitary time-evolution]\label{evolution}
The evolution of a \emph{closed} quantum system is described by a \emph{unitary transformation}.
That is,
the state $\ket{\Psi_1}$ of the system at time $t_1$ is related to the state $\ket{\Psi_2}$ of the system at time $t_2$
by a unitary operator $U$, which depends only on the times $t_1$ and $t_2$,
in such a way that
$$\ket{\Psi_2}=U\ket{\Psi_1}.$$
\qed
\end{postulate}

The forth postulate of quantum mechanics is about \emph{measurements} on quantum systems.
This is the so-called \emph{Born rule}, i.e, \emph{the probability interpretation of the wave function}.

\begin{postulate}[The Born rule]\label{Born-rule}
A
quantum measurement is described by an \emph{observable}, $M$, a Hermitian operator
on the state space of the system being measured.
The observable has a spectral decomposition,
$$M=\sum_{m} m E_m,$$
where $E_m$ is the projector onto the eigenspace of $M$ with eigenvalue $m$.
The possible outcomes of the measurement correspond to the eigenvalues, $m$, of the observable.
If the state of the quantum system is $\ket{\Psi}$ immediately before the measurement
then the \emph{probability} that result $m$ occurs is given by
$$\bra{\Psi}E_m\ket{\Psi},$$
and the state of the system after the measurement is
\begin{equation*}
  \frac{E_m\ket{\Psi}}{\sqrt{\bra{\Psi}E_m\ket{\Psi}}}.
\end{equation*}
\qed
\end{postulate}

Thus, the Born rule, Postulate~\ref{Born-rule}, uses the \emph{notion of probability}.
However, the operational characterization of the notion of probability is not given in the Born rule,
and therefore the relation of its statement to a specific infinite sequence of outcomes of quantum measurements which are being generated by an infinitely repeated measurements is \emph{unclear}.
In this paper we
will
fix this point.

Throughout this paper
\emph{we keep Postulates~\ref{state_space}, \ref{composition}, and \ref{evolution}
in their original forms without any modifications}.
We then propose Postulate~\ref{Tadaki-rule} below
as a \emph{refinement} of Postulate~\ref{Born-rule},
based on the notion of
\emph{Martin-L\"of randomness with respect to Bernoulli measure}.


\section{A refinement of the Born rule}
\label{SecTadaki-rule}

In this section
we introduce a refinement of the Born rule, Postulate~\ref{Tadaki-rule}.
We propose to replace the Born rule by it.

In Section~\ref{MLRam} we have introduced the notion of 
Martin-L\"of randomness with respect to an arbitrary probability measure.
To state Postulate~\ref{Tadaki-rule} we use this notion where
the probability measure is chosen to be Bernoulli measure.
Namely, we use the notion of \emph{Martin-L\"of randomness with respect to Bernoulli measure}.
In order to introduce this notion, we first
review the notions of \emph{finite probability space} and \emph{Bernoulli measure}.
Both of them are
notions
from \emph{measure theory}.

\begin{definition}[Finite probability space]\label{def-FPS}
Let $\Omega$ be an alphabet. A \emph{finite probability space on $\Omega$} is a function $P\colon\Omega\to [0,1]$
such that
\begin{enumerate}
  \item $P(a)\ge 0$ for every $a\in \Omega$, and
  \item $\sum_{a\in \Omega}P(a)=1$.
\end{enumerate}
The set of all finite probability spaces on $\Omega$ is denoted by $\PS(\Omega)$.

Let $P\in\PS(\Omega)$.
The set $\Omega$ is called the \emph{sample space} of $P$,
and elements
of
$\Omega$ are called \emph{sample points} or \emph{elementary events}
of $P$.
For each $A\subset\Omega$, we define $P(A)$ by
$$P(A):=\sum_{a\in A}P(a).$$
A subset of $\Omega$ is called an \emph{event} on $P$, and
$P(A)$ is called the \emph{probability} of $A$
for every event $A$
on $P$.
\qed
\end{definition}

Let $P\in\PS(\Omega)$.
For each $\sigma\in\Omega^*$, we use $P(\sigma)$ to denote
$P(\sigma_1)P(\sigma_2)\dots P(\sigma_n)$
where $\sigma=\sigma_1\sigma_2\dots\sigma_n$ with $\sigma_i\in\Omega$.
For each subset $S$ of $\Omega^*$, we use $P(S)$ to denote
$$\sum_{\sigma\in S}P(\sigma).$$

Consider a function $r\colon\Omega^*\to[0,1]$ such that $r(\sigma)=P(\sigma)$ for every $\sigma\in\Omega^*$.
It is then easy to see that the function $r$ is a probability measure representation over $\Omega$.
The probability measure $\mu_r$ induced by $r$ is
called
a \emph{Bernoulli measure on $\Omega^\infty$}, denoted
$\lambda_{P}$.
The Bernoulli measure $\lambda_{P}$ on $\Omega^\infty$
satisfies that
\begin{equation}\label{BmPosgsigma=Psigma}
  \Bm{P}{\osg{\sigma}}=P(\sigma)
\end{equation}
for every $\sigma\in \Omega^*$,
which follows from \eqref{mr}.

The notion of \emph{Martin-L\"of randomness with respect to Bernoulli measure} is defined as follows.
We call it the \emph{Martin-L\"of $P$-randomness} in this paper,
since it depends on a finite probability space $P$.
This notion was, in essence, introduced by Martin-L\"{o}f~\cite{M66},
as well as the notion of Martin-L\"of randomness
with respect to Lebesgue measure.

\begin{definition}[%
Martin-L\"of $P$-randomness,
Martin-L\"{o}f \cite{M66}]\label{ML_P-randomness}
Let $P\in\PS(\Omega)$.
    For any $\alpha\in\Omega^\infty$, we say that $\alpha$ is \emph{Martin-L\"{o}f $P$-random} if
    $\alpha$ is Martin-L\"{o}f random with respect to $\lambda_{P}$.\qed
\end{definition}

Now, let us introduce a refinement of Postulate~\ref{Born-rule}.
Let $\Omega$ be an alphabet consisting of reals.
Suppose that $\Omega$ is the set of all possible measurement outcomes in a quantum measurement.
Let us identify the form of the postulate of quantum measurements as it ought to be,
from a general point of view.
Consider an infinite sequence $\alpha$ of the outcomes of quantum measurements such as
$$\alpha=a_1 a_2 a_3 a_4 a_5 a_6 a_7 a_8 \dotsc\dotsc$$
with $a_i\in\Omega$,
which is being generated as measurements progressed.
All that the experimenter of quantum measurements can obtain through the measurements about quantum system is
such
a \emph{specific} infinite sequence of outcomes in $\Omega$ of the measurements
which are being generated by infinitely repeated measurements.
Thus, \emph{the object about which the postulate of quantum measurements makes a statement should be
the properties of
a
specific infinite sequence
$\alpha\in\Omega^\infty$
of outcomes of the measurements.}

Suggested by this consideration, we propose to replace the Born rule, Postulate~\ref{Born-rule},
by the following postulate:

\begin{postulate}[Refinement of the Born rule for pure states]\label{Tadaki-rule}
A
quantum measurement is described by an \emph{observable}, $M$,
a Hermitian operator on the state space of the system being measured.
The observable has a spectral decomposition,
$$M=\sum_{m\in\Omega} m E_m,$$
where $E_m$ is the projector onto the eigenspace of $M$ with eigenvalue $m$.
The set of possible outcomes of the measurement is the spectrum $\Omega$ of $M$.
Suppose that the measurements are repeatedly performed over identical quantum systems
whose states are all $\ket{\Psi}$, and
the infinite sequence $\alpha\in \Omega^\infty$ of measurement outcomes is being generated.
Then $\alpha$ is Martin-L\"of $P$-random, where $P$ is a finite probability space on $\Omega$
such that
$$P(m)=\bra{\Psi}E_m\ket{\Psi}$$
for every $m\in\Omega$.
For each of the measurements, the state of the system immediately after the measurement is
\begin{equation}\label{pms}
  \frac{E_m\ket{\Psi}}{\sqrt{\bra{\Psi}E_m\ket{\Psi}}},
\end{equation}
where $m$ is the corresponding measurement outcome.
\qed
\end{postulate}

Note that the function $P$ appearing in Postulate~\ref{Tadaki-rule} is certainly
a finite probability space on $\Omega$,
since $\sum_{m\in\Omega}E_m=I$ holds for
the projectors $E_m$.

\section{Verification of the validity of Postulate~\ref{Tadaki-rule}}
\label{Validity}

Let us begin to check the validity of Postulate~\ref{Tadaki-rule}.
Based on
the results of the work \cite{T16arXiv},
we can see that Postulate~\ref{Tadaki-rule} is certainly a refinement of
Postulate~\ref{Born-rule}, the Born rule,
from the point of view of our intuitive understanding of the notion of probability.

First of all, what is ``probability''?
In particular, what is ``probability'' in quantum mechanics?
It would seem very difficult to answer this question
\emph{completely} and \emph{sufficiently}.
However, we may enumerate the \emph{necessary} conditions which the notion of probability is
considered to have to satisfy
\emph{according to
our intuitive understanding of the notion of probability}.
In the subsequent three subsections,
we check that Postulate~\ref{Tadaki-rule} satisfies these necessary conditions.

\subsection{The frequency interpretation}

The first necessary condition
which the notion of probability is considered to have to satisfy
is the \emph{law of large numbers}, i.e., the \emph{frequency interpretation}.
Actually, according to Postulate~\ref{Tadaki-rule}
we can show that the law of large numbers holds for the infinite sequence $\alpha\in \Omega^\infty$
of measurement outcomes appearing
in Postulate~\ref{Tadaki-rule}.
This is confirmed by the following theorem.
See Tadaki~\cite[Theorem 11]{T16arXiv} for the proof.

\begin{theorem}[The law of large numbers]\label{FI}
Let $\Omega$ be an alphabet, and let $P\in\PS(\Omega)$.
For every $\alpha\in\Omega^\infty$, if $\alpha$ is Martin-L\"of $P$-random
then for every $a\in\Omega$ it holds that
$$\lim_{n\to\infty} \frac{N_a(\rest{\alpha}{n})}{n}=P(a),$$
where $N_a(\sigma)$ denotes the number of the occurrences of $a$ in $\sigma$ for every $a\in\Omega$ and $\sigma\in\Omega^*$.
\qed
\end{theorem}

From Theorem~\ref{FI} we see that
$$\forall\,m\in\Omega\;\,\lim_{n\to\infty}\frac{N_m(\rest{\alpha}{n})}{n}=\bra{\Psi}E_m\ket{\Psi}$$
holds for the infinite sequence $\alpha\in \Omega^\infty$ in Postulate~\ref{Tadaki-rule}.
Thus, the frequency interpretation, which is expected
from
Postulate~\ref{Born-rule},
holds.

\subsection{Elementary event with probability one occurs certainly}

The second necessary condition which the notion of probability is considered to have to satisfy is
the condition that \emph{an elementary event with probability one occurs certainly}.
Actually, according to Postulate~\ref{Tadaki-rule} we can show that
an elementary event with probability one always occurs
in the infinite sequence $\alpha\in \Omega^\infty$
of measurement outcomes appearing
in Postulate~\ref{Tadaki-rule}.
This fact that \emph{an elementary event with probability one occurs certainly in quantum mechanics}
is derived
in Tadaki~\cite{T16arXiv,T17SCIS}
in the context of the conventional quantum mechanics.
For completeness, we include the derivation of this fact as follows:

Recall that
there is a postulate about quantum measurements with no reference to the notion of probability.
This is given in
von Neumann \cite[Section III.3]{vN55},
and describes a spacial case of quantum measurements
where the measurement of an observable is performed upon a quantum system
in an \emph{eigenstate} of the observable being measured, i.e.,
a state represented by an eigenvector of the observable being measured.
We here refer to this postulate in the form described in Dirac \cite[Section 10]{D58}.
\begin{postulate}[Dirac \cite{D58}]\label{Dirac}
If the dynamical system is in an eigenstate of a real dynamical variable $\xi$, belonging to the eigenvalue $\xi'$,
then a measurement of $\xi$ will certainly gives as result the number $\xi'$.
\qed
\end{postulate}
Here, the ``dynamical system'' means quantum system.
Any observable is a ``real dynamical variable'' mentioned in Postulate~\ref{Dirac}.

Based on Postulates~\ref{state_space}, \ref{Born-rule}, and \ref{Dirac} above,
we can show the fact that an elementary event ``with probability one'' occurs certainly
in the conventional quantum mechanics.
As we already pointed out, an operational characterization of the notion of probability is missing
in Postulate~\ref{Born-rule}.
However, Postulate~\ref{Born-rule} surely mentions the word ``probability''.
Based on
this mention,
we can \emph{literally} derive
the fact above.

Now, for deriving this fact, let us consider a quantum system with finite-dimensional state space,
and a measurement described by an observable $M$ performed upon the quantum system.
Suppose that the probability of getting result $m_0$ is one  in the measurement performed upon the system
in a state represented by a state vector $\ket{\Psi}$.
Let
$M=\sum_{m}m E_m$
be a spectral decomposition of the observable $M$,
where $E_m$ is the projector onto the eigenspace of $M$ with eigenvalue $m$.
Then, it follows from Postulate~\ref{Born-rule} that
$$\bra{\Psi}E_{m_0}\ket{\Psi}=1.$$
This implies that $\ket{\Psi}$ is an eigenvector of $M$ belonging to the eigenvalue $m_0$,
since $\ket{\Psi}$ is a unit vector according to
the
convention for state vectors adopted in Postulate~\ref{state_space}.
Thus,
we have that immediately before the measurement, the quantum system is
in an eigenstate of the observable $M$, belonging to the eigenvalue $m_0$.
It follows from Postulate~\ref{Dirac} that
the measurement of $M$ will \emph{certainly} gives as result the number $m_0$.
Hence, it turns out that
\emph{an elementary event with probability one occurs certainly in the conventional quantum mechanics}.

Now, let us turn to
see
the validity of Postulate~\ref{Tadaki-rule}.
Theorem~\ref{one_probability} below confirms that an event with probability one always occurs
in the infinite sequence $\alpha\in \Omega^\infty$
of measurement outcomes appearing
in Postulate~\ref{Tadaki-rule}.
To see this, suppose that $\bra{\Psi}E_{m_0}\ket{\Psi}=1$ holds for a certain $m_0\in\Omega$
in the setting of Postulate~\ref{Tadaki-rule}.
Then, by Postulate~\ref{Tadaki-rule},
the infinite sequence $\alpha\in \Omega^\infty$ of measurement outcomes being generated is
Martin-L\"of $P$-random, where $P$ is a finite probability space on $\Omega$
such that $P(m_0)=1$.
It follows from Theorem~\ref{one_probability} that
the infinite sequence $\alpha$ consists only of $m_0$.
Thus, the event $m_0$ with probability one always occurs
in the infinitely repeated measurements
in the setting of Postulate~\ref{Tadaki-rule},
as desired.
This result strengthens the validity of Postulate~\ref{Tadaki-rule}.

\begin{theorem}\label{one_probability}
Let $P\in\PS(\Omega)$, and let $a\in\Omega$.
Let $\alpha\in\Omega^\infty$.
Suppose that $\alpha$ is Martin-L\"of $P$-random and $P(a)=1$.
Then $\alpha$ consists only of $a$, i.e., $\alpha=aaaaaa\dotsc\dotsc$.\qed
\end{theorem}

Theorem~\ref{one_probability} was, in essence, pointed out by Martin-L\"{o}f~\cite{M66}.
See Tadaki~\cite[Theorem 9]{T16arXiv} for the proof.

\subsection{Self-consistency on some level}
\label{constcs}

This subsection considers and verifies the third necessary condition which
the notion of probability is thought to have to satisfy.
This is about the \emph{self-consistency} of the notion of probability.
Namely,
we can verify the \emph{self-consistency} of Postulate~\ref{Tadaki-rule} on some level.
We do this, based on the arguments given in Tadaki~\cite[Sections 5.3 and 5.4]{T16arXiv}.
This result suggests
that Postulate~\ref{Tadaki-rule} is not too strong.
The detail of the verification of the self-consistency is given as follows.

We assume Postulate~\ref{Tadaki-rule}.
Let $\mathcal{S}$ be an arbitrary quantum system with state space
of finite dimension.
Consider a measurement over $\mathcal{S}$ described by an arbitrary observable $M$.
Then the observable $M$ has a spectral decomposition
$$M=\sum_{m\in\Omega} m E_m,$$
where $E_m$ is the projector onto the eigenspace of $M$ with eigenvalue $m$.
The set of possible outcomes of the measurement is the spectrum $\Omega$ of $M$.

Assume that an observer $A$ performs an infinite reputation of the measurement described by
the observable $M$ over an infinite copies of $\mathcal{S}$ all prepared in an identical state $\ket{\Psi}$.
Then the infinite sequence $\alpha\in \Omega^\infty$ of measurement outcomes is being generated as
$$\alpha=a_1 a_2 a_3 a_4 a_5 a_6 a_7 a_8 \dotsc\dotsc$$
with $a_i\in\Omega$.
According to Postulate~\ref{Tadaki-rule},
$\alpha$ is Martin-L\"of $P$-random, where $P$ is a finite probability space on $\Omega$
such that $P(m)=\bra{\Psi}E_m\ket{\Psi}$ for every $m\in\Omega$.
Consider another observer $B$ who wants to adopt the following subsequence $\beta$ of $\alpha$
as the outcomes of measurements:
$$\beta=a_2 a_3 a_5 a_7 a_{11} a_{13} a_{17} \dotsc\dotsc,$$
where the observer $B$ only takes into account
the $n$th measurements in the original infinite sequence $\alpha$ of measurements
such that $n$ is a prime number.
According to Postulate~\ref{Tadaki-rule}, $\beta$ has to be Martin-L\"of $P$-random, as well.
However, is this true?

Consider this problem in a general setting.
Assume as before that an observer $A$ performs an infinite reputation of the measurement described by
the observable $M$ over an infinite copies of $\mathcal{S}$ all prepared in an identical state $\ket{\Psi}$.
Then the infinite sequence $\alpha\in \Omega^\infty$ of measurement outcomes is being generated.
According to Postulate~\ref{Tadaki-rule},
$\alpha$ is Martin-L\"of $P$-random, where $P$ is a finite probability space on $\Omega$
such that $P(m)=\bra{\Psi}E_m\ket{\Psi}$ for every $m\in\Omega$.
Now, let $f\colon\N^+\to\N^+$ be an injection.
Consider another observer $B$ who wants to adopt the following sequence $\beta$
as the outcomes of the measurements:
$$\beta=\alpha(f(1))\alpha(f(2))\alpha(f(3))\alpha(f(4))\alpha(f(5))\dotsc\dotsc,$$
instead of $\alpha$.
According to Postulate~\ref{Tadaki-rule}, $\beta$ has to be Martin-L\"of $P$-random, as well.
However, is this true?

We can confirm this \emph{by restricting the ability of $B$}, that is, by assuming that
every observer can select elements from the original
infinite
sequence $\alpha$
\emph{only in an effective manner}.
This means that the function $f\colon\N^+\to\N^+$ has to be a \emph{computable} function, i.e.,
a \emph{total recursive} function.
Theorem~\ref{cpucs} below shows this result.

\begin{theorem}[Closure property under computable shuffling]\label{cpucs}
Let $\Omega$ be an alphabet, and let $P\in\PS(\Omega)$.
Let $\alpha\in\Omega^\infty$.
Suppose that $\alpha$ is Martin-L\"of $P$-random.
Then, for every injective function $f\colon\N^+\to\N^+$, if $f$ is computable then the infinite sequence
\begin{equation*}
  \alpha_f:=\alpha(f(1)) \alpha(f(2)) \alpha(f(3)) \alpha(f(4)) \dotsc\dotsc\dotsc
\end{equation*}
is Martin-L\"of $P$-random.
\end{theorem}

\begin{proof}
See Tadaki~\cite[Theorem 15]{T16arXiv} for the proof.
\end{proof}

In other words, Theorem~\ref{cpucs} states that Martin-L\"of $P$-random sequences are
\emph{closed under computable shuffling}.

Furthermore,
we can show another type of consistency, i.e.,
the consistency based on the fact that
Martin-L\"of $P$-random sequences for an arbitrary finite probability space are
\emph{closed under the selection by a partial computable selection function}.
See Tadaki~\cite[Section 5.4]{T16arXiv} for the detail of this fact and its implication to the consistency.

\subsection{An operational characterization of the notion of probability in general}

Postulate~\ref{Tadaki-rule} is based on the notion of Martin-L\"of $P$-randomness.
In general,
we can use this notion to present \emph{an operational characterization of the notion of probability},
and we can reformulate probability theory
based on the notion of Martin-L\"of $P$-randomness \cite{T14,T15,T16arXiv}.
For example, we can represent the notion of conditional probability and the notion of
the independence of events/random variables in terms of Martin-L\"of $P$-randomness.
Thus,
\emph{Martin-L\"of $P$-randomness is thought to reflect all the properties of the notion of probability
from our intuitive understanding of the notion of probability}.
Hence, Postulate~\ref{Tadaki-rule}, which uses the notion of Martin-L\"of $P$-randomness,
is thought to be a \emph{rigorous reformulation of Postulate~\ref{Born-rule}}.
The detail of the operational characterization
of the notion of probability
by Martin-L\"of $P$-randomness
is reported in
the work \cite{T16arXiv}.

We will later show that Postulate~\ref{Tadaki-rule} can be derived from a
general
postulate,
called the \emph{principle of typicality}, together with
Postulates~\ref{state_space}, \ref{composition}, and \ref{evolution}.
This
suggests a further validity of Postulate~\ref{Tadaki-rule}.

\section{Mixed states}
\label{Mixed_States}

Postulate~\ref{Born-rule}
given in Section~\ref{QM}
is the Born rule for \emph{pure states}.
In this section we consider the Born rule for \emph{mixed states} and
its \emph{refinement} by algorithmic randomness.
We first recall that the Born rule for mixed states is given
in the following form
(see
Nielsen and Chuang
\cite[Subsection 2.4.2]{NC00}).

\begin{postulate}[The Born rule for mixed states]\label{Born-rule2}
A
quantum measurement is described by an \emph{observable}, $M$, a Hermitian operator
on the state space of the system being measured.
The observable has a spectral decomposition,
$$M=\sum_{m} m E_m,$$
where $E_m$ is the projector onto the eigenspace of $M$ with eigenvalue $m$.
The possible outcomes of the measurement correspond to the eigenvalues, $m$, of the observable.
If the state of the quantum system is represented by a density matrix $\rho$ immediately before the measurement then the \emph{probability} that result $m$ occurs is given by
$$\tr(E_m\rho),$$
and the state of the system after the measurement is
\begin{equation*}
  \frac{E_m\rho E_m}{\tr(E_m\rho)}.
\end{equation*}
\qed
\end{postulate}

We propose a refinement of Postulate~\ref{Born-rule2} by algorithmic randomness in what follows.
For that purpose, we first note that, according to Postulate~\ref{Tadaki-rule},
the result of the quantum measurements forms
a Martin-L\"of $P$-random infinite sequence of pure states,
each of which is of the form of \eqref{pms}.
On the other hand,
in the conventional quantum mechanics this measurement result is described as a mixed state.
Suggested by
these facts,
we propose
a \emph{mathematical definition of the notion of a mixed state}
in terms of Martin-L\"of $P$-randomness, as follows.

\begin{definition}[Mixed state and its density matrix]\label{def-mixed-state}
Let $\mathcal{S}$ be a quantum system with state space $\mathcal{H}$ of finite dimension,
and let $\Omega$ be a non-empty finite set of state vectors in $\mathcal{H}$.
\begin{enumerate}
\item
  An infinite sequence $\alpha$ over $\Omega$ is called a \emph{mixed state} of $\mathcal{S}$ if
  there exists a finite probability space $P$ on $\Omega$ such that
  $\alpha$ is Martin-L\"of $P$-random.
\item
  For any mixed state $\alpha$ of $\mathcal{S}$, the \emph{density matrix} $\rho$ of $\alpha$
  is defined by
  $$\rho:=\sum_{\ket{\Psi}\in\Omega} P(\ket{\Psi})\ket{\Psi}\bra{\Psi},$$
  where  $P$ is a finite probability space on $\Omega$ for which $\alpha$ is Martin-L\"of $P$-random.
  \qed
\end{enumerate}
\end{definition}

The following theorem is immediate from Theorem~\ref{FI}.
Based on this theorem,
we can see that
the notion of density matrix defined
by
(ii) of Definition~\ref{def-mixed-state} is \emph{well-defined}.

\begin{theorem}\label{uniquness}
Let $\Omega$ be an alphabet.
Let $P,Q\in\PS(\Omega)$.
If there exists $\alpha\in\Omega^\infty$ which is both
Martin-L\"of $P$-random and Martin-L\"of $Q$-random,
then $P=Q$.
\qed
\end{theorem}

Note also that
the definition of density matrix given in
Definition~\ref{def-mixed-state} is
the same form as in the conventional quantum mechanics.
Using this rigorous definition of mixed state,
we propose to replace the Born rule
for
mixed states, Postulate~\ref{Born-rule2},
by the following
postulate,
which has
an operational form based on Martin-L\"of $P$-randomness.

\begin{postulate}[Refinement of the Born rule for mixed states]\label{Tadaki-rule2}
A
quantum measurement is described by an \emph{observable}, $M$,
a Hermitian operator on the state space of the system being measured.
The observable has a spectral decomposition,
$$M=\sum_{m\in\Omega} m E_m,$$
where $E_m$ is the projector onto the eigenspace of $M$ with eigenvalue $m$.
The set of possible outcomes of the measurement is the spectrum $\Omega$ of $M$.
Suppose that the measurements are repeatedly performed over a mixed state
with
a density matrix $\rho$.
Then the infinite sequence of outcomes generated by the measurements is
a Martin-L\"of $P$-random infinite sequence over $\Omega$,
where $P$ is a finite probability space on $\Omega$ such that
$$P(m)=\tr(E_m\rho)$$
for every $m\in\Omega$.
Moreover,
the resulting sequence of pure states with outcome $m$ 
is a mixed state
with
the density matrix
\begin{equation*}
  \frac{E_m\rho E_m}{\tr(E_m\rho)}.
\end{equation*}
\qed
\end{postulate}

\section{The many-worlds interpretation of quantum mechanics}
\label{MWI}

In what follows,
we consider the validity of our new rules, Postulates~\ref{Tadaki-rule} and \ref{Tadaki-rule2},
from the point of view of
the \emph{many-worlds interpretation of quantum mechanics} (\emph{MWI}, for short)
introduced by Everett~\cite{E57} in 1957.
More specifically, we
refine  the argument of MWI
by adding to it a postulate, called the \emph{principle of typicality},
and then we derive Postulates~\ref{Tadaki-rule} and \ref{Tadaki-rule2}
in the framework of
the \emph{refinement} of MWI.

\subsection{Framework of MWI}
\label{FOMWI}

To begin with, we review the
original
framework of MWI,
introduced by Everett~\cite{E57}.
Actually, in what follows \emph{we reformulate
the original framework of MWI
in a form of mathematical rigor from a modern point of view}.
The point of our rigorous treatment of the
original
framework of MWI is
the use of the notion of
\emph{probability measure representation} and
its \emph{induction of
probability measure},
as
presented in Subsection~\ref{MR}.

We stress that MWI is more than just an interpretation of quantum mechanics.
It aims to derive Postulate~\ref{Born-rule}, the Born rule, from the remaining postulates, i.e.,
Postulates~\ref{state_space}, \ref{composition}, and \ref{evolution}.
In this sense, Everett~\cite{E57} proposed MWI as a ``metatheory'' of quantum mechanics.
The point is that in MWI the measurement process is fully treated
as the interaction between a system being measured and an apparatus measuring it, based only on
Postulates~\ref{state_space}, \ref{composition}, and \ref{evolution}.
Then MWI tries to derive Postulate~\ref{Born-rule} in such a setting.

Now,
let us investigate the
setting
of MWI \emph{in terms of our terminology
in a form of mathematical rigor}.
Let $\mathcal{S}$ be an arbitrary quantum system with state space $\mathcal{H}$ of finite dimension $K$.
Consider a measurement over $\mathcal{S}$ described by an arbitrary observable $M$.
Let $\Omega$ be the spectrum of $M$.%
\footnote{The spectrum $\Omega$ is finite, and therefore it is an alphabet.}
Let
$$M=\sum_{k=1}^K f(k)\ket{\phi_k}\bra{\phi_k}$$
be a spectral decomposition of the observable $M$,
where $\{\ket{\phi_1},\dots,\ket{\phi_K}\}$ is an orthonormal basis of $\mathcal{H}$ and
$f\colon\{1,\dots,K\}\to\Omega$ is a surjection.
Let $\mathcal{A}$ be an apparatus performing the measurement of $M$,
which is a quantum system with state space $\ssoa$.%
\footnote{The dimension of the state space $\ssoa$ of the apparatus $\mathcal{A}$ is
not necessarily finite.
Even in the case where the state space $\overline{\mathcal{H}}$ is of infinite dimension,
the mathematical subtleness
which arises from
the infinite dimensionality does not matter,
to the extent of
our
treatment of $\ssoa$ and operators
on
it in this paper.}
According to Postulates~\ref{state_space}, \ref{composition}, and \ref{evolution},
the measurement process of the observable $M$ is described by a unitary operator $U$ such that
\begin{equation}\label{single_measurement}
  U\ket{\phi_k}\otimes\ket{\Phi_{\mathrm{init}}}=\ket{\phi_k}\otimes\ket{\Phi[f(k)]}
\end{equation}
for every $k=1,\dots,K$ \cite{vN55}.
Actually,
$U$ describes the interaction between the system $\mathcal{S}$ and the apparatus $\mathcal{A}$.
The vector $\ket{\Phi_{\mathrm{init}}}\in\ssoa$ is
the initial state of the apparatus $\mathcal{A}$, and $\ket{\Phi[m]}\in\ssoa$ is
a final state of the apparatus $\mathcal{A}$ for each $m\in\Omega$,
with $\braket{\Phi[m]}{\Phi[m']}=\delta_{m,m'}$.
For every $m\in\Omega$, the state $\ket{\Phi[m]}$ indicates that
\emph{the apparatus $\mathcal{A}$ records
the value $m$ of the observable $M$ of the system $\mathcal{S}$}.
By the unitary interaction~\eqref{single_measurement} as a measurement process,
a correlation (i.e., entanglement) is generated between the system and the apparatus.
For each $m\in\Omega$,
let $E_m$ be the projector onto the eigenspace of $M$ with eigenvalue $m$.
Then, the equality~\eqref{single_measurement}
can be
rewritten as the form that
\begin{equation}\label{single_measurement2}
  U\ket{\Psi}\otimes\ket{\Phi_{\mathrm{init}}}=\sum_{m\in\Omega}(E_m\ket{\Psi})\otimes\ket{\Phi[m]}
\end{equation}
for every $\ket{\Psi}\in\mathcal{H}$.

In the framework of MWI,
we consider countably infinite copies of the system $\mathcal{S}$ prepared in an identical state,
and consider a countably infinite repetition of the measurements of the identical observable $M$
performed over each of such copies in a sequential order,
where each of the measurements is described
by the unitary time-evolution \eqref{single_measurement2}
(and equivalently by \eqref{single_measurement}).
As repetitions of the measurement progressed,
correlations between the systems and the apparatuses are being generated in sequence
in the superposition of the total system consisting of the systems and the apparatuses.
The detail
is described as follows.

For each $n\in\N^+$, let $\mathcal{S}_n$ be the $n$th copy of the system $\mathcal{S}$ and
$\mathcal{A}_n$ the $n$th copy of the apparatus $\mathcal{A}$.
All $\mathcal{S}_n$ are prepared in an identical state $\ket{\Psi}$,
and all $\mathcal{A}_n$ are prepared in an identical state $\ket{\Phi_{\mathrm{init}}}$.
The measurement of the observable $M$ is performed over each $\mathcal{S}_n$ one by one
in the increasing order of $n$,
by interacting each $\mathcal{S}_n$ with $\mathcal{A}_n$
according to the unitary time-evolution \eqref{single_measurement2}.
For each $n\in\N^+$,
let $\mathcal{H}_n$ be the state space of the total system consisting of
the first $n$ copies
$\mathcal{S}_1, \mathcal{A}_1, \mathcal{S}_2, \mathcal{A}_2,\dots,\mathcal{S}_n, \mathcal{A}_n$
of the system $\mathcal{S}$ and the apparatus $\mathcal{A}$.
These successive interactions between the copies of
the system $\mathcal{S}$ and the apparatus $\mathcal{A}$ as measurement processes
proceed in the following manner:

The
starting
state of the total system,
which consists of $\mathcal{S}_1$ and $\mathcal{A}_1$,
is $\ket{\Psi}\otimes\ket{\Phi_{\mathrm{init}}}\in\mathcal{H}_1$.
Immediately after the measurement of $M$ over $\mathcal{S}_1$,
the total system results in the state
\begin{align*}
  \sum_{m_1\in\Omega} (E_{m_1}\ket{\Psi})\otimes\ket{\Phi[m_1]}\in\mathcal{H}_1
\end{align*}
by the interaction~\eqref{single_measurement2} as a measurement process.
In general,
immediately before the measurement of $M$ over $\mathcal{S}_n$,
the state of the total system,
which consists of
$\mathcal{S}_1, \mathcal{A}_1, \mathcal{S}_2, \mathcal{A}_2,\dots,\mathcal{S}_n, \mathcal{A}_n$,
is
\begin{align*}
  \sum_{m_1,\dots,m_{n-1}\in\Omega}
  (E_{m_1}\ket{\Psi})\otimes\dots\otimes(E_{m_{n-1}}\ket{\Psi})\otimes\ket{\Psi}
  \otimes\ket{\Phi[m_1]}\otimes\dots\otimes\ket{\Phi[m_{n-1}]}\otimes\ket{\Phi_{\mathrm{init}}}
\end{align*}
in $\mathcal{H}_n$,
where
$\ket{\Psi}$ is the initial state of $\mathcal{S}_n$ and
$\ket{\Phi_{\mathrm{init}}}$ is the initial state of $\mathcal{A}_n$.
Immediately after the measurement of $M$ over $\mathcal{S}_n$,
the total system results in the state
\begin{align}
   &\sum_{m_1,\dots,m_{n}\in\Omega}
  (E_{m_1}\ket{\Psi})\otimes\dots\otimes(E_{m_{n}}\ket{\Psi})
  \otimes\ket{\Phi[m_1]}\otimes\dots\otimes\ket{\Phi[m_{n}]} \label{total_system0} \\
  =&\sum_{m_1,\dots,m_{n}\in\Omega}
  (E_{m_1}\ket{\Psi})\otimes\dots\otimes(E_{m_{n}}\ket{\Psi})
  \otimes\ket{\Phi[m_1\dots m_{n}]} \label{total_system}
\end{align}
in $\mathcal{H}_n$,
by the interaction \eqref{single_measurement2} as a measurement process between
the system $\mathcal{S}_n$ prepared in the state $\ket{\Psi}$
and the apparatus $\mathcal{A}_n$ prepared in the state $\ket{\Phi_{\mathrm{init}}}$.
The vector $\ket{\Phi[m_1\dots m_n]}$ denotes
the vector $\ket{\Phi[m_1]}\otimes\dots\otimes\ket{\Phi[m_n]}$ which represents
the state of $\mathcal{A}_1,\dots,\mathcal{A}_n$.
This state indicates that \emph{the apparatuses $\mathcal{A}_1,\dots,\mathcal{A}_n$ record
the values $m_1\dots m_n$ of the observables $M$ of $\mathcal{S}_1,\dots,\mathcal{S}_n$, respectively}.

In the superposition \eqref{total_system},
on letting $n\to\infty$,
the length of each of the records $m_1\dots m_n$ of the values of the observable $M$
in the apparatuses $\mathcal{A}_1,\dots,\mathcal{A}_n$ diverges to infinity.
The
consideration of
this
limiting case results in the
definition of a \emph{world}.
Namely, a \emph{world} is defined as
an
infinite sequence of records of the values of the observable
in the apparatuses.
Thus, in the case described
so far,
a world is an infinite sequence over $\Omega$,
and the finite records $m_1\dots m_n$ in each state $\ket{\Phi[m_1\dots m_n]}$
in the superposition~\eqref{total_system} of the total system is a \emph{prefix} of a world.

Then,
for aiming at deriving Postulate~\ref{Born-rule}, MWI
assigns
``weight'' to each of worlds.
Namely, it introduces a \emph{probability measure} on the set of all worlds in the following manner.
First, MWI introduces a probability measure representation on the set of prefixes of worlds, i.e.,
the set $\Omega^*$ in this case.
This probability measure representation is given by a function $r\colon\Omega^*\to[0,1]$ with
\begin{equation}\label{rpmwi}
  r(m_1\dotsc m_n)=\prod_{k=1}^n\bra{\Psi}E_{m_k}\ket{\Psi},
\end{equation}
which is the square of the norm of each state
$(E_{m_1}\ket{\Psi})\otimes\dots\otimes(E_{m_{n}}\ket{\Psi})\otimes\ket{\Phi[m_1\dots m_{n}]}$
in the superposition~\eqref{total_system}.
Due to
$\sum_{m\in\Omega}E_m=I$,
it is easy to check that $r$ is certainly a probability measure representation over $\Omega$.
We call the probability measure representation $r$
\emph{the
measure representation
for the prefixes of
worlds}.
Then MWI tries to derive Postulate~\ref{Born-rule}
by adopting
the probability measure \emph{induced by}
the
measure representation $r$
for the prefixes of
worlds
as \emph{the probability measure on the set of all worlds}.

For later use,
we define
the notion of \emph{world} and
the
notion of
\emph{the
measure representation
for the prefixes of
worlds}
in a general setting,
as in the following:
First, recall that the Born rule, Postulate~\ref{Born-rule},
can be generalized to the \emph{general measurement postulate},
which is based on the notion of \emph{measurement operators},
instead of the notion of projectors as in Postulate~\ref{Born-rule}
(Nielsen and Chuang \cite[Subsection 2.2.3]{NC00}).
Measurement operators are a collection $\{M_m\}_{m\in\Omega}$ of
operators, acting on the state space of the system being measured,
which satisfy
the \emph{completeness equation},
\begin{equation}\label{completeness-equation}
  \sum_{m\in\Omega}M_m^{\dagger} M_m=I.
\end{equation}
Here, $\Omega$ is an alphabet and is the set of all possible outcomes of the measurement.
Actually, Postulate~\ref{Born-rule} is shown to be equivalent to the general measurement postulate,
when augmented with an ancilla system and the ability to perform unitary transformations,
based on Postulates~\ref{composition} and \ref{evolution}
(see Nielsen and Chuang \cite[Subsection 2.2.8]{NC00} for the detail).
Similarly,
based on Postulates~\ref{composition} and \ref{evolution},
the unitary time-evolution \eqref{single_measurement2} of measurement process,
which is based on projectors,
can be generalized to the following form of unitary time-evolution of measurement process,
which is based on measurement operators
$\{M_m\}_{m\in\Omega}$:
\begin{equation}\label{single_measurement-mo}
  U\ket{\Psi}\otimes\ket{\Phi_{\mathrm{init}}}=\sum_{m\in\Omega}(M_m\ket{\Psi})\otimes\ket{\Phi[m]}
\end{equation}
for every state $\ket{\Psi}$ of the system being measured.
The vector $\ket{\Phi_{\mathrm{init}}}$ is
the initial state of the apparatus measuring the system, and
$\ket{\Phi[m]}$ is a final state of the apparatus for each $m\in\Omega$,
with $\braket{\Phi[m]}{\Phi[m']}=\delta_{m,m'}$.
For every $m\in\Omega$, the state $\ket{\Phi[m]}$ indicates that
\emph{the apparatus records the value $m$
as the measurement outcome}.

To see this, let $\{M_m\}_{m\in\Omega}$ be arbitrary measurement operators
acting on the state space $\mathcal{H}$ of the system $\mathcal{S}$ being measured.
We introduce an ancilla system $\ancilla$ with state space $\mathcal{H}^{\mathrm{a}}$ of
finite dimension
$\#\Omega$,
in addition to the original system $\mathcal{S}$.
Let $\{\ket{\Phi_m^{\mathrm{a}}}\}_{m\in\Omega}$ be an orthonormal basis of $\mathcal{H}^{\mathrm{a}}$.
We consider a unitary operator $U_1$ acting on $\mathcal{H}\otimes\mathcal{H}^{\mathrm{a}}$ such that
\begin{equation}\label{single_measurement-umo}
  U_1\ket{\Psi}\otimes\ket{\Phi_{\mathrm{init}}^{\mathrm{a}}}=
  \sum_{m\in\Omega}(M_m\ket{\Psi})\otimes\ket{\Phi_m^{\mathrm{a}}}
\end{equation}
for every $\ket{\Psi}\in\mathcal{H}$,
where $\ket{\Phi_{\mathrm{init}}^{\mathrm{a}}}$ is a specific state of $\ancilla$.
It is easy to see that such a unitary operator $U_1$ exists, due to \eqref{completeness-equation}.
Let $\{E_m\}_{m\in\Omega}$ be a collection of projectors such that
$E_m=\ket{\Phi_m^{\mathrm{a}}}\bra{\Phi_m^{\mathrm{a}}}$ for every $m\in\Omega$.
We then consider an apparatus $\mathcal{A}$
whose interaction with the ancilla system $\ancilla$ as a measurement process
is described by the following unitary time-evolution:
\begin{equation}\label{single_measurement2mo}
  U_2\ket{\Phi^{\mathrm{a}}}\otimes\ket{\Theta_{\mathrm{init}}}=
  \sum_{m\in\Omega}(E_m\ket{\Phi^{\mathrm{a}}})\otimes\ket{\Theta[m]}
\end{equation}
for every $\ket{\Phi^{\mathrm{a}}}\in\mathcal{H}^{\mathrm{a}}$,
where the vector $\ket{\Theta_{\mathrm{init}}}$ is
the initial state of the apparatus $\mathcal{A}$, and $\ket{\Theta[m]}$ is
a final state of the apparatus $\mathcal{A}$ for each $m\in\Omega$,
with $\braket{\Theta[m]}{\Theta[m']}=\delta_{m,m'}$.
Note that the unitary time-evolution \eqref{single_measurement2mo} is
a spacial case of
the unitary time-evolution \eqref{single_measurement2} where each $E_m$ is of rank $1$.
Then, we can see that
the sequential application of $U_1$ and $U_2$ to the composite system
consisting of the system $\mathcal{S}$, the ancilla system $\ancilla$, and the apparatus $\mathcal{A}$
results in the single unitary time-evolution $U$ given by \eqref{single_measurement-mo}.
Here we newly regard the composite system consisting of $\ancilla$ and $\mathcal{A}$ as an apparatus,
the state $\ket{\Phi_{\mathrm{init}}}:=\ket{\Phi_{\mathrm{init}}^{\mathrm{a}}}\otimes\ket{\Theta_{\mathrm{init}}}$ as an initial state of the
renewed
apparatus,
and
the state $\ket{\Phi[m]}:=\ket{\Phi_m^{\mathrm{a}}}\otimes\ket{\Theta[m]}$ is
a final state of the
renewed
apparatus for each $m\in\Omega$.
In this manner, based on Postulates~\ref{composition} and \ref{evolution},
the unitary time-evolution \eqref{single_measurement2} which describes the measurement process
based on an observable
can be generalized over the general measurement scheme
which is described by measurement operators.

Based on the above argument,
the general definitions of
the notion of \emph{world} and
the notion of
\emph{the
measure representation
for the prefixes of
worlds}
are
given as follows.

\begin{definition}[World and the measure representation for the prefixes of worlds]\label{pmrpwst}
Consider an arbitrary
finite-dimensional
quantum system $\mathcal{S}$ and
a
measurement over $\mathcal{S}$
described by arbitrary \emph{measurement operators $\{M_m\}_{m\in\Omega}$},
where the measurement process is described by \eqref{single_measurement-mo} as an interaction
of the system $\mathcal{S}$ with an apparatus $\mathcal{A}$.
We suppose the following situation:
\begin{enumerate}
\item There are countably infinite copies $\mathcal{S}_1, \mathcal{S}_2, \mathcal{S}_3 \dotsc$
  of the system $\mathcal{S}$ and countably infinite copies
  $\mathcal{A}_1, \mathcal{A}_2, \mathcal{A}_3, \dotsc$ of the apparatus $\mathcal{A}$.
\item For each $n\in\N^+$, the system $\mathcal{S}_n$ is prepared in a state $\ket{\Psi_n}$,%
  \footnote{In Definition~\ref{pmrpwst}, all $\ket{\Psi_n}$ are not required to be an identical state.}
  while the apparatus $\mathcal{A}_n$ is prepared in a state $\ket{\Phi_{\mathrm{init}}}$,
  and then the measurement described by $\{M_m\}_{m\in\Omega}$ is performed over $\mathcal{S}_n$
  by interacting it with the apparatus $\mathcal{A}_n$ according to
  the unitary time-evolution
  \eqref{single_measurement-mo}.
\item Starting the measurement described by $\{M_m\}_{m\in\Omega}$ over $\mathcal{S}_1$,
  the measurement described by $\{M_m\}_{m\in\Omega}$ over each $\mathcal{S}_n$ is
  performed in the increasing order of $n$.
\end{enumerate}
We then note that,
for each $n\in\N^+$,
immediately after the measurement described by $\{M_m\}_{m\in\Omega}$ over $\mathcal{S}_n$,
the state of the total system consisting of
$\mathcal{S}_1, \mathcal{A}_1, \mathcal{S}_2, \mathcal{A}_2,\dots,\mathcal{S}_n, \mathcal{A}_n$ is
$$\ket{\Theta_n}:=\sum_{m_1,\dots,m_n\in\Omega}\ket{\Theta(m_1,\dots,m_n)},$$
where
$\ket{\Theta(m_1,\dots,m_n)}
:=(M_{m_1}\ket{\Psi_1})\otimes\dots\otimes(M_{m_{n}}\ket{\Psi_n})
\otimes\ket{\Phi[m_1]}\otimes\dots\otimes\ket{\Phi[m_n]}$.%
\footnote{The state $\ket{\Theta(m_1,\dots,m_n)}$ corresponds to the state \eqref{total_system0}
in the spacial case where the measurements of an observable are treated.}
The vectors $M_{m_1}\ket{\Psi_1},\dots,M_{m_{n}}\ket{\Psi_n}$ are states of
$\mathcal{S}_1,\dots,\mathcal{S}_n$, respectively, and
the vectors $\ket{\Phi[m_1]},\dots,\ket{\Phi[m_n]}$ are
states of $\mathcal{A}_1,\dots,\mathcal{A}_n$, respectively.
The state
vector
$\ket{\Theta_n}$ of the total system is normalized while
each of the vectors $\{\ket{\Theta(m_1,\dots,m_n)}\}_{m_1,\dots,m_n\in\Omega}$
is not necessarily normalized.
Then,
\emph{the
measure representation for the prefixes of
worlds}
is defined
as
a
function $p\colon\Omega^*\to[0,1]$ such that
\begin{equation}\label{p=bTmkT}
  p(m_1\dotsc m_n)=\braket{\Theta(m_1,\dots,m_n)}{\Theta(m_1,\dots,m_n)}.
\end{equation}
Moreover,
an infinite sequence over $\Omega$,
i.e., an infinite sequence of
possible
outcomes of the measurement described by $\{M_m\}_{m\in\Omega}$,
is called a \emph{world}.
\qed
\end{definition}

In Definition~\ref{pmrpwst},
it is easy to check that the function $p$ defined by \eqref{p=bTmkT} is
certainly a probability measure representation over $\Omega$.

Now, let us return to the specific situation considered above,
where the infinite repetition of the measurements of the observable $M$ is
treated.
Each of these measurements is described by the unitary time-evolution \eqref{single_measurement2}.
It is then easy to see that the projectors $\{E_m\}_{m\in\Omega}$ in \eqref{single_measurement2} satisfy
the completeness equation,
\begin{equation*}%
  \sum_{m\in\Omega}E_m^{\dagger} E_m=I.
\end{equation*}
Thus, $\{E_m\}_{m\in\Omega}$ are measurement operators
which describe the measurement of $M$ via
the unitary time-evolution \eqref{single_measurement2} as a measurement process.
Therefore, Definition~\ref{pmrpwst} can be applied to this situation.
Thus, according to Definition~\ref{pmrpwst},
we see that a world is certainly an infinite sequence over $\Omega$ in this situation, and
moreover
we
see that the function $r$ given by \eqref{rpmwi} is certainly
the
measure representation for the prefixes of
worlds
in this situation.
Furthermore, it follows from \eqref{mr}, \eqref{BmPosgsigma=Psigma}, and \eqref{rpmwi} that
the probability measure induced by
the
measure representation $r$ for the prefixes of
worlds
is just the Bernoulli measure $\lambda_P$ on $\Omega^\infty$,
where $P$ is a finite probability space on $\Omega$ such that
\begin{equation}\label{Pm=ip}
  P(m)=\bra{\Psi}E_m\ket{\Psi}
\end{equation}
for every $m\in\Omega$.

\subsection{Failure of the original MWI}

We continue to investigate
the original framework of MWI based on
the infinite repetition of the measurements of the observable $M$,
whose setting is developed in the previous subsection.
Let $R\subset\Omega^\infty$ be a ``typical'' property with respect to the Bernoulli measure $\lambda_P$.
Namely, let $R$ be
a set $R\in\mathcal{B}_{\Omega}$
such that
$$\Bm{P}{R}=1.$$
For example,
we can consider as $R$ the set of all worlds for which the \emph{frequency interpretation} holds,
i.e., as the set of all $\alpha\in\Omega^\infty$ such that
$$\lim_{n\to\infty}\frac{N_m(\rest{\alpha}{n})}{n}=\bra{\Psi}E_m\ket{\Psi}$$
holds for every $m\in\Omega$ (see Theorem~\ref{FI} for the notation).
Actually, using Theorems~\ref{MLmae} and \ref{FI} we can show that $\Bm{P}{R}=1$
holds for this specific
$R$.
Then,
by definition,
the property $R$ holds in ``almost all'' worlds.
Based on
this
fact,
MWI insists that Postulate~\ref{Born-rule} has been derived from
Postulates~\ref{state_space}, \ref{composition}, and \ref{evolution}.
In this argument
of MWI,
however, what is typical is just a set $R$ of worlds and not an individual world.
For the purpose of deriving  Postulate~\ref{Born-rule},
we should consider the notion of typicality applied to an individual world and not to a set of worlds. 
However, Everett~\cite{E57} does not consider the notion of typicality for an individual world
rigorously or seriously.

Apart from the definition of typicality for an individual world,
the point is \emph{whether our world $\alpha$ is in the set $R$, or not}.
However, Everett~\cite{E57} does not seem to mention this point explicitly.
Can we deduce that  our world $\alpha$ is in $R$ only due to the fact that $\Bm{P}{R}=1$?
Obviously, we cannot do so.
The reason is as follows:
Say
it
was true.
We then have to deduce also that $\alpha\in R\setminus\{\alpha\}$
since $R\setminus\{\alpha\}\in\mathcal{B}_{\Omega}$ and $\Bm{P}{R\setminus\{\alpha\}}=1$.
However, this leads to an apparent contradiction.%
\footnote{If we accept the \emph{thesis} that our world $\alpha$ is in $A$
for every $A\in\mathcal{B}_{\Omega}$ with $\Bm{P}{A}=1$,
then we have that $\alpha\in\Omega^\infty\setminus\{\alpha\}$
since $\Omega^\infty\setminus\{\alpha\}\in\mathcal{B}_{\Omega}$ and
$\Bm{P}{\Omega^\infty\setminus\{\alpha\}}=1$,
which leads to the apparent contradiction that $\alpha\neq\alpha$.
Thus, we cannot accept this thesis anyhow.}
Hence,
the argument
by Everett
is unclear in this regard.

Moreover,  as we already pointed out,
there is no operational characterization of the notion of probability in Postulate~\ref{Born-rule},
the Born rule,
while it makes a statement about the ``probability'' of measurement outcomes.
Therefore, what MWI has to show for deriving Postulate~\ref{Born-rule} is unclear
although the argument in MWI
itself
is rather operational.
We have no adequate criterion to confirm
that we have certainly accomplished the derivation of Postulate~\ref{Born-rule} based on MWI,
since Postulate~\ref{Born-rule} is vague from an operational point of view.
By contrast,
the replacement of Postulate~\ref{Born-rule} by Postulate~\ref{Tadaki-rule} makes this clear
since there is no ambiguity in Postulate~\ref{Tadaki-rule} from an operational point of view.

In the next section
we make the argument by Everett clear
by introducing
\emph{the principle of typicality}, i.e., Postulate~\ref{POT} below.
Then, in the subsequent sections we derive Postulate~\ref{Tadaki-rule} and Postulate~\ref{Tadaki-rule2}
from the principle of typicality \emph{in a unified manner}.

\section{The principle of typicality}
\label{SecPOT}

As we saw in the
preceding
section,
for deriving Postulate~\ref{Born-rule}, i.e., the Born rule,
the original
MWI~\cite{E57}
would seem
to have wanted
to assume
that our world is ``typical'' or ``random'' among many coexisting worlds.
However,
the proposal of
the
MWI by Everett was nearly a decade earlier than the advent of algorithmic randomness.
Actually,
Everett \cite{E57} proposed MWI in 1957
while the notion of Martin-L\"of randomness was
introduced by Martin-L\"of \cite{M66} in 1966.
Thus, the assumption of ``typicality'' by Everett in MWI was not
rigorous from a mathematical point of view.

The notion of ``typicality'' or ``randomness'' is just
the
research object of algorithmic randomness.
Based on
the notion of
Martin-L\"of randomness with respect to a probability measure,
we introduce a postulate, called
the \emph{principle of typicality} as follows.
The principle of typicality is considered to be
a \emph{refinement} and therefore a \emph{clarification}
of the obscure assumption of
``typicality''
by Everett \cite{E57}.

\begin{postulate}[The principle of typicality]\label{POT}
Our world is \emph{typical}.
Namely, our world is Martin-L\"of random with respect to the probability measure on the set of all worlds,
induced by
the
measure representation
for the prefixes of
worlds,
in the superposition of the total system
which consists of
systems being measured and apparatuses measuring them.
\qed
\end{postulate}

Let $\Omega$ be an arbitrary alphabet,
and $\mu$ be an arbitrary probability measure on $\Omega^\infty$.
Recall the basic idea of Martin-L\"of randomness, presented in Section~\ref{MLRam}, that
\emph{we think of an infinite sequence over $\Omega$ as
random with respect to the probability measure $\mu$ if it is in no effective null set for $\mu$}.
In the idea, we identify $\Omega$ with the set of all possible measurement outcomes and
infinite sequences over $\Omega$ with worlds.
Moreover, we identify the probability measure $\mu$
with the probability measure
induced by
the
measure representation
for the prefixes of
worlds.
Then,
the specific implementation of the basic idea based on this identification
naturally
results
in Postulate~\ref{POT}, the principle of typicality.
This is the \emph{underlying idea} of the principle of typicality.

In the case of the setting of the original framework of MWI based on
the infinite repetition of the measurements of the observable $M$,
which we
developed
in Subsection~\ref{FOMWI},
the
measure representation for the prefixes of
worlds
is $r$ given by \eqref{rpmwi},
and the probability measure induced by $r$ is $\lambda_P$,
where $P$ is given by \eqref{Pm=ip}.
Hence, Postulate~\ref{POT} is precisely Postulate~\ref{Tadaki-rule} in this case.
In Section~\ref{POT-pure-states},
we
describe the detail of the derivation of Postulate~\ref{Tadaki-rule} from Postulate~\ref{POT}
(together with Postulates~\ref{state_space}, \ref{composition}, and \ref{evolution})
in this case.
In this way,
we have made clear the argument by Everett~\cite{E57},
based on
Postulate~\ref{POT}, the principle of typicality.

Furthermore,
we can derive Postulate~\ref{Tadaki-rule2} from Postulate~\ref{POT}
together with Postulates~\ref{state_space}, \ref{composition}, and \ref{evolution}
in several scenarios of the setting of measurements.
We can do this by considering more complicated interactions between systems and
apparatuses as measurement processes than one used for deriving Postulate~\ref{Tadaki-rule}.
Recall that in our framework
a mixed state, on which measurements are performed in the setting of Postulate~\ref{Tadaki-rule2},
is an infinite sequence of pure states,
as defined in (i) of Definition~\ref{def-mixed-state}.
This
implies
that
we have to perform measurements on a mixed state while generating it.
We
have investigated several scenarios which implement this setting.
In all the scenarios which we have considered so far,
Postulate~\ref{Tadaki-rule2} can be derived from Postulate~\ref{POT}
together with Postulates~\ref{state_space}, \ref{composition}, and \ref{evolution}.
In Section~\ref{POT-mixed-states} below,
we describe the detail of the derivation of Postulate~\ref{Tadaki-rule2} from
Postulate~\ref{POT}
(together with Postulates~\ref{state_space}, \ref{composition}, and \ref{evolution})
in the simplest scenario,
where a mixed state being measured is
a Martin-L\"of $P$-random infinite sequence over \emph{mutually orthogonal} pure states.
In Section~\ref{POT-mixed-states2} we describe the detail of the same derivation
\emph{in a more general scenario}, where a mixed state being measured is
a Martin-L\"of $P$-random infinite sequence over general \emph{mutually non-orthogonal} pure states.

\section{Further mathematical preliminaries}
\label{FMP}

In order to show the results in the subsequent sections, we need
several theorems
on Martin-L\"of $P$-randomness
from
Tadaki \cite{T14,T15,T16arXiv}.
These theorems played
a key part
in developing
an \emph{operational characterization of the notion of probability}
based on Martin-L\"of $P$-randomness in
Tadaki \cite{T14,T15,T16arXiv}.
We enumerate them and their corollaries in this section.
Let $\Omega$ be an arbitrary alphabet throughout this section.

\subsection{Avoiding events with probability zero}

First, as an elaboration of Theorem~\ref{one_probability}, we can show the following theorem.

\begin{theorem}\label{zero_probability}
Let $P\in\PS(\Omega)$, and let $a\in\Omega$.
Suppose that $\alpha$ is a Martin-L\"of $P$-random infinite sequence over $\Omega$ and $P(a)=0$.
Then $\alpha$ does not contain $a$.
\end{theorem}

\begin{proof}
See Tadaki~\cite[Theorem 10]{T16arXiv} for the proof.
\end{proof}

Combining this with the law of large numbers, we can derive a stronger result, as follows.

\begin{theorem}\label{nonzero-prob-appears}
Let $P\in\PS(\Omega)$.
If $\alpha$ is a Martin-L\"of $P$-random infinite sequence over $\Omega$, then
$$\{\alpha(n)\mid n\in\N^+\}=\{a\in\Omega\mid P(a)>0\}.$$
\end{theorem}

\begin{proof}
The result follows from Theorems~\ref{zero_probability} and \ref{FI}.
\end{proof}

\subsection{Marginal probability}

The following two theorems are frequently used
in the rest of this paper.

\begin{theorem}\label{contraction}
Let $P\in\PS(\Omega)$.
Let $\alpha$ be a Martin-L\"of $P$-random infinite sequence over $\Omega$, and
let $a$ and $b$ be distinct elements
of
$\Omega$.
Suppose that $\beta$ is an infinite sequence
over $\Omega\setminus\{b\}$
obtained by replacing all occurrences of $b$ by $a$ in $\alpha$.
Then $\beta$ is Martin-L\"of $Q$-random,
where $Q\in\PS(\Omega\setminus\{b\})$ such that $Q(x):=P(a)+P(b)$ if $x=a$ and $Q(x):=P(x)$ otherwise.
\end{theorem}

\begin{proof}
See Tadaki~\cite[Theorem 13]{T16arXiv} for the proof.
\end{proof}

\begin{theorem}\label{contraction2}
Let $\Omega$ and $\Theta$ be alphabets.
Let $P\in\PS(\Omega\times\Theta)$, and let $\alpha\in(\Omega\times\Theta)^\infty$.
Suppose that $\beta$ is an infinite sequence over $\Omega$ obtained from $\alpha$
by replacing each element $(m,l)$ in $\alpha$ by $m$.
If $\alpha$ is Martin-L\"of $P$-random then $\beta$ is Martin-L\"of $Q$-random,
where $Q\in\PS(\Omega)$ such that
$$Q(m):=\sum_{l\in\Theta}P(m,l)$$
for every $m\in\Omega$.
\end{theorem}

\begin{proof}
The result
is obtained by applying Theorem~\ref{contraction} repeatedly.
\end{proof} 

\subsection{Conditional probability}

The notion of \emph{conditional probability} in a finite probability space can be represented by
the notion of Martin-L\"of $P$-randomness in a natural manner as follows.

First, we recall
the notion of conditional probability in a finite probability space.
Let $P\in\PS(\Omega)$, and let $B\subset\Omega$ be an event on the finite probability space $P$.
Suppose that $P(B)>0$.
Then, for each event $A\subset\Omega$,
the \emph{conditional probability of A given B}, denoted
$P(A|B)$, is defined as $P(A\cap B)/P(B)$.
This notion defines a finite probability space $P_B\in\PS(B)$ such that
$P_B(a)=P(\{a\}|B)$ for every $a\in B$.

When an infinite sequence $\alpha\in\Omega^\infty$ contains infinitely many elements from $B$,
$\cond{B}{\alpha}$ is defined as an infinite sequence in $B^\infty$ obtained from $\alpha$
by eliminating all elements
of
$\Omega\setminus B$ occurring in $\alpha$.
If $\alpha$ is Martin-L\"of $P$-random for the finite probability space $P$ and $P(B)>0$,
then $\alpha$ contains infinitely many elements from $B$
due to Theorem~\ref{appearance-infinitely-many} below.
Therefore, $\cond{B}{\alpha}$ is properly defined in this case.
Note that the notion of $\cond{B}{\alpha}$ in our
theory
is introduced by
Tadaki \cite{T14,T16arXiv},
suggested by
the notion of \emph{partition}
in the theory of \emph{collectives}
introduced
by von Mises \cite{vM64} (see Tadaki \cite{T16arXiv} for the detail).

We can then show Theorem~\ref{conditional_probability} below, which states that
Martin-L\"of $P$-random sequences are \emph{closed under conditioning}.

\begin{theorem}\label{appearance-infinitely-many}
Let $P\in\PS(\Omega)$, and let $\alpha\in\Omega^*$.
Suppose that $\alpha$ is Martin-L\"of $P$-random.
For every $m\in\Omega$, if $m$ appears in $\alpha$ then
$m$ appears in $\alpha$ infinitely many times.
\end{theorem}

\begin{proof}
The result follows from Theorems~\ref{nonzero-prob-appears} and \ref{FI}.
\end{proof}

\begin{theorem}[Closure property under conditioning, Tadaki~\cite{T14}]\label{conditional_probability}
Let $P\in\PS(\Omega)$,
and let $B\subset\Omega$ be an event on the finite probability space $P$ with $P(B)>0$.
For every $\alpha\in\Omega^\infty$,
if $\alpha$ is Martin-L\"of $P$-random
then $\cond{B}{\alpha}$ is Martin-L\"of $P_B$-random for the finite probability space $P_B$.
\end{theorem}

\begin{proof}
See Tadaki~\cite[Theorem 18]{T16arXiv} for the proof.
\end{proof}

\subsection{Occurrence of a specific event}

Let $P\in\PS(\Omega)$, and let $A\subset\Omega$ be an event on the finite probability space $P$.
For each $\alpha\in\Omega^\infty$,
we use $\chara{A}{\alpha}$ to denote the infinite binary sequence such that,
for every $n\in\N^+$,
its $n$th element $(\chara{A}{\alpha})(n)$ is $1$ if $\alpha(n)\in A$ and $0$ otherwise.
The pair $(P,A)$ induces a finite probability space $\charaps{P}{A}\in\PS(\{0,1\})$ such that
$(\charaps{P}{A})(1)=P(A)$ and $(\charaps{P}{A})(0)=1-P(A)$.
Note that the notions of $\chara{A}{\alpha}$ and $\charaps{P}{A}$ in our theory together correspond to
the notion of \emph{mixing} in the theory of collectives by von Mises \cite{vM64}.

We can then show the following theorem.

\begin{theorem}\label{charaA}
Let $P\in\PS(\Omega)$, and
let $A\subset\Omega$.
For every $\alpha\in\Omega^\infty$, if $\alpha$ is Martin-L\"of $P$-random
then $\chara{A}{\alpha}$ is Martin-L\"of $\charaps{P}{A}$-random
for the finite probability space $\charaps{P}{A}$.
\end{theorem}

\begin{proof}
See Tadaki~\cite[Theorem 17]{T16arXiv} for the proof.
\end{proof}

\section{Derivation of Postulate~\ref{Tadaki-rule} from the principle of typicality}
\label{POT-pure-states}

In this section, we derive Postulate~\ref{Tadaki-rule} from Postulate~\ref{POT}, the principle of typicality,
together with Postulates~\ref{state_space}, \ref{composition}, and \ref{evolution}.

For deriving Postulate~\ref{Tadaki-rule},
consider a quantum measurement described by an observable $M$,
as in Postulate~\ref{Tadaki-rule}.
Here,
we call a Hermitian operator on the state space of the system being measured an \emph{observable}.
The observable $M$ has a spectral decomposition,
$$M=\sum_{m\in\Omega} m E_m,$$
where $E_m$ is the projector onto the eigenspace of $M$ with eigenvalue $m$.
Then, according to Postulates~\ref{state_space}, \ref{composition}, and \ref{evolution},
the measurement process of the observable $M$ is described by a unitary operator $U$
satisfying \eqref{single_measurement2}.
Based on the equation~\eqref{single_measurement2} we see, in particular, that
\emph{the set of possible outcomes of the measurement of $M$ is the spectrum $\Omega$ of $M$},
as stated in Postulate~\ref{Tadaki-rule} as one of its conclusions.

Now, let us assume
that \emph{the measurements of the observable $M$ are
repeatedly performed over identical quantum systems whose states are all $\ket{\Psi}$, and
the infinite sequence $\alpha\in \Omega^\infty$ of measurement outcomes is being generated},
as assumed in Postulate~\ref{Tadaki-rule}.
This is just the situation that we considered in
Subsection~\ref{FOMWI}.
In other words,
under the assumption above, we are considering
the total system consisting of the copies of
the system $\mathcal{S}$ and the apparatus $\mathcal{A}$
in
the setting of the original framework of MWI based on
the infinite repetition of the measurements of the observable $M$, which we
have developed
in Subsection~\ref{FOMWI}
based on Postulates~\ref{state_space}, \ref{composition}, and \ref{evolution}.
In this situation, according to Definition~\ref{pmrpwst},
a \emph{world} is an infinite sequence over $\Omega$ and
the
\emph{measure representation for the prefixes of
worlds}
is the function $r\colon\Omega^*\to[0,1]$ given by \eqref{rpmwi},
as we saw in Subsection~\ref{FOMWI}.
Furthermore, as we also saw in Subsection~\ref{FOMWI},
the probability measure induced by $r$ is the Bernoulli measure $\lambda_P$ on $\Omega^\infty$,
where $P$ is given by \eqref{Pm=ip}.
The infinite sequence $\alpha$ in the assumption
is \emph{our world}
under
the infinite repetition of the measurements of $M$
in this setting.

Then, it follows from Postulate~\ref{POT}, the principle of typicality,
that \emph{$\alpha$ is Martin-L\"of random with respect to
the probability measure $\lambda_P$ on $\Omega^\infty$}.
Hence, according to Definition~\ref{ML_P-randomness},
we see that
\emph{$\alpha$ is Martin-L\"of $P$-random
with
the finite probability space $P$ on $\Omega$
satisfying \eqref{Pm=ip}}, as stated in Postulate~\ref{Tadaki-rule} as one of its conclusions.
Using Theorem~\ref{zero_probability}  we
note
that $P(\alpha(n))>0$ for every $n\in\N^+$.
Thus, it follows from \eqref{Pm=ip} that
\begin{equation}\label{non-zero-inner-product}
  \bra{\Psi}E_{\alpha(n)}\ket{\Psi}>0
\end{equation}
for every $n\in\N^+$.

Let $n$ be an arbitrary positive integer.
In the superposition~\eqref{total_system0} of the total system consisting of
$\mathcal{S}_1, \mathcal{A}_1, \mathcal{S}_2, \mathcal{A}_2,\dots,\mathcal{S}_n, \mathcal{A}_n$,
consider the specific state
\begin{equation}\label{specific_state}
  (E_{m_1}\ket{\Psi})\otimes\dots\otimes(E_{m_{n}}\ket{\Psi})
  \otimes\ket{\Phi[m_1]}\otimes\dots\otimes\ket{\Phi[m_{n}]}
\end{equation}
such that $m_{k}=\alpha(k)$ for every $k=1,\dots,n$.
Note from \eqref{non-zero-inner-product} that
the vector \eqref{specific_state} is non-zero and therefore can be a state vector certainly,
when it is normalized.
Since $\alpha$ is our world,
the state \eqref{specific_state} is
\emph{the state of the total system consisting of
$\mathcal{S}_1, \mathcal{A}_1, \mathcal{S}_2, \mathcal{A}_2,\dots,\mathcal{S}_n, \mathcal{A}_n$
that we have perceived immediately after the
measurement of the observable $M$
over the system $\mathcal{S}_n$}.
Thus, since $m_{n}=\alpha(n)$,
the state of
the system $\mathcal{S}_n$
immediately after the measurement of $M$ over it \emph{in our world}
is given by
\begin{equation*}
  \frac{E_{\alpha(n)}\ket{\Psi}}{\sqrt{\bra{\Psi}E_{\alpha(n)}\ket{\Psi}}},
\end{equation*}
where the vector is normalized since a state vector has to be a unit vector
according to a convention adopted in Postulate~\ref{state_space}.
Since the positive integer $n$ is arbitrary,
this
means
that
in our world,
\emph{for each of the measurements of $M$,
the state of the system $\mathcal{S}$ immediately after the measurement is
\begin{equation*}
  \frac{E_m\ket{\Psi}}{\sqrt{\bra{\Psi}E_m\ket{\Psi}}},
\end{equation*}
where $m$ is the corresponding measurement outcome},
as stated in the last sentence of Postulate~\ref{Tadaki-rule} as one of its conclusions.

Hence, Postulate~\ref{Tadaki-rule} is derived from Postulate~\ref{POT}
together with Postulates~\ref{state_space}, \ref{composition}, and \ref{evolution}.

\section{Derivation I of Postulate~\ref{Tadaki-rule2} from the principle of typicality}
\label{POT-mixed-states}

In this section and the next section,
we derive Postulate~\ref{Tadaki-rule2} from Postulate~\ref{POT}, the principle of typicality,
together with Postulates~\ref{state_space}, \ref{composition}, and \ref{evolution}.
In this section, in particular,
we do this in a specific scenario of the setting of measurements
which is considered to be the simplest among many other scenarios.
In order to derive Postulate~\ref{Tadaki-rule2},
we have to consider more complicated interaction between systems and
apparatuses than one used for deriving Postulate~\ref{Tadaki-rule} in the preceding section.
Recall that a mixed state on which the measurements are performed is
a Martin-L\"of $P$-random infinite sequence of pure states,
as defined in Definition~\ref{def-mixed-state}.
Hence, we have to perform measurements on the mixed state while generating it by other measurements.
In the simplest scenario, a mixed state being measured is
a Martin-L\"of $P$-random infinite sequence over \emph{mutually orthogonal} pure states.

Thus, for deriving Postulate~\ref{Tadaki-rule2},
we consider an infinite repetition of two successive measurements of
observables $A$ and $B$ over identical systems prepared
initially in identical
states $\ket{\Psi}$.
In this setting,
while generating a mixed state by the measurements of $A$,
we are performing the measurements of $B$ over the mixed state.

\subsection{Repeated once of measurements}
\label{ROM}

First, the repeated once of the infinite repetition of measurements, i.e.,
the two successive measurements of observables $A$ and $B$, is described as follows.

Let $\mathcal{S}$ be an arbitrary quantum system with state space $\mathcal{H}$ of finite dimension $K$.
Consider arbitrary two measurements over $\mathcal{S}$ described by observables $A$ and $B$.
Let $\Omega$ and $\Theta$ be the spectrums of $A$ and $B$, respectively.
Let
\begin{equation*}
  A=\sum_{k=1}^K f(k)\ket{\psi_k}\bra{\psi_k}
\end{equation*}
be a spectral decomposition of the observable $A$,
where $\{\ket{\psi_1},\dots,\ket{\psi_K}\}$ is an orthonormal basis of $\mathcal{H}$ and
$f\colon\{1,\dots,K\}\to\Omega$ is a surjection,
and let
\begin{equation}\label{spectral-decompositionB}
  B=\sum_{k=1}^K g(k)\ket{\phi_k}\bra{\phi_k}
\end{equation}
be a spectral decomposition of the observable $B$,
where $\{\ket{\phi_1},\dots,\ket{\phi_K}\}$ is an orthonormal basis of $\mathcal{H}$ and
$g\colon\{1,\dots,K\}\to\Theta$ is a surjection.
According to Postulates~\ref{state_space}, \ref{composition}, and \ref{evolution},
the measurement processes of the observables $A$ and $B$ are described
by the following unitary operators $U_A$ and $U_B$, respectively:
\begin{align}
  &U_A\ket{\psi_k}\otimes\ket{\Phi_A^{\mathrm{init}}}=\ket{\psi_k}\otimes\ket{\Phi_A[f(k)]},
  \label{single_measurementA} \\
  &U_B\ket{\phi_k}\otimes\ket{\Phi_B^{\mathrm{init}}}=\ket{\phi_k}\otimes\ket{\Phi_B[g(k)]}
  \label{single_measurementB}
\end{align}
for every $k=1,\dots,K$.
The vectors $\ket{\Phi_A^{\mathrm{init}}}$ and $\ket{\Phi_B^{\mathrm{init}}}$ are
the initial states of the apparatuses measuring $A$ and $B$, respectively, and
$\ket{\Phi_A[m]}$ and $\ket{\Phi_B[l]}$ are
the final states of
ones
for each $m\in\Omega$ and $l\in\Theta$,
with $\braket{\Phi_A[m]}{\Phi_A[m']}=\delta_{m,m'}$ and $\braket{\Phi_B[l]}{\Phi_B[l']}=\delta_{l,l'}$.
For every $m\in\Omega$, the state $\ket{\Phi_A[m]}$ indicates that
\emph{the apparatus measuring the observable $A$ of the system $\mathcal{S}$
records the value $m$ of $A$}, and
for every $l\in\Theta$, the state $\ket{\Phi_B[l]}$ indicates that
\emph{the apparatus measuring the observable $B$ of the system $\mathcal{S}$
records the value $l$ of $B$}.

For each $m\in\Omega$, let $E_m$ be the projector onto the eigenspace of $A$ with eigenvalue $m$,
and for each $l\in\Theta$, let $F_l$ be the projector onto the eigenspace of $B$ with eigenvalue $l$.
Then, the equalities~\eqref{single_measurementA} and \eqref{single_measurementB} can be rewritten,
respectively, as the forms that
\begin{align}
  U_A\ket{\Psi}\otimes\ket{\Phi_A^{\mathrm{init}}}&=\sum_{m\in\Omega}(E_m\ket{\Psi})\otimes\ket{\Phi_A[m]},
  \label{single_measurementA2} \\
  U_B\ket{\Psi}\otimes\ket{\Phi_B^{\mathrm{init}}}&=\sum_{l\in\Theta}(F_l\ket{\Psi})\otimes\ket{\Phi_B[l]}
  \label{single_measurementB2}
\end{align}
for every $\ket{\Psi}\in\mathcal{H}$.

\subsection{\boldmath Infinite repetition of the measurements of the observables $A$ and $B$}

As in
Subsection~\ref{FOMWI}
and Section~\ref{POT-pure-states},
we consider countably infinite copies of the system $\mathcal{S}$.
We prepare each of the copies in an identical state $\ket{\Psi}$,
and then perform the successive measurements of the observables $A$ and $B$ over
each of the copies of the system $\mathcal{S}$ \emph{one by one},
by interacting each of the copies of the system $\mathcal{S}$ with
each of the copies of the apparatuses measuring $A$ and $B$
according to
the unitary time-evolution \eqref{single_measurementA2} and then
the unitary time-evolution \eqref{single_measurementB2}.
For each $n\in\N^+$,
let $\mathcal{H}_n$ be the state space of the total system consisting of
the first $n$ copies of the system $\mathcal{S}$ and the two apparatuses measuring $A$ and $B$.
These successive interactions
as measurement processes
between the copies of the system $\mathcal{S}$ and the copies of the two apparatuses
proceed in the following manner.

The
starting
state of the total system,
which consists of the first copy of the system $\mathcal{S}$ and
the two apparatuses measuring $A$ and $B$,
is $\ket{\Psi}\otimes\ket{\Phi_A^{\mathrm{init}}}\otimes\ket{\Phi_B^{\mathrm{init}}}\in\mathcal{H}_1$.
In this state of the total system,
the measurement of $A$ is performed over the first copy of the system $\mathcal{S}$.
Immediately after
that,
the total system results in the state
\begin{align*}
  \sum_{m_1\in\Omega} (E_{m_1}\ket{\Psi})\otimes\ket{\Phi_A[m_1]}\otimes\ket{\Phi_B^{\mathrm{init}}}
  \in\mathcal{H}_1
\end{align*}
by the interaction~\eqref{single_measurementA2} as the measurement process of $A$.
Then, the measurement of $B$ is performed over the first copy of the system $\mathcal{S}$.
Immediately after that, the total system results in the state
\begin{align*}
  \sum_{m_1\in\Omega}\sum_{l_1\in\Theta}
  (F_{l_1}E_{m_1}\ket{\Psi})\otimes\ket{\Phi_A[m_1]}\otimes\ket{\Phi_B[l_1]}
  \in\mathcal{H}_1
\end{align*}
by the interaction~\eqref{single_measurementB2} as the measurement process of $B$.

In general,
immediately before the measurement of $A$ over the $n$th copy of the system $\mathcal{S}$,
the state of the total system,
which consists of the first $n$ copies of the system $\mathcal{S}$ and
the two apparatuses measuring $A$ and $B$, is
\begin{align*}
  \sum_{m_1,\dots,m_{n-1}\in\Omega}&\sum_{l_1,\dots,l_{n-1}\in\Theta}
  (F_{l_1}E_{m_1}\ket{\Psi})\otimes\dots\otimes(F_{l_{n-1}}E_{m_{n-1}}\ket{\Psi})\otimes\ket{\Psi} \\
  &\otimes\ket{\Phi_A[m_1]}\otimes\ket{\Phi_B[l_1]}\otimes\dots
  \otimes\ket{\Phi_A[m_{n-1}]}\otimes\ket{\Phi_B[l_{n-1}]}
  \otimes\ket{\Phi_A^{\mathrm{init}}}\otimes\ket{\Phi_B^{\mathrm{init}}}
\end{align*}
in $\mathcal{H}_n$.
Then, immediately after the measurement of $A$ over the $n$th copy of the system $\mathcal{S}$,
the total system results in the state
\begin{equation}\label{total_systemA}
\begin{split}
  \sum_{m_1,\dots,m_{n}\in\Omega}&\sum_{l_1,\dots,l_{n-1}\in\Theta}
  (F_{l_1}E_{m_1}\ket{\Psi})\otimes\dots\otimes(F_{l_{n-1}}E_{m_{n-1}}\ket{\Psi})
  \otimes(E_{m_n}\ket{\Psi}) \\
  &\otimes\ket{\Phi_A[m_1]}\otimes\ket{\Phi_B[l_1]}\otimes\dots
  \otimes\ket{\Phi_A[m_{n-1}]}\otimes\ket{\Phi_B[l_{n-1}]}
  \otimes\ket{\Phi_A[m_{n}]}\otimes\ket{\Phi_B^{\mathrm{init}}}
\end{split}
\end{equation}
in $\mathcal{H}_n$ by the interaction \eqref{single_measurementA2},
as the measurement process of $A$,
between the $n$th copy of the system $\mathcal{S}$ in the state $\ket{\Psi}$ and
the $n$th copy of the apparatus measuring $A$ in the state $\ket{\Phi_A^{\mathrm{init}}}$.
Then,
immediately after the measurement of $B$ over the $n$th copy of the system $\mathcal{S}$,
the total system results in the state
\begin{align}
  &\sum_{m_1,\dots,m_{n}\in\Omega}\sum_{l_1,\dots,l_{n}\in\Theta}
  (F_{l_1}E_{m_1}\ket{\Psi})\otimes\dots\otimes(F_{l_{n}}E_{m_{n}}\ket{\Psi}) \nonumber \\
  &\hspace*{32mm}\otimes\ket{\Phi_A[m_1]}\otimes\ket{\Phi_B[l_1]}\otimes\dots
  \otimes\ket{\Phi_A[m_{n}]}\otimes\ket{\Phi_B[l_{n}]} \label{total_systemAB0} \\
  & \nonumber \\
  &=\sum_{m_1,\dots,m_{n}\in\Omega}\sum_{l_1,\dots,l_{n}\in\Theta}
  (F_{l_1}E_{m_1}\ket{\Psi})\otimes\dots\otimes(F_{l_{n}}E_{m_{n}}\ket{\Psi})
  \otimes\ket{\Phi[(m_1,l_1)\dots (m_{n},l_{n})]} \label{total_systemAB}
\end{align}
in $\mathcal{H}_n$ by the interaction \eqref{single_measurementB2},
as the measurement process of $B$,
between the $n$th copy of the system $\mathcal{S}$ in the state $\ket{\Psi}$
and the $n$th copy of the apparatus measuring $B$ in the state $\ket{\Phi_B^{\mathrm{init}}}$,
where the state $\ket{\Phi[(m_1,l_1)\dots (m_n,l_n)]}$ denotes
$\ket{\Phi_A[m_1]}\otimes\ket{\Phi_B[l_1]}\otimes\dots
\otimes\ket{\Phi_A[m_{n}]}\otimes\ket{\Phi_B[l_{n}]}$, and indicates that
\emph{the first $n$ copies of the two apparatuses measuring $A$ and $B$ record
the values $(m_1,l_1)\dots (m_n,l_n)$ of the observables $A$ and $B$, in this order,
of the first $n$ copies of the system $\mathcal{S}$}.

Then, for applying Postulate~\ref{POT}, the principle of typicality,
we have to identify worlds and
the
measure representation for the prefixes of
worlds
in this situation.
Note, from \eqref{single_measurementA2} and \eqref{single_measurementB2},
that the sequential application of $U_A$ and $U_B$ to the composite system
consisting of the system $\mathcal{S}$
and the two apparatuses measuring $A$ and $B$ respectively over the initial state
$\ket{\Psi}\otimes\ket{\Phi_A^{\mathrm{init}}}\otimes\ket{\Phi_B^{\mathrm{init}}}$
results in the following single unitary time-evolution $U$:
\begin{align*}
  U\ket{\Psi}\otimes\ket{\Phi_A^{\mathrm{init}}}\otimes\ket{\Phi_B^{\mathrm{init}}}
  &=\sum_{m\in\Omega}\sum_{l\in\Theta}
    (F_{l}E_{m}\ket{\Psi})\otimes\ket{\Phi_A[m]}\otimes\ket{\Phi_B[l]} \\
  &=\sum_{m\in\Omega}\sum_{l\in\Theta}
    (F_{l}E_{m}\ket{\Psi})\otimes\ket{\Phi[(m,l)]},
\end{align*}
where $U=U_BU_A$.
It is then easy to check that
a collection $\{F_lE_m\}_{(m,l)\in\Omega\times\Theta}$ forms measurement operators.%
\footnote{%
We can
confirm that the collection $\{F_lE_m\}_{(m,l)\in\Omega\times\Theta}$ forms measurement operators,
by checking
directly
that the collection satisfies the completeness equation. 
Note, however, that
the unitarity of $U$ \emph{already} guarantees this fact.}
Thus, the successive measurements of $A$ and $B$ in this order can be regarded
as the \emph{single measurement} which is described by the measurement operators
$\{F_lE_m\}_{(m,l)\in\Omega\times\Theta}$.
Hence, we can apply Definition~\ref{pmrpwst},
where \emph{only} the infinite repetition of a \emph{single measurement} is treated,
to this scenario of the setting of measurements.
Therefore, according to Definition~\ref{pmrpwst},
we see that a \emph{world} is an infinite sequence over $\Omega\times\Theta$ and
the
\emph{measure representation for the prefixes of
worlds}
is given by a function $r\colon(\Omega\times\Theta)^*\to[0,1]$ with
$$r((m_1,l_1)\dotsc (m_n,l_n)):=\prod_{k=1}^n\bra{\Psi}E_{m_k}F_{l_k}E_{m_k}\ket{\Psi},$$
which is the square of the norm of each state
$$(F_{l_1}E_{m_1}\ket{\Psi})\otimes\dots\otimes(F_{l_{n}}E_{m_{n}}\ket{\Psi})
\otimes\ket{\Phi[(m_1,l_1)\dots (m_{n},l_{n})]}$$
in the superposition~\eqref{total_systemAB}.
It follows from \eqref{mr} that
the probability measure induced by the probability measure representation $r$ is
a Bernoulli measure $\lambda_P$ on $(\Omega\times\Theta)^\infty$,
where $P$ is a finite probability space on $\Omega\times\Theta$ such that
\begin{equation}\label{Pm=ipMS}
  P(m,l)=\bra{\Psi}E_mF_lE_m\ket{\Psi}
\end{equation}
for every $m\in\Omega$ and $l\in\Theta$.
The finite records $(m_1,l_1)\dots (m_n,l_n)\in(\Omega\times\Theta)^*$
in each state $\ket{\Phi[(m_1,l_1)\dots (m_n,l_n)]}$
in the superposition~\eqref{total_systemAB} of the total system is a \emph{prefix} of a world.

\subsection{Application of the principle of typicality}
\label{APOT}

Now, let us
apply Postulate~\ref{POT} to the setting developed above.

Let $\omega$ be our world in the infinite repetition of the measurements of $A$ and $B$
in the above setting.
Then $\omega$ is an infinite sequence over $\Omega\times\Theta$.
Since the Bernoulli measure $\lambda_P$ on $(\Omega\times\Theta)^\infty$ is
the probability measure induced by the
measure representation
$r$
for the prefixes of
worlds
in the above setting,
it follows from Postulate~\ref{POT} that \emph{$\omega$ is Martin-L\"of random
with respect to the measure $\lambda_P$ on $(\Omega\times\Theta)^\infty$}.
Therefore, \emph{$\omega$ is Martin-L\"of $P$-random,
where $P$ is the finite probability space on $\Omega\times\Theta$ satisfying \eqref{Pm=ipMS}}.

Let $\alpha$ and $\beta$ be infinite sequences over $\Omega$ and $\Theta$, respectively,
such that $(\alpha(n),\beta(n))=\omega(n)$ for every $n\in\N^+$.
The sequence $\alpha$ is \emph{the infinite sequence of records of the values of the observable $A$
in the corresponding apparatuses measuring $A$ in our world}.
A prefix of $\alpha$ is the finite records $m_1\dots m_n$ in a specific state
$\ket{\Phi_A[m_1]}\otimes\dots\otimes\ket{\Phi_A[m_n]}$ of the
apparatuses,
each of which has measured the observable $A$ of the corresponding copy of $\mathcal{S}$,
in the superposition~\eqref{total_systemAB} of the total system.
Put briefly, the infinite sequence $\alpha$ is \emph{the infinite sequence
of outcomes of the infinitely repeated measurements of the observable $A$
over the infinite copies of $\mathcal{S}$ in our world}.
On the other hand, the sequence $\beta$ is
\emph{the infinite sequence of records of the values of the observable $B$
in the corresponding apparatuses measuring $B$ in our world}.
A prefix of $\beta$ is the finite records $l_1\dots l_n$ in a specific state
$\ket{\Phi_B[l_1]}\otimes\dots\otimes\ket{\Phi_B[l_n]}$ of the
apparatuses,
each of which has measured the observable $B$ of the corresponding copy of $\mathcal{S}$,
in the superposition~\eqref{total_systemAB} of the total system.
Put briefly, the infinite sequence $\beta$ is  \emph{the infinite sequence
of outcomes of the infinitely repeated measurements of the observable $B$
over the infinite copies of $\mathcal{S}$ in our world}.

Using Theorem~\ref{nonzero-prob-appears}, we
note that $P(\omega(n))>0$ for every $n\in\N^+$.
Thus, it follows from \eqref{Pm=ipMS} that
\begin{equation}\label{non-zero-inner-productMS}
  \bra{\Psi}E_{\alpha(n)}F_{\beta(n)}E_{\alpha(n)}\ket{\Psi}>0
\end{equation}
for every $n\in\N^+$.

By Theorem~\ref{contraction2} we have that $\alpha$ is Martin-L\"of $Q$-random, where
$Q$ is a finite probability space on $\Omega$
such that $Q(m):=\sum_{l\in\Theta}P(m,l)$ for every $m\in\Omega$.
Since
$\sum_{l\in\Theta} F_l=I$,
it follows from \eqref{Pm=ipMS} that
\begin{equation}\label{Pm=ipMSQ}
  Q(m)=\bra{\Psi}E_m\ket{\Psi}
\end{equation}
for every $m\in\Omega$, as expected from the points of view of
both the conventional quantum mechanics and Postulate~\ref{Tadaki-rule}.
Note from Theorem~\ref{nonzero-prob-appears} that $Q(\alpha(n))>0$ for every $n\in\N^+$.
Thus, it follows from \eqref{Pm=ipMSQ} that
\begin{equation}\label{non-zero-inner-productMSQ}
  \bra{\Psi}E_{\alpha(n)}\ket{\Psi}>0
\end{equation}
for every $n\in\N^+$.%
\footnote{The inequality \eqref{non-zero-inner-productMSQ} can be directly derived
from the inequality \eqref{non-zero-inner-productMS},
since the equality $\bra{\Psi}E_{\alpha(n)}\ket{\Psi}=0$ implies the equality $E_{\alpha(n)}\ket{\Psi}=0$.
However, we can prove this inequality certainly using Theorem~\ref{nonzero-prob-appears}.
}

\subsection{\boldmath Mixed state resulting from the measurements of the observable $A$}
\label{MSRFMOA}

We calculate the mixed state resulting from the measurements of the observable $A$.

Let $n$ be an arbitrary positive integer.
In the superposition~\eqref{total_systemA} of the total system consisting of
the first $n$ copies of the system $\mathcal{S}$ and the two apparatuses measuring $A$ and $B$
immediately after the measurement of $A$ over the $n$th copy of the system $\mathcal{S}$,
consider the specific state
\begin{equation}\label{specific_stateMS}
\begin{split}
  &(F_{l_1}E_{m_1}\ket{\Psi})\otimes\dots\otimes(F_{l_{n-1}}E_{m_{n-1}}\ket{\Psi})
  \otimes(E_{m_n}\ket{\Psi}) \\
  &\otimes\ket{\Phi_A[m_1]}\otimes\ket{\Phi_B[l_1]}\otimes\dots
  \otimes\ket{\Phi_A[m_{n-1}]}\otimes\ket{\Phi_B[l_{n-1}]}
  \otimes\ket{\Phi_A[m_{n}]}\otimes\ket{\Phi_B^{\mathrm{init}}}
\end{split}
\end{equation}
such that $(m_{k},l_{k})=\omega(k)$ for every $k=1,\dots,n-1$ and $m_n=\alpha(n)$.
Due to \eqref{non-zero-inner-productMS} and \eqref{non-zero-inner-productMSQ},
we note here that the vector \eqref{specific_stateMS} is non-zero and therefore can be
a state vector certainly, when it is normalized.
Since $\omega$ is our world, the state \eqref{specific_stateMS} is
\emph{the state of the total system, consisting of
the first $n$ copies of the system $\mathcal{S}$ and the two apparatuses measuring $A$ and $B$,
that we have perceived
after the $n$th measurement of $A$}.
Therefore, \emph{we perceive}
$$\frac{E_{m_n}\ket{\Psi}}{\sqrt{\bra{\Psi}E_{m_n}\ket{\Psi}}}$$
as the state of the $n$th copy of the system $\mathcal{S}$ at this time. 
Note that $m_n=\alpha(n)$ and $n$ is arbitrary.
Thus, for every positive integer $n$, the state of the $n$th copy of the system $\mathcal{S}$
immediately after the measurement of $A$ over it \emph{in our world} is given by
\begin{equation}\label{def-u-a}
  \ket{\Upsilon^A_n}:=\frac{E_{\alpha(n)}\ket{\Psi}}{\sqrt{\bra{\Psi}E_{\alpha(n)}\ket{\Psi}}}.
\end{equation}
Since $E_mE_{m'}=\delta_{m,m'}E_m$,
note that
\begin{equation}\label{ipdMSQ}
  \braket{\Upsilon^A_n}{\Upsilon^A_{n'}}=\delta_{\alpha(n),\alpha(n')}.
\end{equation}

Let $\Gamma^A:=\{\ket{\Upsilon^A_n}\mid n\in\N^+\}$,
and let $\gamma_A$ be an infinite sequence over $\Gamma^A$ such that
$\gamma_A(n):=\ket{\Upsilon^A_n}$ for every $n\in\N^+$.
Note that $\Gamma^A$ is
an alphabet.
Since $\alpha$ is Martin-L\"of $Q$-random,
it follows from \eqref{ipdMSQ} and Theorem~\ref{nonzero-prob-appears} that
the infinite sequence $\gamma_A$ is Martin-L\"of $Q'$-random,
where $Q'$ is a finite probability space on $\Gamma^A$ such that
\begin{equation}\label{Q'U=Qa}
  Q'(\ket{\Upsilon^A_n})=Q(\alpha(n))
\end{equation}
for every $n\in\N^+$.
Thus, according to (i) of Definition~\ref{def-mixed-state},
the infinite sequence $\gamma_A$ is a mixed state of $\mathcal{S}$.
The mixed state $\gamma_A$ can be
interpreted
as
\emph{the infinite sequence of states of $\mathcal{S}$
resulting from the infinitely repeated measurements of the observable $A$
over the infinite copies of $\mathcal{S}$}.
Note here that the mixed state $\gamma_A$ is
an infinite sequence over \emph{mutually orthogonal} pure states due to \eqref{ipdMSQ}.

Then, according to (ii) of Definition~\ref{def-mixed-state},
the density matrix $\rho_A$ of $\gamma_A$ is given by
\begin{equation*}
  \rho_A:=\sum_{\ket{\psi}\in \Gamma^A}Q'(\ket{\psi})\ket{\psi}\bra{\psi}.
\end{equation*}
Note from \eqref{Q'U=Qa} and \eqref{Pm=ipMSQ} that
$Q'(\ket{\Upsilon^A_n})=\bra{\Psi}E_{\alpha(n)}\ket{\Psi}$ for every $n\in\N^+$.
Thus, it follows from Theorem~\ref{nonzero-prob-appears}, \eqref{def-u-a}, and \eqref{ipdMSQ} that
$$\rho_A=\sum_{m:Q(m)>0}E_m\ket{\Psi}\bra{\Psi}E_m,$$
where the sum is over all $m\in\Omega$ such that $Q(m)>0$.
Hence, using \eqref{Pm=ipMSQ} we have that
\begin{equation}\label{rhoA}
  \rho_A=\sum_{m\in\Omega}E_m\ket{\Psi}\bra{\Psi}E_m,
\end{equation}
as expected from the point of view of the conventional quantum mechanics.

\subsection{\boldmath The infinite sequence $\beta$ of measurement outcomes of the observable $B$}
\label{ISBMOOB}

We determine the property of
the infinite sequence $\beta$ of records of the values of the observable $B$
in the
corresponding
apparatuses measuring $B$ in our world, in terms of the observable $B$ and
the mixed state $\gamma_A$ resulting from the measurements of the observable $A$ in our world.

Using a theorem symmetrical to Theorem~\ref{contraction2}
we have that $\beta$ is Martin-L\"of $R$-random,
where $R$ is a finite probability space on $\Theta$ such that
$R(l):=\sum_{m\in\Omega}P(m,l)$ for every $l\in\Theta$.
It follows from \eqref{Pm=ipMS} and \eqref{rhoA} that
\begin{equation}\label{Pm=ipMSR}
  R(l)=\sum_{m\in\Omega}\bra{\Psi}E_mF_lE_m\ket{\Psi}=\tr(F_l\rho_A)
\end{equation}
for every $l\in\Theta$, as expected from the point of view of the conventional quantum mechanics,
i.e., Postulate~\ref{Born-rule2},
since $\beta$ is
the infinite sequence of records of the values of the observable $B$ in the corresponding apparatuses
measuring $B$ in our world.
Note from Theorem~\ref{nonzero-prob-appears} that $R(\beta(n))>0$ for every $n\in\N^+$.
Therefore, we have
that
\begin{equation}\label{non-zero-inner-productMSR}
  \tr(F_{\beta(n)}\rho_A)>0
\end{equation}
for every $n\in\N^+$.

\subsection{\boldmath Mixed state resulting from the measurements of the observable $B$}

It is worthwhile to calculate the mixed state resulting from the measurements of the observable $B$
although it is not necessary for deriving Postulate~\ref{POT}.

Let $n$ be an arbitrary positive integer.
In the superposition~\eqref{total_systemAB0} of the total system consisting of
the first $n$ copies of the system $\mathcal{S}$ and the two apparatuses measuring $A$ and $B$
immediately after the measurement of $B$ over the $n$th copy of the system $\mathcal{S}$,
consider the specific state
\begin{equation}\label{specific_stateMSAB}
  (F_{l_1}E_{m_1}\ket{\Psi})\otimes\dots\otimes(F_{l_{n}}E_{m_{n}}\ket{\Psi})
  \otimes\ket{\Phi_A[m_1]}\otimes\ket{\Phi_B[l_1]}\otimes\dots
  \otimes\ket{\Phi_A[m_{n}]}\otimes\ket{\Phi_B[l_{n}]}
\end{equation}
such that $(m_{k},l_{k})=\omega(k)$ for every $k=1,\dots,n$.
Due to \eqref{non-zero-inner-productMS},
we note here that the vector \eqref{specific_stateMSAB} is non-zero and therefore can be
a state vector certainly, when it is normalized.
Since $\omega$ is our world, the state \eqref{specific_stateMSAB} is
\emph{the state of the total system, consisting of
the first $n$ copies of the system $\mathcal{S}$ and the two apparatuses measuring $A$ and $B$,
that we have perceived after the $n$th measurement of $B$}.
Therefore, \emph{we perceive
$$\frac{F_{l_n}E_{m_n}\ket{\Psi}}{\sqrt{\bra{\Psi}E_{m_n}F_{l_n}E_{m_n}\ket{\Psi}}}$$
as the state of the $n$th copy of the system $\mathcal{S}$ at this time}. 
Note that $m_n=\alpha(n)$ and $l_n=\beta(n)$, and $n$ is arbitrary.
Thus, for every positive integer $n$, the state of the $n$th copy of the system $\mathcal{S}$
immediately after the measurement of $B$ over it \emph{in our world} is given by
\begin{equation}\label{mrMSB}
  \ket{\Upsilon^B_n}
  :=\frac{F_{\beta(n)}E_{\alpha(n)}\ket{\Psi}}{\sqrt{\bra{\Psi}E_{\alpha(n)}F_{\beta(n)}E_{\alpha(n)}\ket{\Psi}}}.
\end{equation}

Let $\Gamma^B:=\{\ket{\Upsilon^B_n}\mid n\in\N^+\}$,
and let $\gamma_B$ be an infinite sequence over $\Gamma^B$ such that
$\gamma_B(n):=\ket{\Upsilon^B_n}$ for every $n\in\N^+$.
Note that $\Gamma^B$ is an alphabet.
For each $\ket{\psi}\in\Gamma^B$,
we define $S(\ket{\psi})$ as the set
$$\{(\alpha(n),\beta(n))\mid\ket{\psi}=\ket{\Upsilon^B_n}\}.$$
Since $\omega$ is Martin-L\"of $P$-random,
it follows from
Theorem~\ref{nonzero-prob-appears} and
Theorem~\ref{contraction} that
the infinite sequence $\gamma_B$ is Martin-L\"of $P'$-random,
where $P'$ is a finite probability space on $\Gamma^B$ such that
\begin{equation}\label{P'U=Pa}
  P'(\ket{\psi})=\sum_{x\in S(\ket{\psi})}P(x)
\end{equation}
for every $\ket{\psi}\in\Gamma^B$.
Thus, according to (i) of Definition~\ref{def-mixed-state},
the infinite sequence $\gamma_B$ is a mixed state of $\mathcal{S}$.
The mixed state $\gamma_B$ can be
interpreted
as
\emph{the infinite sequence of states of $\mathcal{S}$
resulting from the infinitely repeated measurements of the observable $B$
over the infinite copies of $\mathcal{S}$}.

Then, according to (ii) of Definition~\ref{def-mixed-state},
the density matrix $\rho_B$ of $\gamma_B$ is given by
\begin{equation*}
  \rho_B:=\sum_{\ket{\psi}\in\Gamma^B}P'(\ket{\psi})\ket{\psi}\bra{\psi}.
\end{equation*}
Note from \eqref{P'U=Pa} and \eqref{Pm=ipMS} that
$$P'(\ket{\psi})=\sum_{(m,l)\in S(\ket{\psi})}\bra{\Psi}E_mF_lE_m\ket{\Psi}$$
for every $\ket{\psi}\in\Gamma^B$.
Thus,
since $\omega$ is Martin-L\"of $P$-random,
using
Theorem~\ref{nonzero-prob-appears}
and \eqref{mrMSB}
we have that
\begin{equation*}
  \rho_B
  =\sum_{(m,l):P(m,l)>0}F_lE_m\ket{\Psi}\bra{\Psi}E_mF_l,
\end{equation*}
where the sum is over all $(m,l)\in\Omega\times\Theta$ such that $P(m,l)>0$.
It follows from \eqref{Pm=ipMS} that
\begin{equation*}
  \rho_B=\sum_{(m,l)\in\Omega\times\Theta}F_lE_m\ket{\Psi}\bra{\Psi}E_mF_l.
\end{equation*}
Hence, using \eqref{rhoA} we have that
\begin{equation*}
  \rho_B=\sum_{l\in\Theta} F_l\rho_A F_l,
\end{equation*}
as expected from the point of view of the conventional quantum mechanics.

\subsection{\boldmath Mixed state conditioned by a specific outcome of the measurements of $B$}
\label{MSCSOMB}

We calculate the mixed state resulting from the measurements of the observable $B$,
conditioned by a specific measurement outcome of $B$.
This mixed state is the post-measurement state
mentioned in the last sentence of Postulate~\ref{Tadaki-rule2}.

Let $\Lambda:=\{\beta(n)\mid n\in\N^+\}$,
which is the set of all
measurement outcomes of $B$ appearing in our world.
Let $l_0$ be an arbitrary element of $\Lambda$, and let $C:=\{(m,l_0)\mid m\in\Omega\}$.
Since $\omega$ is Martin-L\"of $P$-random,
it follows from Theorem~\ref{nonzero-prob-appears} that $P(C)>0$.
Let $\delta:=\cond{C}{\omega}$.
Recall
that $\cond{C}{\omega}$ is defined as an infinite sequence in $C^\infty$ obtained from
$\omega$ by eliminating all elements of $(\Omega\times\Theta)\setminus C$ occurring in $\omega$.
In other words, $\delta$ is the subsequence of $\omega$ such that
the outcomes of measurements of $B$ equal to $l_0$.
We assume that for each $k\in\N^+$,
the $k$th element of $\delta$ is originally the $n_k$th element of $\omega$ before the elimination.
It follows that $\delta(k)=(\alpha(n_k),l_0)$ for every $k\in\N^+$.
Since $\omega$ is Martin-L\"of $P$-random,
using Theorem~\ref{conditional_probability} we have that
$\delta$ is Martin-L\"of $P_C$-random for the finite probability space $P_C$ on $C$.

Let $\zeta$ be an infinite sequence over $\Omega$ such that $\zeta(k):=\alpha(n_k)$ for every $k\in\N$.
The
sequence $\zeta$ can be phrased as \emph{the infinite sequence of outcomes of the
measurements of the observable $A$,
resulting from the infinitely repeated measurements of the observable $B$
following the measurement of the observable $A$,
over the infinite copies of $\mathcal{S}$ in our world,
conditioned by the specific measurement outcome $l_0$ of $B$}.
Since $\delta$ is Martin-L\"of $P_C$-random,
it is easy to show that $\zeta$ is Martin-L\"of $P_{l_0}$-random,
where $P_{l_0}$ is a finite probability space on $\Omega$ such that
$P_{l_0}(m)=P_C(m,l_0)$ for every $m\in\Omega$.
We see that
\begin{equation}\label{Pl0m=tr}
  P_{l_0}(m)=P_C(m,l_0)=\frac{P(m,l_0)}{\sum_{m'\in\Omega}P(m',l_0)}
  =\frac{\bra{\Psi}E_mF_{l_0}E_m\ket{\Psi}}{\tr(F_{l_0}\rho_A)}
\end{equation}
for each $m\in\Omega$, where the last equality follows from \eqref{Pm=ipMS} and \eqref{Pm=ipMSR}.
Note that the denominator $\tr(F_{l_0}\rho_A)$ on the most right-hand side of \eqref{Pl0m=tr} is certainly
non-zero due to \eqref{non-zero-inner-productMSR}.

Let $\mu$ be a subsequence sequence of $\gamma_B$ such that
$\mu(k):=\gamma_B(n_k)$ for every $k\in\N^+$.
Then, it follows from \eqref{mrMSB} that
\begin{equation}\label{mu=k=FEP}
  \mu(k)=\ket{\Upsilon^B_{n_k}}
  =\frac{F_{l_0}E_{\zeta(k)}\ket{\Psi}}{\sqrt{\bra{\Psi}E_{\zeta(k)}F_{l_0}E_{\zeta(k)}\ket{\Psi}}}
\end{equation}
for every $k\in\N^+$.
Let $\Xi^B:=\{\mu(k)\mid k\in\N^+\}$.
Note that $\Xi^B$ is an alphabet.
For each $\ket{\psi}\in\Xi^B$,
we define $T(\ket{\psi})$ as the set
$\{\zeta(k)\mid\ket{\psi}=\mu(k)\}$.
Since $\zeta$ is Martin-L\"of $P_{l_0}$-random,
it follows from Theorem~\ref{nonzero-prob-appears} and Theorem~\ref{contraction}
that the infinite sequence $\mu$ is Martin-L\"of $P_{l_0}'$-random,
where $P_{l_0}'$ is a finite probability space on $\Xi^B$ such that
\begin{equation}\label{Pl0'mk=Pl0zk2fb}
  P_{l_0}'(\ket{\psi})=\sum_{m\in T(\ket{\psi})}P_{l_0}(m)
\end{equation}
for every $\ket{\psi}\in\Xi^B$.
Thus, according to (i) of Definition~\ref{def-mixed-state},
the infinite sequence $\mu$ is a mixed state of $\mathcal{S}$.
The mixed state $\mu$
can be
phrased
as
\emph{the infinite sequence of states of $\mathcal{S}$
resulting from the infinitely repeated measurements of the observable $B$
over the infinite copies of $\mathcal{S}$, conditioned by the specific measurement outcome $l_0$ of $B$}.

Then, according to (ii) of Definition~\ref{def-mixed-state},
the density matrix $\rho_{l_0}$ of $\mu$ is given by
\begin{equation*}
  \rho_{l_0}:=\sum_{\ket{\psi}\in\Xi^B}P_{l_0}'(\ket{\psi})\ket{\psi}\bra{\psi}.
\end{equation*}
Thus, since $\zeta$ is Martin-L\"of $P_{l_0}$-random,
using Theorem~\ref{nonzero-prob-appears}, \eqref{mu=k=FEP}, and \eqref{Pl0'mk=Pl0zk2fb}
we have that
\begin{equation*}
  \rho_{l_0}
  =\sum_{m:P_{l_0}(m)>0}
  P_{l_0}(m)\frac{F_{l_0}E_m\ket{\Psi}\bra{\Psi}E_mF_{l_0}}{\bra{\Psi}E_{m}F_{l_0}E_{m}\ket{\Psi}},
\end{equation*}
where the sum is over all $m\in\Omega$ such that $P_{l_0}(m)>0$.
It follows from \eqref{Pl0m=tr} that
\begin{equation*}
  \rho_{l_0}=\sum_{m:P_{l_0}(m)>0}\frac{F_{l_0}E_m\ket{\Psi}\bra{\Psi}E_mF_{l_0}}{\tr(F_{l_0}\rho_A)}.
\end{equation*}
Note from \eqref{Pl0m=tr} that, for every $m\in\Omega$,
$P_{l_0}(m)=0$ if and only if $F_{l_0}E_m\ket{\Psi}=0$.
Thus, using \eqref{rhoA} we finally have that
\begin{equation*}
  \rho_{l_0}
  =\sum_{m\in\Omega}\frac{F_{l_0}E_m\ket{\Psi}\bra{\Psi}E_mF_{l_0}}{\tr(F_{l_0}\rho_A)}
  =\frac{F_{l_0}\rho_A F_{l_0}}{\tr(F_{l_0}\rho_A)}.
\end{equation*}

Recall
that $l_0$ is an arbitrary element of $\Lambda$.
Hence, in summary, we see that for every $l\in\Lambda$,
\begin{equation}\label{post-measuremnt-stateD1}
  \rho_{l}:=\frac{F_{l}\rho_A F_{l}}{\tr(F_{l}\rho_A)}
\end{equation}
is the density matrix of
\emph{the mixed state resulting from the infinitely repeated measurements of the observable $B$
following the measurement of the observable $A$,
over the infinite copies of $\mathcal{S}$ in our world,
conditioned by the specific measurement outcome $l$ of $B$}.
The result is just as expected from the aspect of the conventional quantum mechanics, i.e.,
Postulate~\ref{Born-rule2}.

\subsection{Derivation of Postulate~\ref{Tadaki-rule2}}

At last,
we show that Postulate~\ref{Tadaki-rule2} can be derived from Postulate~\ref{POT}
together with Postulates~\ref{state_space}, \ref{composition}, and \ref{evolution}
in the setting developed
in the preceding subsections.

For deriving Postulate~\ref{Tadaki-rule2},
we consider the quantum measurements described by the observable $B$
introduced in Subsection~\ref{ROM}
in
the above setting,
as the quantum measurements described by the observable $M$ in Postulate~\ref{Tadaki-rule2}.
Recall that the spectrum of $B$ is $\Theta$.
Then the spectrum decomposition \eqref{spectral-decompositionB} of $B$ is rewritten as
$$B=\sum_{l\in\Theta} l F_l,$$
in the same form as
stated
for $M$ in Postulate~\ref{Tadaki-rule2}.

Suppose that the measurements of the observable $B$ are
repeatedly performed over a mixed state with a density matrix $\rho$,
as assumed in Postulate~\ref{Tadaki-rule2}.
In the setting developed in the preceding subsections,
this mixed state is the infinite sequence $\gamma_A$ over $\Gamma^A$
and its density matrix $\rho$ is the density matrix $\rho_A$ of the mixed state $\gamma_A$,
according to the arguments in Subsection~\ref{MSRFMOA}.
The infinite sequence of outcomes generated by the measurements,
which is mentioned in Postulate~\ref{Tadaki-rule2},
is the infinite sequence $\beta$ over $\Theta$ in our world $\omega$,
according to the arguments in Subsection~\ref{APOT}.
Based on this identification, we can see the following:

First, since $\beta$ is an infinite sequence over $\Theta$,
the set of possible outcomes of the measurement is the spectrum $\Theta$ of $B$,
as stated in Postulate~\ref{Tadaki-rule2} as one of its conclusions.

Secondly, recall that $\beta$ is Martin-L\"of $R$-random,
where $R$ is the finite probability space on $\Theta$ satisfying \eqref{Pm=ipMSR},
according to the arguments in Subsection~\ref{ISBMOOB}.
Thus, since $\rho_A=\rho$, we see that \emph{in our world},
the infinite sequence of outcomes generated by the measurements,
which is mentioned in Postulate~\ref{Tadaki-rule2},
is a Martin-L\"of $V$-random infinite sequence over $\Theta$,
where $V$ is a finite probability space on $\Theta$ such that
$$V(l)=\tr(F_l\rho)$$
for every $l\in\Theta$,
as stated in Postulate~\ref{Tadaki-rule2} as one of its conclusions.

Thirdly, according to the arguments in Subsection~\ref{MSCSOMB},
the set of all outcomes of the measurements of $B$ in our world is $\Lambda$,
and for every $l\in\Lambda$,
the density matrix of the mixed state
resulting from the infinitely repeated measurements of $B$
in our world,
conditioned by the specific measurement outcome $l$ of $B$,
is $\rho_l$ given by \eqref{post-measuremnt-stateD1}.
Since $\rho_A=\rho$,
this means that
\emph{in our world},
the resulting
sequence of pure states with outcome $l$ is
a mixed state with the density matrix
\begin{equation*}
  \frac{F_l\rho F_l}{\tr(F_l\rho)},
\end{equation*}
as stated in the last sentence of Postulate~\ref{Tadaki-rule2} as one of its conclusions.

Hence, we have derived Postulate~\ref{Tadaki-rule2} from Postulate~\ref{POT},
the principle of typicality,
together with Postulates~\ref{state_space}, \ref{composition}, and \ref{evolution}
in the setting developed so far in the preceding subsections,
i.e., in \emph{the simplest scenario}
where a mixed state being measured is
an
infinite sequence over \emph{mutually orthogonal} pure states.

\section{Derivation II of Postulate~\ref{Tadaki-rule2} from the principle of typicality}
\label{POT-mixed-states2}

According to (i) of Definition~\ref{def-mixed-state},
a mixed state is commonly an infinite sequence over mutually \emph{non}-orthogonal pure states.
Thus, the setting investigated in Section~\ref{POT-mixed-states} deals with a special case regarding
the form of a mixed state being measured.
This setting is the simplest one.
We can make the mixed state being measured,
which is denoted by $\gamma_A$ in Section~\ref{POT-mixed-states},
into a \emph{general form}
by introducing an ancilla system $\mathcal{S}_a$ in addition to the original system $\mathcal{S}$,
preparing the composite system of $\mathcal{S}$ and $\ancilla$ in an appropriate entangled state,
and then performing the measurement of $A$ over $\ancilla$ instead of $\mathcal{S}$.
We can derive Postulate~\ref{POT} in this generalized setting,
as well.
In this section, we do this.

First, we explain this idea \emph{in the terminology of the conventional quantum mechanics}.
The argument here is based on Nielsen and Chuang \cite[Exercise 2.82]{NC00}.
Let $\ket{\Psi_1},\dots,\ket{\Psi_n}$ be arbitrary states of a system $\mathcal{S}$,
and let $p_1,\dots,p_n$ be
non-negative
reals with $p_1+\dots +p_n=1$.
We introduce
an ancilla
system $\ancilla$ whose orthonormal basis is $\ket{b_1},\dots,\ket{b_n}$.
Then, consider a state
$$\ket{\Psi_{\mathrm{comp}}}:=\sum_{i=1}^n\sqrt{p_i}\ket{b_i}\otimes\ket{\Psi_i}$$
of the composite system consisting of $\ancilla$ and $\mathcal{S}$.
If we perform a measurement in the basis $\{\ket{b_i}\}$ over the system $\ancilla$,
in the state $\ket{\Psi_{\mathrm{comp}}}$ of the composite system,
then the resulting state of the system $\mathcal{S}$ is a ``mixed state''
where each state $\ket{\Psi_i}$ appears with ``probability'' $p_i$,
and its density matrix $\rho$ is given by
$$\rho:=\sum_{i=1}^n p_i\ket{\Psi_i}\bra{\Psi_i}.$$
Thus, in this manner, we can ``prepare'' a ``mixed state''
in a general form.

In a similar manner, we can prepare an arbitrary mixed state which is an infinite sequence
over \emph{not necessarily mutually orthogonal} pure states,
in our rigorous setting based on Definition~\ref{def-mixed-state}.
Thus, in this section we derive Postulate~\ref{Tadaki-rule2} from Postulate~\ref{POT}
together with Postulates~\ref{state_space}, \ref{composition}, and \ref{evolution}
in this generalized scenario where the mixed state being measured is an infinite sequence over
not necessarily mutually orthogonal pure states.
For that purpose, we consider more complicated interaction between systems and
apparatuses than one used for deriving Postulate~\ref{Tadaki-rule2} in the preceding section.
In the setting, we also have to perform measurements on the mixed state
while generating it by other measurements,
in a similar manner to the preceding section.

Thus, for deriving Postulate~\ref{Tadaki-rule2},
we consider an infinite repetition of two successive measurements of observables $A$ and $B$
where the measurements of $A$ are performed over each of countably infinite copies of
an ancilla system $\ancilla$ and
the measurements of $B$ are performed over each of countably infinite copies of
a system $\mathcal{S}$.
Each of the countably infinite copies of the composite system consisting of $\mathcal{S}$ and $\ancilla$
is prepared in an identical (entangled) initial state $\ket{\Psi_\mathrm{comp}}$.
In the setting, while generating a mixed state by the measurements of $A$,
we are performing the measurements of $B$ over the mixed state,
as we did in the
setting
of the preceding section.

\subsection{Repeated once of measurements}
\label{d2ROM}

First, the repeated once of the infinite repetition of measurements, i.e.,
the two successive measurements of observables $A$ and $B$,
is described as follows.

Let $\mathcal{S}$ be an arbitrary quantum system with state space $\mathcal{H}$ of finite dimension $L$,
and let $\ancilla$ be another arbitrary quantum system with state space $\mathcal{H}^{\mathrm{a}}$
of finite dimension $K$.
Consider arbitrary two measurements over $\ancilla$ and $\mathcal{S}$
described by observables $A$ and $B$, respectively.
Suppose that the observable $A$ is \emph{non-degenerate}, i.e.,
every eigenspace of $A$ has a dimension of one.
Let $\Omega$ and $\Theta$ be the spectrums of $A$ and $B$, respectively.
Let
\begin{equation*}
  A=\sum_{k=1}^K f(k)\ket{\psi_k}\bra{\psi_k}
\end{equation*}
be a spectral decomposition of the observable $A$,
where $\{\ket{\psi_1},\dots,\ket{\psi_K}\}$ is an orthonormal basis of $\mathcal{H}^{\mathrm{a}}$ and
$f\colon\{1,\dots,K\}\to\Omega$ is a \emph{bijection},
and let
\begin{equation}\label{d2spectral-decompositionB}
  B=\sum_{\ell=1}^L g(\ell)\ket{\phi_{\ell}}\bra{\phi_{\ell}}
\end{equation}
be a spectral decomposition of the observable $B$,
where $\{\ket{\phi_1},\dots,\ket{\phi_L}\}$ is an orthonormal basis of $\mathcal{H}$ and
$g\colon\{1,\dots,L\}\to\Theta$ is a surjection.
According to Postulates~\ref{state_space}, \ref{composition}, and \ref{evolution},
the measurement processes of the observables $A$ and $B$ are described
by the following unitary operators $U_A$ and $U_B$, respectively:
\begin{equation} \label{d2single_measurementA}
  U_A\ket{\psi_k}\otimes\ket{\Phi_A^{\mathrm{init}}}=\ket{\psi_k}\otimes\ket{\Phi_A[f(k)]}
\end{equation}
for every $k=1,\dots,K$,
and
\begin{equation}  \label{d2single_measurementB}
  U_B\ket{\phi_{\ell}}\otimes\ket{\Phi_B^{\mathrm{init}}}=\ket{\phi_{\ell}}\otimes\ket{\Phi_B[g(\ell)]}
\end{equation}
for every $\ell=1,\dots,L$.
The vectors $\ket{\Phi_A^{\mathrm{init}}}$ and $\ket{\Phi_B^{\mathrm{init}}}$ are
the initial states of the apparatuses measuring $A$ and $B$, respectively, and
$\ket{\Phi_A[m]}$ and $\ket{\Phi_B[l]}$ are
the final states of
ones
for each $m\in\Omega$ and $l\in\Theta$,
with $\braket{\Phi_A[m]}{\Phi_A[m']}=\delta_{m,m'}$ and $\braket{\Phi_B[l]}{\Phi_B[l']}=\delta_{l,l'}$.
For every $m\in\Omega$, the state $\ket{\Phi_A[m]}$ indicates that
\emph{the apparatus measuring the observable $A$ of the system $\ancilla$
records the value $m$ of $A$}, and
for every $l\in\Theta$, the state $\ket{\Phi_B[l]}$ indicates that
\emph{the apparatus measuring the observable $B$ of the system $\mathcal{S}$
records the value $l$ of $B$}.

For each $m\in\Omega$, let $E_m$ be the projector onto the eigenspace of $A$ with eigenvalue $m$,
and for each $l\in\Theta$, let $F_l$ be the projector onto the eigenspace of $B$ with eigenvalue $l$.
Then, the equalities~\eqref{d2single_measurementA} and \eqref{d2single_measurementB} are rewritten,
respectively, as the forms that
\begin{equation}\label{d2single_measurementA2}
  U_A\ket{\Psi'}\otimes\ket{\Phi_A^{\mathrm{init}}}
  =\sum_{m\in\Omega}(E_m\ket{\Psi'})\otimes\ket{\Phi_A[m]}
\end{equation}
for every $\ket{\Psi'}\in\mathcal{H}^{\mathrm{a}}$, and
\begin{equation}\label{d2single_measurementB2}
  U_B\ket{\Psi}\otimes\ket{\Phi_B^{\mathrm{init}}}
  =\sum_{l\in\Theta}(F_l\ket{\Psi})\otimes\ket{\Phi_B[l]}
\end{equation}
for every $\ket{\Psi}\in\mathcal{H}$.

\subsection{\boldmath Infinite repetition of the measurements of the observables $A$ and $B$}

As in the preceding section,
we consider countably infinite copies of the system $\mathcal{S}$.
Moreover, we consider countably infinite copies of the ancilla system $\ancilla$.
We prepare each of
the
copies of the composite system of $\ancilla$ and $\mathcal{S}$
in an identical state,
and then perform the successive measurements of the observables $A$ and $B$ over
each of the copies of the composite system of $\ancilla$ and $\mathcal{S}$ \emph{one by one},
by interacting each of the copies of the composite system with
each of the copies of the apparatuses measuring $A$ and $B$
according to
the unitary time-evolution \eqref{d2single_measurementA2} and then
the unitary time-evolution \eqref{d2single_measurementB2}.

Let $\{\ket{\Psi_m}\}_{m\in\Omega}$ be
a collection of $K$ \emph{arbitrary} states of the system $\mathcal{S}$,
and let $\{p_m\}_{m\in\Omega}$ be a collection of $K$ \emph{arbitrary}
non-negative
reals such that
$$\sum_{m\in\Omega}p_m=1.$$
Note that $\#\Omega=K$, since the function $f\colon\{1,\dots,K\}\to\Omega$ is a bijection.
In what follows, based on $\{\ket{\Psi_m}\}_{m\in\Omega}$ and $\{p_m\}_{m\in\Omega}$
we consider a mixed state with the density matrix
$$\sum_{m\in\Omega} p_m\ket{\Psi_m}\bra{\Psi_m}$$
as an \emph{arbitrary mixed state of the system $\mathcal{S}$ being measured}.
Thus, without loss of generality, we can assume that
the collection $\{\ket{\Psi_m}\}_{m\in\Omega}$ is \emph{pair-wise linearly independent}, i.e.,
it satisfies that
$\ket{\Psi_m}$ and $\ket{\Psi_{m'}}$ are linearly independent
for every $m,m'\in\Omega$ with $m\neq m'$.
Furthermore, we choose a \emph{specific} collection $\{\ket{\theta_m}\}_{m\in\Omega}$ of $K$ states of
the ancilla system $\ancilla$ such that
each vector $\ket{\theta_m}$ is a unit vector in the eigenspace of $A$ with eigenvalue $m$.%
\footnote{%
Equivalently,
the collection  $\{\ket{\theta_m}\}_{m\in\Omega}$
has to
satisfy the condition that
for every $m\in\Omega$ there exists $\delta\in\R$ such that
$\ket{\theta_m}=e^{i\delta}\ket{\psi_{f^{-1}(m)}}$.}
We then consider a state $\ket{\Psi_{\mathrm{comp}}}$ of
the composite system of $\ancilla$ and $\mathcal{S}$ defined by
\begin{equation*}
  \ket{\Psi_{\mathrm{comp}}}:=\sum_{m\in\Omega}\sqrt{p_m}\ket{\theta_m}\otimes\ket{\Psi_m}.
\end{equation*}
Thus, we prepare each of the infinite copies of the composite system of $\ancilla$ and $\mathcal{S}$
in the identical state $\ket{\Psi_{\mathrm{comp}}}$
immediately before the measurement of the observable $A$ of the system $\ancilla$.

For each $n\in\N^+$,
let $\mathcal{H}_n$ be the state space of the total system consisting of
the first $n$ copies of the two systems $\ancilla$ and $\mathcal{S}$ and
the two apparatuses measuring $A$ and $B$.
The
successive interactions
between the copies of the two systems $\ancilla$ and $\mathcal{S}$
and the copies of the two apparatuses measuring $A$ and $B$,
as measurement processes,
proceed in the following manner.

The
starting
state of the total system,
which consists of the first copy of the two systems $\ancilla$ and $\mathcal{S}$ and
the two apparatuses measuring $A$ and $B$,
is $\ket{\Psi_{\mathrm{comp}}}\otimes\ket{\Phi_A^{\mathrm{init}}}\otimes\ket{\Phi_B^{\mathrm{init}}}\in\mathcal{H}_1$.
In this state of the total system,
the measurement of $A$ is performed over the first copy of the ancilla system $\ancilla$.
Immediately after
that,
the total system results in the state
\begin{align*}
  \sum_{m_1\in\Omega} \sqrt{p_{m_1}}\ket{\theta_{m_1}}\otimes\ket{\Psi_{m_1}}\otimes
  \ket{\Phi_A[m_1]}\otimes\ket{\Phi_B^{\mathrm{init}}} \in\mathcal{H}_1
\end{align*}
by the interaction~\eqref{d2single_measurementA2} as the measurement process of $A$.
Then, the measurement of $B$ is performed over the first copy of the system $\mathcal{S}$.
Immediately after that,
the total system results in the state
\begin{align*}
  \sum_{m_1\in\Omega}\sum_{l_1\in\Theta}
  \sqrt{p_{m_1}}\ket{\theta_{m_1}}\otimes(F_{l_1}\ket{\Psi_{m_1}})\otimes
  \ket{\Phi_A[m_1]}\otimes\ket{\Phi_B[l_1]} \in\mathcal{H}_1
\end{align*}
by the interaction~\eqref{d2single_measurementB2} as the measurement process of $B$.

In general,
immediately before the measurement of $A$ over the $n$th copy of
the ancilla system $\ancilla$,
the state of the total system,
which consists of the first $n$ copies of the two systems $\ancilla$ and $\mathcal{S}$ and
the two apparatuses measuring $A$ and $B$, is
\begin{align*}
  \sum_{m_1,\dots,m_{n-1}\in\Omega}&\sum_{l_1,\dots,l_{n-1}\in\Theta}
  \sqrt{p_{m_1}\dots p_{m_{n-1}}}
  \ket{\theta_{m_1}}\otimes(F_{l_1}\ket{\Psi_{m_1}})\otimes\dots\otimes
  \ket{\theta_{m_{n-1}}}\otimes(F_{l_{n-1}}\ket{\Psi_{m_{n-1}}}) \\
  &\otimes\ket{\Psi_{\mathrm{comp}}} \\
  &\otimes\ket{\Phi_A[m_1]}\otimes\ket{\Phi_B[l_1]}\otimes\dots
  \otimes\ket{\Phi_A[m_{n-1}]}\otimes\ket{\Phi_B[l_{n-1}]}
  \otimes\ket{\Phi_A^{\mathrm{init}}}\otimes\ket{\Phi_B^{\mathrm{init}}}
\end{align*}
in $\mathcal{H}_n$.
Then, immediately after the measurement of $A$ over the $n$th copy of the system $\ancilla$,
the total system results in the state
\begin{equation}\label{d2total_systemA}
\begin{split}
  \sum_{m_1,\dots,m_{n}\in\Omega}&\sum_{l_1,\dots,l_{n-1}\in\Theta}
  \sqrt{p_{m_1}\dots p_{m_n}}
  \ket{\theta_{m_1}}\otimes(F_{l_1}\ket{\Psi_{m_1}})\otimes\dots\otimes
  \ket{\theta_{m_{n-1}}}\otimes(F_{l_{n-1}}\ket{\Psi_{m_{n-1}}}) \\
  &\otimes\ket{\theta_{m_n}}\otimes\ket{\Psi_{m_n}} \\
  &\otimes\ket{\Phi_A[m_1]}\otimes\ket{\Phi_B[l_1]}\otimes\dots
  \otimes\ket{\Phi_A[m_{n-1}]}\otimes\ket{\Phi_B[l_{n-1}]}
  \otimes\ket{\Phi_A[m_{n}]}\otimes\ket{\Phi_B^{\mathrm{init}}}
\end{split}
\end{equation}
in $\mathcal{H}_n$ by the interaction \eqref{d2single_measurementA2},
as the measurement process of $A$,
between
the $n$th copy of the system $\ancilla$
and the $n$th copy of the apparatus measuring $A$ in the state $\ket{\Phi_A^{\mathrm{init}}}$.
Then,
immediately after the measurement of $B$ over the $n$th copy of the system $\mathcal{S}$,
the total system results in the state
\begin{align}
  &\sum_{m_1,\dots,m_{n}\in\Omega}\sum_{l_1,\dots,l_{n}\in\Theta}
  \sqrt{p_{m_1}\dots p_{m_n}}
  \ket{\theta_{m_1}}\otimes(F_{l_1}\ket{\Psi_{m_1}})\otimes\dots\otimes
  \ket{\theta_{m_n}}\otimes(F_{l_n}\ket{\Psi_{m_n}}) \nonumber \\
  &\hspace*{32mm}\otimes\ket{\Phi_A[m_1]}\otimes\ket{\Phi_B[l_1]}\otimes\dots
  \otimes\ket{\Phi_A[m_{n}]}\otimes\ket{\Phi_B[l_{n}]} \label{d2total_systemAB0} \\
  & \nonumber \\
  &=\sum_{m_1,\dots,m_{n}\in\Omega}\sum_{l_1,\dots,l_{n}\in\Theta}
  \sqrt{p_{m_1}\dots p_{m_n}}
  \ket{\theta_{m_1}}\otimes(F_{l_1}\ket{\Psi_{m_1}})\otimes\dots\otimes
  \ket{\theta_{m_n}}\otimes(F_{l_n}\ket{\Psi_{m_n}}) \nonumber \\
  &\hspace*{32mm}\otimes\ket{\Phi[(m_1,l_1)\dots (m_{n},l_{n})]} \label{2dtotal_systemAB}
\end{align}
in $\mathcal{H}_n$ by the interaction \eqref{d2single_measurementB2},
as the measurement process of $B$,
between the $n$th copy of the system $\mathcal{S}$
and the $n$th copy of the apparatus measuring $B$ in the state $\ket{\Phi_B^{\mathrm{init}}}$.
Here
the state $\ket{\Phi[(m_1,l_1)\dots (m_n,l_n)]}$ denotes
$\ket{\Phi_A[m_1]}\otimes\ket{\Phi_B[l_1]}\otimes\dots
\otimes\ket{\Phi_A[m_{n}]}\otimes\ket{\Phi_B[l_{n}]}$, and indicates that
\emph{the first $n$ copies of the two apparatuses measuring $A$ and $B$ record
the values $(m_1,l_1)\dots (m_n,l_n)$ of the observables $A$ and $B$, in this order,
of the first $n$ copies of the two systems $\ancilla$ and $\mathcal{S}$}.

Then, for applying Postulate~\ref{POT}, the principle of typicality,
we have to identify worlds and 
the
measure representation for the prefixes of
worlds
in this situation.
Note, from \eqref{d2single_measurementA2} and \eqref{d2single_measurementB2},
that the sequential application of $U_A$ and $U_B$ to the composite system
consisting of the two systems $\ancilla$ and $\mathcal{S}$
and the two apparatuses measuring $A$ and $B$ respectively over the initial state
$\ket{\Psi_{\mathrm{comp}}}\otimes\ket{\Phi_A^{\mathrm{init}}}\otimes\ket{\Phi_B^{\mathrm{init}}}$
results in the following single unitary time-evolution $U$:
\begin{align*}
  U\ket{\Psi_{\mathrm{comp}}}\otimes\ket{\Phi_A^{\mathrm{init}}}\otimes\ket{\Phi_B^{\mathrm{init}}}
  &=\sum_{m\in\Omega}\sum_{l\in\Theta}
  (E_{m}\otimes F_{l})\ket{\Psi_{\mathrm{comp}}}\otimes\ket{\Phi_A[m]}\otimes\ket{\Phi_B[l]} \\
  &=\sum_{m\in\Omega}\sum_{l\in\Theta}
  (E_{m}\otimes F_{l})\ket{\Psi_{\mathrm{comp}}}\otimes\ket{\Phi[(m,l)]},
\end{align*}
where $U=U_A\otimes U_B$.
It is then easy to check that
a collection $\{E_m\otimes F_l\}_{(m,l)\in\Omega\times\Theta}$ forms measurement operators.
Thus, the successive measurements of $A$ and $B$ in this order can be regarded
as the \emph{single measurement} which is described by the measurement operators
$\{E_m\otimes F_l\}_{(m,l)\in\Omega\times\Theta}$.
Hence, we can apply Definition~\ref{pmrpwst}
to this scenario of the setting of measurements.
Therefore, according to Definition~\ref{pmrpwst},
we see that a \emph{world} is an infinite sequence over $\Omega\times\Theta$ and
the
\emph{measure representation for the prefixes of
worlds}
is given by a function $r\colon(\Omega\times\Theta)^*\to[0,1]$ with
$$r((m_1,l_1)\dotsc (m_n,l_n)):=\prod_{k=1}^n p_{m_k}\bra{\Psi_{m_k}}F_{l_k}\ket{\Psi_{m_k}},$$
which is the square of the norm of each state
$$\sqrt{p_{m_1}\dots p_{m_n}}
\ket{\theta_{m_1}}\otimes(F_{l_1}\ket{\Psi_{m_1}})\otimes\dots\otimes
\ket{\theta_{m_n}}\otimes(F_{l_n}\ket{\Psi_{m_n}})\otimes\ket{\Phi[(m_1,l_1)\dots (m_{n},l_{n})]}$$
in the superposition~\eqref{2dtotal_systemAB}.
It follows from \eqref{mr} that
the probability measure induced by the probability measure representation $r$ is
a Bernoulli measure $\lambda_P$ on $(\Omega\times\Theta)^\infty$,
where $P$ is a finite probability space on $\Omega\times\Theta$ such that
\begin{equation}\label{2dPm=ipMS}
  P(m,l)=p_m\bra{\Psi_m}F_l\ket{\Psi_m}
\end{equation}
for every $m\in\Omega$ and $l\in\Theta$.
The finite records $(m_1,l_1)\dots (m_n,l_n)\in(\Omega\times\Theta)^*$
in each state $\ket{\Phi[(m_1,l_1)\dots (m_n,l_n)]}$
in the superposition~\eqref{2dtotal_systemAB} of the total system is a \emph{prefix} of a world.

\subsection{Application of the principle of typicality}
\label{d2APOT}

Now, let us
apply Postulate~\ref{POT} to the setting developed above.

Let $\omega$ be our world in the infinite repetition of the measurements of $A$ and $B$
in the above setting.
Then $\omega$ is an infinite sequence over $\Omega\times\Theta$.
Since the Bernoulli measure $\lambda_P$ on $(\Omega\times\Theta)^\infty$ is
the probability measure induced by
the
measure representation $r$
for the prefixes of
worlds
in the above setting,
it follows from Postulate~\ref{POT} that \emph{$\omega$ is Martin-L\"of random
with respect to the measure $\lambda_P$ on $(\Omega\times\Theta)^\infty$}.
Therefore, \emph{$\omega$ is Martin-L\"of $P$-random,
where $P$ is the finite probability space on $\Omega\times\Theta$ satisfying \eqref{2dPm=ipMS}}.

Let $\alpha$ and $\beta$ be infinite sequences over $\Omega$ and $\Theta$, respectively,
such that $(\alpha(n),\beta(n))=\omega(n)$ for every $n\in\N^+$.
The sequence $\alpha$ is
\emph{the infinite sequence of records of the values of the observable $A$
in the corresponding apparatuses measuring $A$ in our world}.
A prefix of $\alpha$ is the finite records $m_1\dots m_n$ in a specific state
$\ket{\Phi_A[m_1]}\otimes\dots\otimes\ket{\Phi_A[m_n]}$ of the
apparatuses,
each of which has measured the observable $A$ of the corresponding copy of $\ancilla$,
in the superposition~\eqref{2dtotal_systemAB} of the total system.
Put briefly, the infinite sequence $\alpha$ is \emph{the infinite sequence
of outcomes of the infinitely repeated measurements of the observable $A$
over the infinite copies of $\ancilla$ in our world}.
On the other hand, the sequence $\beta$ is
\emph{the infinite sequence of records of the values of the observable $B$
in the corresponding apparatuses measuring $B$ in our world}.
A prefix of $\beta$ is the finite records $l_1\dots l_n$ in a specific state
$\ket{\Phi_B[l_1]}\otimes\dots\otimes\ket{\Phi_B[l_n]}$ of the
apparatuses,
each of which has measured the observable $B$ of the corresponding copy of $\mathcal{S}$,
in the superposition~\eqref{2dtotal_systemAB} of the total system.
Put briefly, the infinite sequence $\beta$ is \emph{the infinite sequence
of outcomes of the infinitely repeated measurements of the observable $B$
over the infinite copies of $\mathcal{S}$ in our world}.

Using Theorem~\ref{nonzero-prob-appears}, we note that $P(\omega(n))>0$ for every $n\in\N^+$.
Thus, it follows from \eqref{2dPm=ipMS} that
\begin{equation}\label{2dnon-zero-inner-productMS}
  p_{\alpha(n)}\bra{\Psi_{\alpha(n)}}F_{\beta(n)}\ket{\Psi_{\alpha(n)}}>0
\end{equation}
for every $n\in\N^+$.

By Theorem~\ref{contraction2} we have that $\alpha$ is Martin-L\"of $Q$-random, where
$Q$ is a finite probability space on $\Omega$
such that $Q(m):=\sum_{l\in\Theta}P(m,l)$ for every $m\in\Omega$.
Since
$\sum_{l\in\Theta} F_l=I$ and $\ket{\Psi_m}$ is a unit vector,
it follows from \eqref{2dPm=ipMS} that
\begin{equation}\label{2dPm=ipMSQ}
  Q(m)=p_m
\end{equation}
for every $m\in\Omega$,
as expected from the point of view of the conventional quantum mechanics
presented at the beginning of this section.
Note from Theorem~\ref{nonzero-prob-appears} that $Q(\alpha(n))>0$ for every $n\in\N^+$.
Thus, it follows from \eqref{2dPm=ipMSQ} that
\begin{equation}\label{2dnon-zero-inner-productMSQ}
  p_{\alpha(n)}>0
\end{equation}
for every $n\in\N^+$.

\subsection{\boldmath Mixed state resulting from the measurements of the observable $A$}
\label{d2MSRFMOA}

We calculate the mixed state resulting from the measurements of the observable $A$.

Let $n$ be an arbitrary positive integer.
In the superposition~\eqref{d2total_systemA} of the total system consisting of
the first $n$ copies of the two systems $\ancilla$ and $\mathcal{S}$ and
the two apparatuses measuring $A$ and $B$
immediately after the measurement of $A$ over the $n$th copy of the system $\ancilla$,
consider the specific state
\begin{equation}\label{2dspecific_stateMS}
\begin{split}
  &\sqrt{p_{m_1}\dots p_{m_n}}
  \ket{\theta_{m_1}}\otimes(F_{l_1}\ket{\Psi_{m_1}})\otimes\dots\otimes
  \ket{\theta_{m_{n-1}}}\otimes(F_{l_{n-1}}\ket{\Psi_{m_{n-1}}}) \\
  &\otimes\ket{\theta_{m_n}}\otimes\ket{\Psi_{m_n}} \\
  &\otimes\ket{\Phi_A[m_1]}\otimes\ket{\Phi_B[l_1]}\otimes\dots
  \otimes\ket{\Phi_A[m_{n-1}]}\otimes\ket{\Phi_B[l_{n-1}]}
  \otimes\ket{\Phi_A[m_{n}]}\otimes\ket{\Phi_B^{\mathrm{init}}}
\end{split}
\end{equation}
such that $(m_{k},l_{k})=\omega(k)$ for every $k=1,\dots,n-1$ and $m_n=\alpha(n)$.
Due to \eqref{2dnon-zero-inner-productMS} and \eqref{2dnon-zero-inner-productMSQ},
we note here that the vector \eqref{2dspecific_stateMS} is non-zero and therefore can be
a state vector certainly, when it is normalized.
Since $\omega$ is our world, the state \eqref{2dspecific_stateMS} is
\emph{the state of the total system,
consisting of the first $n$ copies of the two systems $\ancilla$ and $\mathcal{S}$ and
the two apparatuses measuring $A$ and $B$,
that we have perceived after the $n$th measurement of $A$}.
Therefore, \emph{we perceive} $\ket{\Psi_{m_n}}$
as the state of the $n$th copy of the system $\mathcal{S}$ at this time. 
Note that $m_n=\alpha(n)$ and $n$ is arbitrary.
Thus, for every positive integer $n$, the state of the $n$th copy of the system $\mathcal{S}$
immediately after the measurement of $A$ over it \emph{in our world} is given by
\begin{equation}\label{2ddef-u-a}
  \ket{\Psi_{\alpha(n)}}.
\end{equation}

Let $\Gamma^A:=\{\ket{\Psi_{\alpha(n)}}\mid n\in\N^+\}$,
and let $\gamma_A$ be an infinite sequence over $\Gamma^A$ such that
$\gamma_A(n):=\ket{\Psi_{\alpha(n)}}$ for every $n\in\N^+$.
Note that $\Gamma^A$ is
an alphabet.
Recall that the collection $\{\ket{\Psi_m}\}_{m\in\Omega}$ is pair-wise linearly independent.
Therefore, the collection $\{\ket{\Psi_m}\}_{m\in\Omega}$ is \emph{pair-wise distinct}, i.e.,
$\ket{\Psi_m}\neq\ket{\Psi_{m'}}$ for every $m,m'\in\Omega$ with $m\neq m'$. 
Thus, since $\alpha$ is Martin-L\"of $Q$-random,
it follows from Theorem~\ref{nonzero-prob-appears}
that the infinite sequence $\gamma_A$ is Martin-L\"of $Q'$-random,
where $Q'$ is a finite probability space on $\Gamma^A$ such that
\begin{equation}\label{2dQ'U=Qa}
  Q'(\ket{\Psi_{\alpha(n)}})=Q(\alpha(n))
\end{equation}
for every $n\in\N^+$.
Thus, according to (i) of Definition~\ref{def-mixed-state},
the infinite sequence $\gamma_A$ is a mixed state of $\mathcal{S}$.
The mixed state $\gamma_A$ can be interpreted as
\emph{the infinite sequence of states of $\mathcal{S}$
resulting from the infinitely repeated measurements of the observable $A$
over the infinite copies of $\ancilla$}.
Note that the mixed state $\gamma_A$ is
an infinite sequence over pure states which are \emph{not necessarily mutually orthogonal}.

Then, according to (ii) of Definition~\ref{def-mixed-state},
the density matrix $\rho_A$ of $\gamma_A$ is given by
\begin{equation*}
  \rho_A:=\sum_{\ket{\psi}\in \Gamma^A}Q'(\ket{\psi})\ket{\psi}\bra{\psi}.
\end{equation*}
Note from \eqref{2dQ'U=Qa} and \eqref{2dPm=ipMSQ} that
$Q'(\ket{\Psi_{\alpha(n)}})=p_{\alpha(n)}$ for every $n\in\N^+$.
Thus, since the collection $\{\ket{\Psi_m}\}_{m\in\Omega}$ is pair-wise distinct,
it follows from Theorem~\ref{nonzero-prob-appears} that
$$\rho_A=\sum_{m:Q(m)>0}p_m\ket{\Psi_m}\bra{\Psi_m},$$
where the sum is over all $m\in\Omega$ such that $Q(m)>0$.
Hence, using \eqref{2dPm=ipMSQ} we have that
\begin{equation}\label{2drhoA}
  \rho_A=\sum_{m\in\Omega}p_m\ket{\Psi_m}\bra{\Psi_m},
\end{equation}
as expected from the point of view of the conventional quantum mechanics
presented at the beginning of this section.

\subsection{\boldmath The infinite sequence $\beta$ of measurement outcomes of the observable $B$}
\label{d2ISBMOOB}

We determine the property of
the infinite sequence $\beta$ of records of the values of the observable $B$
in the
corresponding
apparatuses measuring $B$ in our world, in terms of the observable $B$ and
the mixed state $\gamma_A$ resulting from the measurements of the observable $A$ in our world.

Using a theorem symmetrical to Theorem~\ref{contraction2}
we have that $\beta$ is Martin-L\"of $R$-random,
where $R$ is a finite probability space on $\Theta$ such that
$R(l):=\sum_{m\in\Omega}P(m,l)$ for every $l\in\Theta$.
It follows from \eqref{2dPm=ipMS} and \eqref{2drhoA} that
\begin{equation}\label{2dPm=ipMSR}
  R(l)=\sum_{m\in\Omega}p_m\bra{\Psi_m}F_l\ket{\Psi_m}=\tr(F_l\rho_A)
\end{equation}
for every $l\in\Theta$, as expected from the point of view of the conventional quantum mechanics,
i.e., Postulate~\ref{Born-rule2}.
Note from Theorem~\ref{nonzero-prob-appears} that $R(\beta(n))>0$ for every $n\in\N^+$.
Therefore, we have
that
\begin{equation}\label{2dnon-zero-inner-productMSR}
  \tr(F_{\beta(n)}\rho_A)>0
\end{equation}
for every $n\in\N^+$.

\subsection{\boldmath Mixed state resulting from the measurements of the observable $B$}

It is worthwhile to calculate the mixed state resulting from the measurements of the observable $B$.

Let $n$ be an arbitrary positive integer.
In the superposition~\eqref{d2total_systemAB0} of the total system consisting of
the first $n$ copies of the two systems $\ancilla$ and $\mathcal{S}$ and
the two apparatuses measuring $A$ and $B$
immediately after the measurement of $B$ over the $n$th copy of the system $\mathcal{S}$,
consider the specific state
\begin{equation}\label{2dspecific_stateMSAB}
\begin{split}
  &\sqrt{p_{m_1}\dots p_{m_n}}
  \ket{\theta_{m_1}}\otimes(F_{l_1}\ket{\Psi_{m_1}})\otimes\dots\otimes
  \ket{\theta_{m_n}}\otimes(F_{l_n}\ket{\Psi_{m_n}}) \\
  &\hspace*{16mm}\otimes\ket{\Phi_A[m_1]}\otimes\ket{\Phi_B[l_1]}\otimes\dots
  \otimes\ket{\Phi_A[m_{n}]}\otimes\ket{\Phi_B[l_{n}]}
\end{split}
\end{equation}
such that $(m_{k},l_{k})=\omega(k)$ for every $k=1,\dots,n$.
Due to \eqref{2dnon-zero-inner-productMS},
we note here that the vector \eqref{2dspecific_stateMSAB} is non-zero and therefore can be
a state vector certainly, when it is normalized.
Since $\omega$ is our world, the state \eqref{2dspecific_stateMSAB} is
\emph{the state of the total system, consisting of
the first $n$ copies of the two systems $\ancilla$ and $\mathcal{S}$
and the two apparatuses measuring $A$ and $B$,
that we have perceived after
the $n$th measurement of $B$}.
Therefore, \emph{we perceive
$$\frac{F_{l_n}\ket{\Psi_{m_n}}}{\sqrt{\bra{\Psi_{m_n}}F_{l_n}\ket{\Psi_{m_n}}}}$$
as the state of the $n$th copy of the system $\mathcal{S}$ at this time}. 
Note that $m_n=\alpha(n)$ and $l_n=\beta(n)$, and $n$ is arbitrary.
Thus, for every positive integer $n$, the state of the $n$th copy of the system $\mathcal{S}$
immediately after the measurement of $B$ over it \emph{in our world} is given by
\begin{equation}\label{2dmrMSB}
  \ket{\Upsilon^B_n}
  :=\frac{F_{\beta(n)}\ket{\Psi_{\alpha(n)}}}{\sqrt{\bra{\Psi_{\alpha(n)}}F_{\beta(n)}\ket{\Psi_{\alpha(n)}}}}.
\end{equation}

Let $\Gamma^B:=\{\ket{\Upsilon^B_n}\mid n\in\N^+\}$,
and let $\gamma_B$ be an infinite sequence over $\Gamma^B$ such that
$\gamma_B(n):=\ket{\Upsilon^B_n}$ for every $n\in\N^+$.
Note that $\Gamma^B$ is an alphabet.
For each $\ket{\psi}\in\Gamma^B$, we define $S(\ket{\psi})$ as the set
$$\{(\alpha(n),\beta(n))\mid\ket{\psi}=\ket{\Upsilon^B_n}\}.$$
Since $\omega$ is Martin-L\"of $P$-random,
it follows from Theorem~\ref{nonzero-prob-appears} and Theorem~\ref{contraction}
that the infinite sequence $\gamma_B$ is Martin-L\"of $P'$-random,
where $P'$ is a finite probability space on $\Gamma^B$ such that
\begin{equation}\label{2dP'U=Pa}
  P'(\ket{\psi})=\sum_{x\in S(\ket{\psi})} P(x)
\end{equation}
for every $\ket{\psi}\in\Gamma^B$.
Thus, according to (i) of Definition~\ref{def-mixed-state},
the infinite sequence $\gamma_B$ is a mixed state of $\mathcal{S}$.
The mixed state $\gamma_B$ can be
interpreted
as
\emph{the infinite sequence of states of $\mathcal{S}$
resulting from the infinitely repeated measurements of the observable $B$
over the infinite copies of $\mathcal{S}$}.

Then, according to (ii) of Definition~\ref{def-mixed-state},
the density matrix $\rho_B$ of $\gamma_B$ is given by
\begin{equation*}
  \rho_B:=\sum_{\ket{\psi}\in\Gamma^B}P'(\ket{\psi})\ket{\psi}\bra{\psi}.
\end{equation*}
Note from \eqref{2dP'U=Pa} and \eqref{2dPm=ipMS} that
$$P'(\ket{\psi})=\sum_{(m,l)\in S(\ket{\psi})}p_m\bra{\Psi_m}F_l\ket{\Psi_m}$$
for every $\ket{\psi}\in\Gamma^B$.
Thus, since $\omega$ is Martin-L\"of $P$-random,
using Theorem~\ref{nonzero-prob-appears} and \eqref{2dmrMSB}
we have that
\begin{equation*}
  \rho_B=\sum_{(m,l):P(m,l)>0}p_mF_l\ket{\Psi_m}\bra{\Psi_m}F_l,
\end{equation*}
where the sum is over all $(m,l)\in\Omega\times\Theta$ such that $P(m,l)>0$.
It follows from \eqref{2dPm=ipMS} that
\begin{equation*}
  \rho_B=\sum_{(m,l)\in\Omega\times\Theta}p_mF_l\ket{\Psi_m}\bra{\Psi_m}F_l.
\end{equation*}
Hence, using \eqref{2drhoA} we have that
\begin{equation*}
  \rho_B=\sum_{l\in\Theta} F_l\rho_A F_l,
\end{equation*}
as expected from the point of view of the conventional quantum mechanics.

\subsection{\boldmath Mixed state conditioned by a specific outcome of the measurements of $B$}
\label{d2MSCSOMB}

We calculate the mixed state resulting from the measurements of the observable $B$,
conditioned by a specific measurement outcome of $B$.
This mixed state is the post-measurement state
mentioned in the last sentence of Postulate~\ref{Tadaki-rule2}.

Let $\Lambda:=\{\beta(n)\mid n\in\N^+\}$,
which is the set of all possible measurement outcomes of $B$ in our world.
Let $l_0$ be an arbitrary element of $\Lambda$, and let $C:=\{(m,l_0)\mid m\in\Omega\}$.
Since $\omega$ is Martin-L\"of $P$-random,
it follows from Theorem~\ref{nonzero-prob-appears} that $P(C)>0$.
Let $\delta:=\cond{C}{\omega}$.
Recall that $\cond{C}{\omega}$ is defined as an infinite sequence in $C^\infty$ obtained from
$\omega$ by eliminating all elements of $(\Omega\times\Theta)\setminus C$ occurring in $\omega$.
In other words, $\delta$ is the subsequence of $\omega$ such that
the outcomes of measurements of $B$ equal to $l_0$.
We
assume that for each $k\in\N^+$,
the $k$th element of $\delta$ is originally the $n_k$th element of $\omega$ before the elimination.
It follows that $\delta(k)=(\alpha(n_k),l_0)$ for every $k\in\N^+$.
Since $\omega$ is Martin-L\"of $P$-random,
using Theorem~\ref{conditional_probability} we have that
$\delta$ is Martin-L\"of $P_C$-random for the finite probability space $P_C$ on $C$.

Let $\zeta$ be an infinite sequence over $\Omega$ such that $\zeta(k):=\alpha(n_k)$
for every $k\in\N^+$.
The
sequence $\zeta$ can be phrased as \emph{the infinite sequence of outcomes of the
measurements of the observable $A$,
resulting from the infinitely repeated measurements of the observable $B$
following the measurement of the observable $A$
over the infinite copies of the two systems $\ancilla$ and $\mathcal{S}$ in our world,
conditioned by the specific measurement outcome $l_0$ of $B$}.
Since $\delta$ is Martin-L\"of $P_C$-random,
it is easy to show that $\zeta$ is Martin-L\"of $P_{l_0}$-random,
where $P_{l_0}$ is a finite probability space on $\Omega$ such that
$P_{l_0}(m)=P_C(m,l_0)$ for every $m\in\Omega$.
We
see that
\begin{equation}\label{d2Pl0m=tr}
  P_{l_0}(m)=P_C(m,l_0)=\frac{P(m,l_0)}{\sum_{m'\in\Omega}P(m',l_0)}
  =\frac{p_m\bra{\Psi_m}F_{l_0}\ket{\Psi_m}}{\tr(F_{l_0}\rho_A)}
\end{equation}
for each $m\in\Omega$, where the last equality follows from \eqref{2dPm=ipMS} and \eqref{2dPm=ipMSR}.
Note that the denominator $\tr(F_{l_0}\rho_A)$ on the most right-hand side of \eqref{d2Pl0m=tr} is certainly
non-zero due to \eqref{2dnon-zero-inner-productMSR}.

Let $\mu$ is a subsequence sequence of $\gamma_B$ such that
$\mu(k):=\gamma_B(n_k)$ for every $k\in\N^+$.
Then, it follows from \eqref{2dmrMSB} that
\begin{equation}\label{mu=k=FEP2}
  \mu(k)=\ket{\Upsilon^B_{n_k}}
  =\frac{F_{l_0}\ket{\Psi_{\zeta(k)}}}{\sqrt{\bra{\Psi_{\zeta(k)}}F_{l_0}\ket{\Psi_{\zeta(k)}}}}
\end{equation}
for every $k\in\N^+$.
Let $\Xi^B:=\{\mu(k)\mid k\in\N^+\}$. Note that $\Xi^B$ is an alphabet.
For each $\ket{\psi}\in\Xi^B$, we define $T(\ket{\psi})$ as the set $\{\zeta(k)\mid\ket{\psi}=\mu(k)\}$.
Since $\zeta$ is Martin-L\"of $P_{l_0}$-random,
it follows from Theorem~\ref{nonzero-prob-appears} and Theorem~\ref{contraction}
that the infinite sequence $\mu$ is Martin-L\"of $P_{l_0}'$-random,
where $P_{l_0}'$ is a finite probability space on $\Xi^B$ such that
\begin{equation}\label{Pl0'mk=Pl0zk2}
  P_{l_0}'(\ket{\psi})=\sum_{m\in T(\ket{\psi})}P_{l_0}(m)
\end{equation}
for every $\ket{\psi}\in\Xi^B$.
Thus, according to (i) of Definition~\ref{def-mixed-state},
the infinite sequence $\mu$ is a mixed state of $\mathcal{S}$.
The mixed state $\mu$ can be
phrased
as
\emph{the infinite sequence of states of the system $\mathcal{S}$,
resulting from the infinitely repeated measurements of the observable $B$
following the measurement of the observable $A$
over the infinite copies of the two systems $\ancilla$ and $\mathcal{S}$ in our world,
conditioned by the specific measurement outcome $l_0$ of $B$}.

Then, according to (ii) of Definition~\ref{def-mixed-state},
the density matrix $\rho_{l_0}$ of $\mu$ is given by
\begin{equation*}
  \rho_{l_0}:=\sum_{\ket{\psi}\in\Xi^B}P_{l_0}'(\ket{\psi})\ket{\psi}\bra{\psi}.
\end{equation*}
Thus, since $\zeta$ is Martin-L\"of $P_{l_0}$-random,
using Theorem~\ref{nonzero-prob-appears}, \eqref{mu=k=FEP2}, and \eqref{Pl0'mk=Pl0zk2}
we have that
\begin{equation*}
  \rho_{l_0}
  =\sum_{m:P_{l_0}(m)>0}
  P_{l_0}(m)\frac{F_{l_0}\ket{\Psi_m}\bra{\Psi_m}F_{l_0}}{\bra{\Psi_m}F_{l_0}\ket{\Psi_m}},
\end{equation*}
where the sum is over all $m\in\Omega$ such that $P_{l_0}(m)>0$.
It follows from \eqref{d2Pl0m=tr} that
\begin{equation*}
  \rho_{l_0}=\sum_{m:P_{l_0}(m)>0}\frac{p_mF_{l_0}\ket{\Psi_m}\bra{\Psi_m}F_{l_0}}{\tr(F_{l_0}\rho_A)}.
\end{equation*}
Note from \eqref{d2Pl0m=tr} that, for every $m\in\Omega$,
$P_{l_0}(m)=0$ if and only if $p_m=0$ or $F_{l_0}\ket{\Psi_m}=0$.
Thus, using \eqref{2drhoA} we finally have that
\begin{equation*}
  \rho_{l_0}
  =\sum_{m\in\Omega}\frac{p_mF_{l_0}\ket{\Psi_m}\bra{\Psi_m}F_{l_0}}{\tr(F_{l_0}\rho_A)}
  =\frac{F_{l_0}\rho_A F_{l_0}}{\tr(F_{l_0}\rho_A)}.
\end{equation*}

Recall that $l_0$ is an arbitrary element of $\Lambda$.
Hence, in summary, we see that for every $l\in\Lambda$,
\begin{equation}\label{post-measuremnt-state}
  \rho_{l}:=\frac{F_{l}\rho_A F_{l}}{\tr(F_{l}\rho_A)}
\end{equation}
is the density matrix of \emph{the mixed state
resulting from the infinitely repeated measurements of the observable $B$
following the measurement of the observable $A$
over the infinite copies of the two systems $\ancilla$ and $\mathcal{S}$ in our world,
conditioned by the specific measurement outcome $l$ of $B$}.
The result is just as expected from the aspect of the conventional quantum mechanics, i.e.,
Postulate~\ref{Born-rule2}.

\subsection{Derivation of Postulate~\ref{Tadaki-rule2}}

At last,
we show that Postulate~\ref{Tadaki-rule2} can be derived from Postulate~\ref{POT}
together with Postulates~\ref{state_space}, \ref{composition}, and \ref{evolution}
in the setting developed
in the preceding subsections.

For deriving Postulate~\ref{Tadaki-rule2},
we consider the quantum measurements described by the observable $B$
introduced in Subsection~\ref{d2ROM}
in the above setting,
as the quantum measurements described by the observable $M$ in Postulate~\ref{Tadaki-rule2}.
Recall that the spectrum of $B$ is $\Theta$.
Then the spectrum decomposition \eqref{d2spectral-decompositionB} of $B$ is rewritten as
$$B=\sum_{l\in\Theta} l F_l,$$
in the same form as
stated
for $M$ in Postulate~\ref{Tadaki-rule2}.

Suppose that the measurements of the observable $B$ are
repeatedly performed over a mixed state with a density matrix $\rho$,
as assumed in Postulate~\ref{Tadaki-rule2}.
In the setting developed in the preceding subsections,
this mixed state is the infinite sequence $\gamma_A$ over $\Gamma^A$,
and its density matrix $\rho$ is the density matrix $\rho_A$ of the mixed state $\gamma_A$,
according to the arguments in Subsection~\ref{d2MSRFMOA}.
The infinite sequence of outcomes generated by the measurements,
which is mentioned in Postulate~\ref{Tadaki-rule2},
is the infinite sequence $\beta$ over $\Theta$ in our world $\omega$,
according to the arguments in Subsection~\ref{d2APOT}.
Based on this identification, we can see the following:

First, since $\beta$ is an infinite sequence over $\Theta$,
the set of possible outcomes of the measurement is the spectrum $\Theta$ of $B$,
as stated in Postulate~\ref{Tadaki-rule2} as one of its conclusions.

Secondly, recall that $\beta$ is Martin-L\"of $R$-random,
where $R$ is the finite probability space on $\Theta$ satisfying \eqref{2dPm=ipMSR},
according to the arguments in Subsection~\ref{d2ISBMOOB}.
Thus, since $\rho_A=\rho$, we see that
\emph{in our world},
the infinite sequence of outcomes generated by the measurements,
which is mentioned in Postulate~\ref{Tadaki-rule2},
is a Martin-L\"of $V$-random infinite sequence over $\Theta$,
where $V$ is a finite probability space on $\Theta$ such that
$$V(l)=\tr(F_l\rho)$$
for every $l\in\Theta$,
as stated in Postulate~\ref{Tadaki-rule2} as one of its conclusions.

Thirdly, according to the arguments in Subsection~\ref{d2MSCSOMB},
the set of all outcomes of the measurements of $B$ in our world is $\Lambda$,
and for every $l\in\Lambda$,
the density matrix of the mixed state
resulting from the infinitely repeated measurements of $B$
in our world,
conditioned by the specific measurement outcome $l$ of $B$,
is $\rho_l$ given by \eqref{post-measuremnt-state}.
Since $\rho_A=\rho$,
this means that
\emph{in our world},
the resulting sequence of pure states with outcome $l$ is
a mixed state with the density matrix
\begin{equation*}
  \frac{F_l\rho F_l}{\tr(F_l\rho)},
\end{equation*}
as stated in the last sentence of Postulate~\ref{Tadaki-rule2} as one of its conclusions.

Hence, we have derived Postulate~\ref{Tadaki-rule2} from Postulate~\ref{POT},
the principle of typicality,
together with Postulates~\ref{state_space}, \ref{composition}, and \ref{evolution}
in the setting developed so far in the preceding subsections,
i.e., in a \emph{more general scenario} than that of Section~\ref{POT-mixed-states},
where a mixed state being measured is
an
infinite sequence over
\emph{general mutually non-orthogonal} pure states.

\section{Clarification of the validity of the postulate for composition of mixed states}
\label{CPFCOMS}

In this section and the next section we investigate the validity of \emph{other} postulates
of the conventional quantum mechanics
\emph{regarding mixed states and density matrices},
in terms of our framework based on Postulate~\ref{POT}, the principle of typicality,
together with Postulates~\ref{state_space}, \ref{composition}, and \ref{evolution}.
We consider two of such postulates.
One of them is Postulate~\ref{COSMS} below,
which is Postulate 4 described in Nielsen and Chuang \cite[Section 2.4.2]{NC00}.
It describes how the density matrix of a mixed state of a composite system is
calculated from the density matrices of the mixed states of the component systems.
The other is Postulate~\ref{density-matrices-probability} in the next section,
which is a part of Postulate 1 described in Nielsen and Chuang \cite[Section 2.4.2]{NC00}.
It treats the ``probabilistic mixture'' of mixed states.

\begin{postulate}[Composition of systems]\label{COSMS}
The state space of a composite physical system is the tensor product of the state spaces of the component physical systems.
Moreover, if we have systems numbered $1$ through $K$, and system number $k$ is prepared in the
state
$\rho_k$,
then the joint state of the total system is
$\rho_1\otimes\rho_2\otimes\dots\otimes\rho_K$.
\qed
\end{postulate}

In this section, we first point out the
failure
of Postulate~\ref{COSMS}.
Postulate~\ref{COSMS} consists of two statements.
The first statement of Postulate~\ref{COSMS} is about the composition of state spaces
by the tensor product operation. It is, of course, guaranteed by Postulate~\ref{composition},
and therefore we accept it as an unconditionally true fact.
The problem is the second statement of Postulate~\ref{COSMS} about the composition of
the density matrices of mixed states of the component systems by the tensor product operation.
We point out the failure of the second statement of Postulate~\ref{COSMS} first in what follows.
We then present a \emph{necessary and sufficient condition}
for the second statement of Postulate~\ref{COSMS} to hold under a certain natural restriction on
the forms
of the mixed states of the component systems,
in terms of our framework based on Definition~\ref{def-mixed-state}.
After that, we
describe
a natural and simple scenario regarding the setting of measurements
in which the second statement of Postulate~\ref{COSMS} holds,
based on Postulate~\ref{POT}
together with Postulates~\ref{state_space}, \ref{composition}, and \ref{evolution}.
In the scenario, we
consider an infinite repetition of independent measurements of
observables over the infinite copies of an identical composite system prepared
in an identical state.

\subsection{Failure of Postulate~\ref{COSMS}}
\label{FoPCOSMS}

First, we see the failure of Postulate~\ref{COSMS} in the contexts of \emph{both}
the conventional quantum mechanics and
our framework based on
Postulate~\ref{POT} and Definition~\ref{def-mixed-state}.
Let $\mathcal{Q}_1$ and $\mathcal{Q}_2$ be a qubit system with
a state space $\mathcal{H}$ of dimension two.

The second statement of Postulate~\ref{COSMS} is obviously invalid
even in the conventional quantum mechanics.
To see this,
consider a composite system consisting of $\mathcal{Q}_1$ and $\mathcal{Q}_2$, and consider
the \emph{Bell state}
$$\ket{\Psi}:=\frac{\ket{00}+\ket{11}}{\sqrt{2}}$$
of the composite system.
Then,
the density matrix of this
state of the composite system is given by
\begin{equation}\label{DMoBell-state}
  \ket{\Psi}\bra{\Psi}
\end{equation}
in the conventional quantum mechanics.
Recall that,
in the conventional quantum mechanics,
the density matrix of the state of the subsystem is given by the \emph{reduced density matrix} of
the density matrix of a state of a composite system
via the \emph{partial trace operation}
(see Nielsen and Chuang \cite[Section 2.4.3]{NC00}).
Thus,
according to
the conventional quantum mechanics,
the density matrix $\rho_1$ of the corresponding state of the subsystem $\mathcal{Q}_1$ is
given by
\begin{equation}\label{rho1tr2kb=fI1}
  \rho_1=\tr_2(\ket{\Psi}\bra{\Psi})=\frac{I}{2},
\end{equation}
where $\tr_2$ is the partial trace operation over the system $\mathcal{Q}_2$,
and $I$ is the identity operator on $\mathcal{H}$.
Similarly, the density matrix $\rho_2$ of the corresponding state of the subsystem $\mathcal{Q}_2$ is
given by
\begin{equation}\label{rho2tr1kb=fI2}
  \rho_2=\tr_1(\ket{\Psi}\bra{\Psi})=\frac{I}{2},
\end{equation}
where $\tr_1$ is the partial trace operation over the system $\mathcal{Q}_1$.

Now, let us apply Postulate~\ref{COSMS} with $K=2$ to the above situation.
Then,
by \eqref{rho1tr2kb=fI1} and \eqref{rho2tr1kb=fI2} we see that
the density matrix of the state of the composite system is given by
$$\rho_1\otimes\rho_2=\frac{I\otimes I}{4}.$$
However, this density matrix
is obviously different from the density matrix \eqref{DMoBell-state}
of the composite system.
Thus, we have a contradiction.
In this way, we see that Postulate~\ref{COSMS} is not valid
in the conventional quantum mechanics.
The second statement of Postulate~\ref{COSMS} does \emph{not} hold \emph{unlimitedly}
even in the conventional quantum mechanics.

The second statement of Postulate~\ref{COSMS} is also invalid in our framework based on
Postulate~\ref{POT} and Definition~\ref{def-mixed-state}.
To see this, consider countably infinite copies
$\mathcal{Q}_1^1, \mathcal{Q}_1^2, \mathcal{Q}_1^3, \dotsc$ of the system $\mathcal{Q}_1$
and countably infinite copies
$\mathcal{Q}_2^1, \mathcal{Q}_2^2, \mathcal{Q}_2^3, \dotsc$ of the system $\mathcal{Q}_2$.
We
choose a specific real $p$ with $0<p<1$, and
consider a finite probability space $P$ on
the alphabet $\{\ket{0},\ket{1}\}\subset \mathcal{H}$
such that $P(\ket{0})=p$ and $P(\ket{1})=1-p$.
We then prepare each of the systems $\mathcal{Q}_1^i$ and $\mathcal{Q}_2^i$
in an identical state $\alpha(i)$
for
every
$i\in\N^+$,
where $\alpha$ is a Martin-L\"of $P$-random infinite sequence over
the alphabet $\{\ket{0},\ket{1}\}$.
This preparation can be possible by the following procedure:
First, we prepare all of the systems $\{\mathcal{Q}_1^i\}$ in the identical state
$\sqrt{p}\ket{0}+\sqrt{1-p}\ket{1}$, and
then perform the measurements described by the observable $\ket{1}\bra{1}$
over all these systems.
Then, according to Postulate~\ref{Tadaki-rule},
which follows from Postulate~\ref{POT} as we saw in Section~\ref{POT-pure-states},
an infinite sequence $\omega$ of measurement outcomes being generated is
a Martin-L\"of $Q$-random infinite sequence over the alphabet $\{0,1\}$,
where $Q$ is a finite probability space on $\{0,1\}$ such that $Q(0)=p$ and $Q(1)=1-p$.
Moreover, an infinite sequence $\alpha$ of the resulting states
from
the measurements is
a Martin-L\"of $P$-random infinite sequence over the alphabet $\{\ket{0},\ket{1}\}$,
where the state of the system $\mathcal{Q}_1^i$ is $\alpha(i)$ for every $i\in\N^+$.
We then
prepare the system $\mathcal{Q}_2^i$ in the state $\alpha(i)$ for each $i\in\N^+$.
This preparation is possible since $\alpha(i)=\ket{\omega(i)}$ for all $i\in\N^+$ and therefore
we know each $\alpha(i)$ from the measurement outcomes $\omega$.
As a result,
both the systems $\mathcal{Q}_1^i$ and $\mathcal{Q}_2^i$ are prepared
in the identical state $\alpha(i)$ for every $i\in\N^+$.

Then, according to (i) of Definition~\ref{def-mixed-state}, the infinite sequence $\alpha$ is
a mixed state of each of the systems $\mathcal{Q}_1$ and $\mathcal{Q}_2$.
Thus, according to (ii) of Definition~\ref{def-mixed-state},
the density matrix $\rho_1$ of the infinite sequence $\alpha$
as a mixed state of the system $\mathcal{Q}_1$ is given by
\begin{equation}\label{rho1=spk0b01-pk1b1}
  \rho_1
  =\sum_{\ket{\Psi}\in\{\ket{0},\ket{1}\}} P(\ket{\Psi})\ket{\Psi}\bra{\Psi}
  =p\ket{0}\bra{0}+(1-p)\ket{1}\bra{1},
\end{equation}
since $\alpha$ is Martin-L\"of $P$-random.
On the other hand,
the density matrix $\rho_2$ of the infinite sequence $\alpha$
as a mixed state of the system $\mathcal{Q}_2$ is given by
\begin{equation}\label{rho2=spk0b01-pk1b1}
  \rho_2
  =p\ket{0}\bra{0}+(1-p)\ket{1}\bra{1},
\end{equation}
due to the same reason.
For each $i\in\N^+$,
since both the systems $\mathcal{Q}_1^i$ and $\mathcal{Q}_2^i$ are in the state $\alpha(i)$,
it follows from Postulate~\ref{composition} that the composite system consisting of
$\mathcal{Q}_1^i$ and $\mathcal{Q}_2^i$ is in the state $\alpha(i)\otimes\alpha(i)$.
Thus, according to (i) of Definition~\ref{def-mixed-state},
$\alpha\otimes\alpha$ is a mixed state of the composite system consisting of
the systems $\mathcal{Q}_1$ and $\mathcal{Q}_2$, where $\alpha\otimes\alpha$ is
an infinite sequence over the alphabet $\{\ket{0}\otimes\ket{0},\ket{1}\otimes\ket{1}\}$
defined by the condition that $(\alpha\otimes\alpha)(i)=\alpha(i)\otimes\alpha(i)$ for all $i\in\N^+$.
Obviously, $\alpha\otimes\alpha$ is Martin-L\"of $R$-random,
where $R$ is a finite probability space on $\{\ket{0}\otimes\ket{0},\ket{1}\otimes\ket{1}\}$ such that
$R(\ket{0}\otimes\ket{0})=p$ and $R(\ket{1}\otimes\ket{1})=1-p$.
Thus, according to (ii) of Definition~\ref{def-mixed-state},
the density matrix $\rho_{1,2}$ of the mixed state $\alpha\otimes\alpha$ of
the composite system consisting of
$\mathcal{Q}_1$ and $\mathcal{Q}_2$ is given by
\begin{equation}\label{rho12=spk0k0b0b01-pk1k1b1b1}
  \rho_{1,2}
  =\sum_{\ket{\Psi}\in\{\ket{0}\otimes\ket{0},\ket{1}\otimes\ket{1}\}} R(\ket{\Psi})\ket{\Psi}\bra{\Psi}
  =p\ket{0}\otimes\ket{0}\bra{0}\otimes\bra{0}+(1-p)\ket{1}\otimes\ket{1}\bra{1}\otimes\bra{1}.
\end{equation}

Now, let us apply Postulate~\ref{COSMS} with $K=2$ to the present situation.
Then, by \eqref{rho1=spk0b01-pk1b1} and \eqref{rho2=spk0b01-pk1b1} we see that
the density matrix of the mixed state $\alpha\otimes\alpha$ of
the composite system consisting of
$\mathcal{Q}_1$ and $\mathcal{Q}_2$ is given by
\[
  \rho_1\otimes\rho_2
  =\Bigl[p\ket{0}\bra{0}+(1-p)\ket{1}\bra{1}\Bigr]\otimes\Bigl[p\ket{0}\bra{0}+(1-p)\ket{1}\bra{1}\Bigr].
\]
However,
this density matrix cannot equal to the density matrix \eqref{rho12=spk0k0b0b01-pk1k1b1b1}
no matter how the real $p$ is chosen under the
constraint
that $0<p<1$.
Thus, we have a contradiction.
In this way, we see that Postulate~\ref{COSMS} is not valid
\emph{also} in our framework based on Postulate~\ref{POT} and Definition~\ref{def-mixed-state}.
The second statement of Postulate~\ref{COSMS} does not hold unlimitedly \emph{also} in our framework.

A
natural question that may arise
based on the above observations
is
what scope the second statement of Postulate~\ref{COSMS} holds
in our framework based on
Definition~\ref{def-mixed-state}.
In the next subsection, we answer this question by presenting a necessary and sufficient condition
for the second statement of Postulate~\ref{COSMS} to hold under a certain natural restriction on
the forms of the mixed states of the component systems.

\subsection{The scope within which the statement of Postulate~\ref{COSMS}
is valid}
\label{RoPCOSMS}

In this subsection, we rephrase the second statement of Postulate~\ref{COSMS} equivalently
in terms of the notion of the \emph{independence} of Martin-L\"of $P$-random infinite sequences.

First, we introduce this notion of \emph{independence} in the following manner:
Let $\Omega_1,\dots,\Omega_K$ be alphabets.
For any $\alpha_1\in\Omega_1^\infty,\dots,\alpha_K\in\Omega_K^\infty$,
we use $\alpha_1\times\dots\times\alpha_K$ to denote an infinite sequence
$\alpha$ over $\Omega_1\times\dots\times\Omega_K$
such that $\alpha(n)=(\alpha_1(n),\dots,\alpha_K(n))$ for every $n\in\N^+$.
On the other hand, for any $P_1\in\PS(\Omega_1),\dots,P_K\in\PS(\Omega_K)$,
we use $P_1\times\dots\times P_K$ to denote a finite probability space
$Q\in\PS(\Omega_1\times\dots\times\Omega_K)$ such that
$Q(a_1,\dots,a_K)=P_1(a_1)\dotsm P_K(a_K)$ for every $a_1\in\Omega_1,\dotsc,a_K\in\Omega_K$.

\begin{definition}[Independence of Martin-L\"of $P$-random infinite sequences]\label{independency-of-ensembles}
Let $\Omega_1,\dotsc,\Omega_K$ be alphabets, and let
$P_1\in\PS(\Omega_1),\dots,P_K\in\PS(\Omega_K)$.
For each $k=1,\dots,K$,
let $\alpha_k$ be a Martin-L\"of $P_k$-random infinite sequence over $\Omega_k$.
We say that $\alpha_1,\dots,\alpha_K$ are
\emph{independent}
if $\alpha_1\times\dots\times\alpha_K$ is Martin-L\"of $P_1\times\dots\times P_K$-random.
\qed
\end{definition}

Note that the notion of the independence of Martin-L\"of $P$-random infinite sequences
in our framework
is introduced by
Tadaki \cite{T15,T16arXiv},
suggested by
the notion of \emph{independence} of \emph{collectives}
in the theory of collectives
introduced
by von Mises \cite{vM64} (see Tadaki \cite{T16arXiv} for the detail).
Theorem~\ref{COSMS-independence} below
leads to
a necessary and sufficient condition
for the second statement of Postulate~\ref{COSMS} to hold
in terms of the notion of the independence of Martin-L\"of $P$-random infinite sequences,
under a certain \emph{natural} restriction on the forms of the mixed states of the component systems,
i.e., the restriction that the state vectors which constitute the mixed state of
each of the component systems
are linearly independent.
To prove Theorem~\ref{COSMS-independence},
in particular, to prove the result (ii) of Theorem~\ref{COSMS-independence},
we need the following two lemmas.
Both of them are immediate results of linear algebra.

\begin{lemma}\label{linear-ind-linear}
Let $\mathcal{H}_1$ and $\mathcal{H}_2$ be complex Hilbert spaces of finite dimension.
Let $\{\ket{\Psi^1_m}\mid 1\le m\le M\}$ and $\{\ket{\Psi^2_n}\mid 1\le n\le N\}$ be
two sets of
linearly independent vectors in $\mathcal{H}_1$ and $\mathcal{H}_2$, respectively.
Then $\{\ket{\Psi^1_m}\otimes\ket{\Psi^2_n}\mid 1\le m\le M\;\&\;1\le n\le N\}$ is
a set of
linearly independent vectors in $\mathcal{H}_1\otimes\mathcal{H}_2$.
\qed
\end{lemma}

\begin{lemma}\label{linear-ind-ketbra}
Let $\mathcal{H}$ be a complex Hilbert space of finite dimension.
Let $\{\ket{\Psi_n}\mid 1\le n\le N\}$ and $\{\ket{\Phi_n}\mid 1\le n\le N\}$ be
two sets of linearly independent vectors in $\mathcal{H}$.
Then $\{\ket{\Psi_n}\bra{\Phi_n}\mid 1\le n\le N\}$ is linearly independent over $\C$.
\qed
\end{lemma}

Let $\mathcal{S}^1,\dots,\mathcal{S}^K$ be arbitrary quantum systems
with state spaces $\mathcal{H}_1,\dots,\mathcal{H}_K$ of finite dimension, respectively.
For each $k=1,\dots,K$, let $\Omega_k$ be a non-empty finite set of state vectors in $\mathcal{H}_k$,
and $\gamma_k$ an infinite sequence over $\Omega_k$.
Then, we use $\Omega_1\otimes\dots\otimes\Omega_K$ to denote the set of all
$\ket{\Psi_1}\otimes\dots\otimes\ket{\Psi_K}$
such that $\ket{\Psi_k}\in\Omega_k$ for every $k=1,\dots,K$.%
\footnote{Note that $\Omega_1\otimes\dots\otimes\Omega_K$ is not a Hilbert space
but just an alphabet consisting of
a finite number of
state vectors.}
We then also use $\gamma_1\otimes\dots\otimes\gamma_K$ to denote an infinite sequence
$\gamma$ over $\Omega_1\otimes\dots\otimes\Omega_K$
such that $\gamma(n)=\gamma_1(n)\otimes\dots\otimes\gamma_K(n)$ for every $n\in\N^+$.

\begin{theorem}\label{COSMS-independence}
Let $\mathcal{S}^1,\dots,\mathcal{S}^K$ be arbitrary quantum systems
with state spaces $\mathcal{H}_1,\dots,\mathcal{H}_K$ of finite dimension, respectively.
For each $k=1,\dots,K$, let $\Omega_k$ be a non-empty finite set of state vectors in $\mathcal{H}_k$,
and $\gamma_k$ an infinite sequence over $\Omega_k$.
Suppose that
$\gamma_1,\dots,\gamma_K$ are mixed states of the systems $\mathcal{S}^1,\dots,\mathcal{S}^K$
with the density matrices $\rho_1,\dots,\rho_K$, respectively.
Then the following hold.
\begin{enumerate}
  \item If $\gamma_1,\dots,\gamma_K$ are independent, then
    $\gamma_1\otimes\dots\otimes\gamma_K$ is a mixed state
    of the composite system consisting of the systems $\mathcal{S}^1,\dots,\mathcal{S}^K$ and
    its density matrix is $\rho_1\otimes\dots\otimes\rho_K$.
  \item Suppose that each of $\Omega_1,\dots,\Omega_K$ is
    a set of linearly independent vectors.
    If $\gamma_1\otimes\dots\otimes\gamma_n$ is a mixed state
    of the composite system consisting of the systems $\mathcal{S}^1,\dots,\mathcal{S}^K$ and
    its density matrix is $\rho_1\otimes\dots\otimes\rho_K$, then
    $\gamma_1,\dots,\gamma_K$ are independent.
\end{enumerate}
\end{theorem}

\begin{proof}
Suppose that $\gamma_1,\dots,\gamma_K$ are mixed states of the systems $\mathcal{S}^1,\dots,\mathcal{S}^K$ with the density matrices $\rho_1,\dots,\rho_K$, respectively.
Then, according to Definition~\ref{def-mixed-state}, for each $k=1,\dots,K$
there exists $P_k\in\PS(\Omega_k)$ such that $\gamma_k$ is Martin-L\"of $P_k$-random and
\begin{equation}\label{thm-ind-dmk}
  \rho_k=\sum_{\ket{\Psi_k}\in\Omega_k} P_k(\ket{\Psi_k})\ket{\Psi_k}\bra{\Psi_k},
\end{equation}
It follows that
\begin{align}
  &\rho_1\otimes\dots\otimes\rho_K \nonumber \\
  &=\sum_{\ket{\Psi_1}\in\Omega_1}\dots\sum_{\ket{\Psi_K}\in\Omega_K}
  P_k(\ket{\Psi_1})\dotsm P_k(\ket{\Psi_K})\ket{\Psi_1}\bra{\Psi_1}\otimes\dots\otimes\ket{\Psi_K}\bra{\Psi_K} \nonumber \\
  &=\sum_{\ket{\Psi_1}\in\Omega_1}\dots\sum_{\ket{\Psi_K}\in\Omega_K}
  \left(P_1\times\dots\times P_K\right)(\ket{\Psi_1},\dots,\ket{\Psi_K})
  \ket{\Psi_1}\otimes\dots\otimes\ket{\Psi_K}\bra{\Psi_1}\otimes\dots\otimes\bra{\Psi_K}.
  \label{rho1rhok=sp1pk}
\end{align}

First, we prove (i) of Theorem~\ref{COSMS-independence}.
For that purpose, suppose that $\gamma_1,\dots,\gamma_K$ are independent.
Then, by Definition~\ref{independency-of-ensembles}, we have that
$\gamma_1\times\dots\times\gamma_K$ is Martin-L\"of $P_1\times\dots\times P_K$-random
over $\Omega_1\times\dots\times\Omega_K$.
We
consider a finite probability space $Q\in\PS(\Omega_1\otimes\dots\otimes\Omega_K)$ such that
\begin{equation}\label{QPs=sP1dPKkP1dkPK}
  Q(\ket{\Psi}):=\sum \left(P_1\times\dots\times P_K\right)(\ket{\Psi_1},\dots,\ket{\Psi_K}),
\end{equation}
for every $\ket{\Psi}\in\Omega_1\otimes\dots\otimes\Omega_K$,
where the sum is over all $(\ket{\Psi_1},\dots,\ket{\Psi_K})\in\Omega_1\times\dots\times\Omega_K$
such that $\ket{\Psi}=\ket{\Psi_1}\otimes\dots\otimes\ket{\Psi_K}$.
It is then easy to see that
$\gamma_1\otimes\dots\otimes\gamma_K$ is Martin-L\"of $Q$-random
over $\Omega_1\otimes\dots\otimes\Omega_K$,
where we use Theorem~\ref{contraction}, if necessary.
Thus, according to Definition~\ref{def-mixed-state},
$\gamma_1\otimes\dots\otimes\gamma_n$ is a mixed state
of the composite system consisting of the systems $\mathcal{S}^1,\dots,\mathcal{S}^K$ and
its density matrix $\rho$ is given by
\[
  \rho:=\sum_{\ket{\Psi}\in\Omega_1\otimes\dots\otimes\Omega_K} Q(\ket{\Psi})\ket{\Psi}\bra{\Psi}.
\]
Hence, using \eqref{rho1rhok=sp1pk} and \eqref{QPs=sP1dPKkP1dkPK}
we have that $\rho_1\otimes\dots\otimes\rho_K=\rho$, as desired.

Next, we prove (ii) of Theorem~\ref{COSMS-independence}.
For that purpose,
consider
a function $f\colon\Omega_1\times\dots\times\Omega_K\to\Omega_1\otimes\dots\otimes\Omega_K$
such that $f(\ket{\Psi_1},\dots,\ket{\Psi_K})=\ket{\Psi_1}\otimes\dots\otimes\ket{\Psi_K}$.
Suppose that each of $\Omega_1,\dots,\Omega_K$ is a set of linearly independent vectors.
Then, it follows from
Lemma~\ref{linear-ind-linear}
that $f$ is a bijection.
We suppose, moreover, that $\gamma_1\otimes\dots\otimes\gamma_n$ is a mixed state
of the composite system consisting of the systems $\mathcal{S}^1,\dots,\mathcal{S}^K$ with
the density matrix $\rho_1\otimes\dots\otimes\rho_K$.
Then, according to Definition~\ref{def-mixed-state},
there exists a finite probability space $R\in\PS(\Omega_1\otimes\dots\otimes\Omega_K)$
such that $\gamma_1\otimes\dots\otimes\gamma_K$ is Martin-L\"of $R$-random
over $\Omega_1\otimes\dots\otimes\Omega_K$ and
\[
  \rho_1\otimes\dots\otimes\rho_K
  =\sum_{\ket{\Psi}\in\Omega_1\otimes\dots\otimes\Omega_K} R(\ket{\Psi})\ket{\Psi}\bra{\Psi}.
\]
Since $f$ is a bijection,
it is easy to see that $\gamma_1\times\dots\times\gamma_K$ is Martin-L\"of $R\circ f$-random over
$\Omega_1\times\dots\times\Omega_K$ and
\[
  \rho_1\otimes\dots\otimes\rho_K
  =\sum_{\ket{\Psi_1}\in\Omega_1}\dots\sum_{\ket{\Psi_K}\in\Omega_K}
  (R\circ f)(\ket{\Psi_1},\dots,\ket{\Psi_K})
  \ket{\Psi_1}\otimes\dots\otimes\ket{\Psi_K}\bra{\Psi_1}\otimes\dots\otimes\bra{\Psi_K}.
\]
Combining this with \eqref{rho1rhok=sp1pk}, it follows from
Lemmas~\ref{linear-ind-linear} and \ref{linear-ind-ketbra}
that $R\circ f=P_1\times\dots\times P_K$.
Thus, 
$\gamma_1\times\dots\times\gamma_K$ is Martin-L\"of $P_1\times\dots\times P_K$-random.
Therefore, according to Definition~\ref{independency-of-ensembles},
we see that $\gamma_1,\dots,\gamma_K$ are independent, as desired.
\end{proof}

Corollary~\ref{cor-COSMS-independence} below is immediate from Theorem~\ref{COSMS-independence}.
It is a precise form of the necessary and sufficient condition
for the second statement of Postulate~\ref{COSMS} to hold
in terms of the notion of the independence of Martin-L\"of $P$-random infinite sequences,
under the
natural
restriction on the forms of the mixed states of the component systems.

\begin{corollary}\label{cor-COSMS-independence}
Let $\mathcal{S}^1,\dots,\mathcal{S}^K$ be arbitrary quantum systems
with state spaces $\mathcal{H}_1,\dots,\mathcal{H}_K$ of finite dimension, respectively.
For each $k=1,\dots,K$, let $\Omega_k$ be a non-empty finite set of state vectors in $\mathcal{H}_k$,
and $\gamma_k$ an infinite sequence over $\Omega_k$.
Suppose that
$\gamma_1,\dots,\gamma_K$ are mixed states of the systems $\mathcal{S}^1,\dots,\mathcal{S}^K$
with the density matrices $\rho_1,\dots,\rho_K$, respectively.
Suppose, moreover, that each of $\Omega_1,\dots,\Omega_K$ is a set of linearly independent vectors.
Then the following conditions are equivalent:
\begin{enumerate}
  \item The infinite sequence $\gamma_1\otimes\dots\otimes\gamma_n$
    over $\Omega_1\otimes\dots\otimes\Omega_K$ is a mixed state of
    the composite system consisting of the systems $\mathcal{S}^1,\dots,\mathcal{S}^K$ and
    its density matrix is $\rho_1\otimes\dots\otimes\rho_K$.
  \item The mixed states $\gamma_1,\dots,\gamma_K$ are independent.\qed
\end{enumerate}
\end{corollary}

\subsection{Natural and simple scenario realizing the statement of Postulate~\ref{COSMS}}
\label{NaSscenarioRSPosCOSMS}

In the rest of this section we
provide
a natural and simple scenario, regarding the setting of measurements,
in which the second statement of Postulate~\ref{COSMS} holds.
We do this
based on Postulate~\ref{POT}, the principle of typicality,
together with Postulates~\ref{state_space}, \ref{composition}, and \ref{evolution},
as we did in Sections~\ref{POT-pure-states}--\ref{POT-mixed-states2}.

Let $\mathcal{S}^1,\dots,\mathcal{S}^K$ be arbitrary quantum systems
with state spaces $\mathcal{H}_1,\dots,\mathcal{H}_K$ of finite dimension, respectively.
Let $A_1,\dots,A_K$ be arbitrary observables of $\mathcal{S}^1,\dots,\mathcal{S}^K$, respectively,
and let $\ket{\Psi_1},\dots,\ket{\Psi_K}$ be arbitrary states of $\mathcal{S}^1,\dots,\mathcal{S}^K$,
respectively.
In the scenario,
we consider an infinite repetition of \emph{independent} measurements of
the observables $A_1,\dots,A_K$,
over infinite copies of the composite system consisting of the systems
$\mathcal{S}^1,\dots,\mathcal{S}^K$
prepared
in the identical state $\ket{\Psi_1}\otimes\dots\otimes\ket{\Psi_K}$,
where the measurement of $A_k$ is performed in the state $\ket{\Psi_k}$ for every $k=1,\dots,K$. 
Then, the application of the principle of typicality to this setting leads to
the satisfaction of the second statement of Postulate~\ref{COSMS}
by
mixed states of the systems $\mathcal{S}^1,\dots,\mathcal{S}^K$ and
the composite system consisting of them in our world.
Actually,
we confirm the second statement of Postulate~\ref{COSMS} in this scenario
by demonstrating that
the condition (ii) of Corollary~\ref{cor-COSMS-independence} is satisfied in the scenario,
based on the principle of typicality.
The detail is
described
in what follows.

\subsection{Repeated once of measurements}
\label{dmtROM}

First, the repeated once of the infinite repetition of
the set of measurements consisting of
the independent measurements of the observables $A_1,\dots,A_K$
is described in detail as follows.

Let $k$ be an arbitrary integer with $1\le k\le K$.
We use $\Omega_k$ to denote the spectrum of the observable $A_k$, and
for each $m\in\Omega_k$
we
use
$E_{k,m}$
to denote the projector onto the eigenspace of $A_k$ with eigenvalue $m$.
Then, according to Postulates~\ref{state_space}, \ref{composition}, and \ref{evolution},
the measurement process of the observable $A_k$ is described
by the following unitary operator $U_{A_k}$:
\begin{equation}\label{dmt-single_measurementAk}
  U_{A_k}\ket{\Psi}\otimes\ket{\Phi_{A_k}^{\mathrm{init}}}
  =\sum_{m\in\Omega_k}(E_{k,m}\ket{\Psi})\otimes\ket{\Phi_{A_k}[m]}
\end{equation}
for every $\ket{\Psi}\in\mathcal{H}_k$.
The vector $\ket{\Phi_{A_k}^{\mathrm{init}}}$ is
the initial state of the apparatus measuring $A_k$, and
$\ket{\Phi_{A_k}[m]}$ is the final state of
the apparatus
measuring $A_k$ for each $m\in\Omega_k$,
with $\braket{\Phi_{A_k}[m]}{\Phi_{A_k}[m']}=\delta_{m,m'}$.
For every $m\in\Omega_k$, the state $\ket{\Phi_{A_k}[m]}$ indicates that
\emph{the apparatus measuring the observable $A_k$ of the system $\mathcal{S}_k$
records the value $m$ of $A_k$}.

Let us
investigate
the whole measurement process of
the set of
measurements of the observables $A_1,\dots,A_K$.
We use $\Omega$ to denote the alphabet $\Omega_1\times\dots\times\Omega_K$, and
use $E_{(m_1,\dots,m_K)}$ to denote the projector $E_{1,m_1}\otimes\dots\otimes E_{K,m_K}$
for each $(m_1,\dots,m_K)\in\Omega$.
Moreover, we use $\ket{\Phi_A^{\mathrm{init}}}$ to denote the state
$\ket{\Phi_{A_1}^{\mathrm{init}}}\otimes\dots\otimes\ket{\Phi_{A_K}^{\mathrm{init}}}$, and
use $\ket{\Phi_A[(m_1,\dots,m_K)]}$ to denote the state
$\ket{\Phi_{A_1}[m_1]}\otimes\dots\otimes\ket{\Phi_{A_K}[m_K]}$
for each $(m_1,\dots,m_K)\in\Omega$.
Then, the sequential application of $U_{A_1},\dots,U_{A_K}$ in \emph{any} order to
the composite system consisting of the systems $\mathcal{S}_1,\dots,\mathcal{S}_K$
and the $K$ apparatuses measuring $A_1,\dots,A_K$ respectively over an initial state
$\ket{\Psi}\otimes\ket{\Phi_A^{\mathrm{init}}}$ with
$\ket{\Psi}\in\mathcal{H}_1\otimes\dots\otimes\mathcal{H}_K$
results in the following single unitary time-evolution $U$:
\begin{equation}\label{dmt-single_measurementAkEm}
  U\ket{\Psi}\otimes\ket{\Phi_A^{\mathrm{init}}}
  =\sum_{m\in\Omega}(E_{m}\ket{\Psi})\otimes\ket{\Phi_A[m]},
\end{equation}
where $U=U_{A_1}\otimes\dots\otimes U_{A_K}$.
It is then easy to check that the collection $\{E_m\}_{m\in\Omega}$ forms
a \emph{complete set of projectors},
i.e., the collection satisfies that $E_mE_{m'}=\delta_{m,m'} E_m$ and $\sum_{m\in\Omega}E_m=I$.
Hence, in particular, the collection $\{E_m\}_{m\in\Omega}$ forms measurement operators.
In this way,
the successive measurements of $A_1,\dots,A_k$ in any order can be regarded
as
a
\emph{single measurement} which is described by
the complete set $\{E_m\}_{m\in\Omega}$ of projectors.

\subsection{\boldmath Infinite repetition of the independent measurements of $A_1,\dots,A_k$}

Similarly in the preceding sections,
we consider countably infinite copies of
each of
the systems $\mathcal{S}_1,\dots,\mathcal{S}_K$
and
each of
the $K$ apparatuses measuring the observable $A_1,\dots,A_K$
of $\mathcal{S}_1,\dots,\mathcal{S}_K$, respectively.
We prepare each of the copies of the composite system consisting of
$\mathcal{S}_1,\dots,\mathcal{S}_K$ and the $K$ apparatuses
in the identical state $\ket{\Psi_1}\otimes\dots\otimes\ket{\Psi_K}\otimes\ket{\Phi_A^{\mathrm{init}}}$,
and then perform \emph{one by one}
the independent measurements of the observables $A_1,\dots,A_K$
over each of the copies of the systems $\mathcal{S}_1,\dots,\mathcal{S}_K$, respectively,
by interacting each of the copies of the systems $\mathcal{S}_1,\dots,\mathcal{S}_K$
with the corresponding copy of the apparatus
according to the unitary time-evolution \eqref{dmt-single_measurementAk}.
Obviously,
this setting can be reduced to the setting which we
developed
in
Subsection~\ref{FOMWI}
for an infinite repetition of the measurement of a single observable over a single system.
The detail is as follows.

Let $\mathcal{S}$ be the composite system consisting of the systems
$\mathcal{S}_1,\dots,\mathcal{S}_K$.
As we saw in the preceding subsection,
the independent measurements of the observables $A_1,\dots,A_K$ over the system
$\mathcal{S}_1,\dots,\mathcal{S}_K$, respectively, can be regarded a \emph{single measurement} over
the system $\mathcal{S}$ according to
the unitary time-evolution \eqref{dmt-single_measurementAkEm} as a measurement process.
In this identification, we
can
regard the $K$ apparatuses measuring the observable $A_1,\dots,A_K$
of $\mathcal{S}_1,\dots,\mathcal{S}_K$, respectively, as a \emph{single apparatus $\mathcal{A}$},
where the vector $\ket{\Phi_{A}^{\mathrm{init}}}$ is the initial state of the apparatus $\mathcal{A}$, and
$\ket{\Phi_{A}[m]}$ is the final state of the apparatus $\mathcal{A}$ for each $m\in\Omega$,
with $\braket{\Phi_{A}[m]}{\Phi_{A}[m']}=\delta_{m,m'}$.
Thus, the infinite repetition of the independent measurements of the observables $A_1,\dots,A_K$,
which we described above,
can be regarded as an infinite repetition of the
single
measurement
described by
the unitary time-evolution
\eqref{dmt-single_measurementAkEm}.
The latter infinite repetition
of measurements
is
exactly the one that we have investigated
over the course of
Subsection~\ref{FOMWI} and Section~\ref{POT-pure-states}.
Thus, we can apply the results of Subsection~\ref{FOMWI} and Section~\ref{POT-pure-states} to
the setting of measurements of this subsection.
Therefore, according to Definition~\ref{pmrpwst},
we see that
in the setting of this subsection
a \emph{world} is an infinite sequence over $\Omega$ and
the \emph{measure representation for the prefixes of worlds}
is given by a function $r\colon\Omega^*\to[0,1]$ with
\begin{equation*}
  r(m_1\dotsc m_n)=\prod_{k=1}^n\bra{\Psi}E_{m_k}\ket{\Psi},
\end{equation*}
just as in Subsection~\ref{FOMWI},
where $\ket{\Psi}$ denotes the state $\ket{\Psi_1}\otimes\dots\otimes\ket{\Psi_K}$.
Then, the probability measure induced by the probability measure representation $r$ is
a Bernoulli measure $\lambda_P$ on $\Omega^\infty$,
where $P$ is a finite probability space on $\Omega$ such that
\begin{equation}\label{dmt-Pm=ipMS}
  P(m)=\bra{\Psi}E_m\ket{\Psi}
\end{equation}
for every $m\in\Omega$,
also just as in Subsection~\ref{FOMWI}.

\subsection{Application of the principle of typicality resulting in
Postulate~\ref{COSMS}}
\label{dmt-APOT}

Now, let us apply Postulate~\ref{POT}, the principle of typicality, to the setting developed
above.

Let $\omega$ be our world in the infinite repetition of the independent measurements of
$A_1,\dots,A_K$
in the setting above.
Then $\omega$ is an infinite sequence over $\Omega$
due to the
arguments
of the preceding subsection. 
Since we can apply the results of
Section~\ref{POT-pure-states} to
this setting of measurements,
we have the following:
According to Postulate~\ref{POT} we have that
\emph{$\omega$ is Martin-L\"of $P$-random,
where $P$ is the finite probability space on $\Omega$ satisfying \eqref{dmt-Pm=ipMS}}.
Moreover,
for every $n\in\N^+$,
the state of the $n$th copy of
the composite system consisting of the systems $\mathcal{S}_1,\dots,\mathcal{S}_K$
immediately after the measurements of $A_1,\dots,A_K$ over them \emph{in our world} is given by
\begin{equation}\label{dmt-def-u-a}
  \frac{E_{\omega(n)}\ket{\Psi}}{\sqrt{\bra{\Psi}E_{\omega(n)}\ket{\Psi}}}
\end{equation}
with
$\bra{\Psi}E_{\omega(n)}\ket{\Psi}>0$.

We use $\omega_1,\dots,\omega_K$ to denote the infinite sequences over $\Omega_1,\dots,\Omega_K$,
respectively, such that $$(\omega_1(n),\dots,\omega_K(n))=\omega(n)$$ for every $n\in\N^+$.
Then we have
$\omega=\omega_1\times\dots\times\omega_K$,
obviously.
Let $k$ be an arbitrary integer with $1\le k\le K$ in the rest of this paragraph.
The sequence $\omega_k$ is \emph{the infinite sequence of records of the values of the observable $A_k$
in the corresponding apparatuses measuring $A_k$ in our world}.
Put briefly, the infinite sequence $\omega_k$ is \emph{the infinite sequence
of outcomes of the infinitely repeated measurements of the observable $A_k$
over the infinite copies of the system $\mathcal{S}_k$
prepared in
the state $\ket{\Psi_k}$ in our world}.
Moreover, since
$$E_{\omega(n)}\ket{\Psi}
=E_{1,\omega_1(n)}\ket{\Psi_1}\otimes\dots\otimes E_{K,\omega_K(n)}\ket{\Psi_K}$$
for every $n\in\N^+$,
it follows from \eqref{dmt-def-u-a} that,
for every $n\in\N^+$,
the state of the $n$th copy of
the system $\mathcal{S}_k$
immediately after the measurement of $A_k$ over it \emph{in our world} is given by
\begin{equation}\label{dmt-def-u-ak}
  \frac{E_{k,\omega_k(n)}\ket{\Psi_k}}{\sqrt{\bra{\Psi_k}E_{k,\omega_k(n)}\ket{\Psi_k}}}
\end{equation}
with $\bra{\Psi_k}E_{k,\omega_k(n)}\ket{\Psi_k}>0$.
On the other hand, since $\omega$ is Martin-L\"of $P$-random,
using a theorem which
is obtained by generalizing Theorem~\ref{contraction2} slightly
we have that $\omega_k$ is Martin-L\"of $P_k$-random,
where $P_k$ is a finite probability space on $\Omega_k$ such that
\[
  P_k(m_k):=
  \sum_{m_1\in\Omega_1}\dots\sum_{m_{k-1}\in\Omega_{k-1}}
  \sum_{m_{k+1}\in\Omega_{k+1}}\dots\sum_{m_K\in\Omega_K}
  P(m_1,\dots,m_K)
\]
for every $m_k\in\Omega_k$.
Since $\sum_{m\in\Omega_{k'}}E_{k',m}=I$ for every $k'$,
it follows from \eqref{dmt-Pm=ipMS} that
\begin{equation}\label{dmt-Pm=ipMSR}
  P_k(m)=\bra{\Psi_k}E_{k,m}\ket{\Psi_k}
\end{equation}
for every $m\in\Omega_k$.
Let $\overline{\Omega}_k$ be
an
alphabet $\{\omega_k(n)\mid n\in\N^+\}$.
Since $\omega_k$ is Martin-L\"of $P_k$-random, it follows from Theorem~\ref{nonzero-prob-appears}
that $\omega_k$ is a Martin-L\"of $\overline{P}_k$-random infinite sequence over $\overline{\Omega}_k$,
where $\overline{P}_k$ is a finite probability space on $\overline{\Omega}_k$ such that
\begin{equation}\label{ovlPkm=Pkm}
  \overline{P}_k(m)=P_k(m)
\end{equation}
for every $m\in\overline{\Omega}_k$.

Let $\overline{\Omega}$ be
an
alphabet $\{\omega(n)\mid n\in\N^+\}$.
Since $\omega$ is Martin-L\"of $P$-random, it follows from Theorem~\ref{nonzero-prob-appears}
that $\omega$ is a Martin-L\"of $\overline{P}$-random infinite sequence over $\overline{\Omega}$,
where $\overline{P}$ is a finite probability space on $\overline{\Omega}$ such that
\begin{equation}\label{ovlPm=Pm}
  \overline{P}(m)=P(m)
\end{equation}
for every $m\in\overline{\Omega}$.
On the other hand, using \eqref{dmt-Pm=ipMS} and \eqref{dmt-Pm=ipMSR} we see that
\begin{equation}\label{dmt-Pm=P1-PK}
  P(m_1,\dots,m_K)=P_1(m_1)\dots P_K(m_K)
\end{equation}
for every $(m_1,\dots,m_K)\in\Omega$.
Thus, it follows from Theorem~\ref{nonzero-prob-appears}
that
\begin{equation}\label{ovlO=ovlO1tdtovlOK}
  \overline{\Omega}=\overline{\Omega}_1\times\dots\times\overline{\Omega}_K.
\end{equation}
Therefore, using \eqref{ovlPkm=Pkm}, \eqref{ovlPm=Pm} and \eqref{dmt-Pm=P1-PK} we have
\begin{equation}\label{ovlP=ovlP1tdtovlPK}
  \overline{P}=\overline{P}_1\times\dots\times\overline{P}_K.
\end{equation}

Let $k$ be an arbitrary integer with $1\le k\le K$ in the rest of this paragraph.
For each $m\in\overline{\Omega}_k$, we use $\ket{\Theta_{k,m}}$ to denote the state
$$\frac{E_{k,m}\ket{\Psi_k}}{\sqrt{\bra{\Psi_k}E_{k,m}\ket{\Psi_k}}},$$
where the denominator is positive
due to Theorem~\ref{nonzero-prob-appears} and \eqref{dmt-Pm=ipMSR}.
Then we use $\Gamma_k$ to denote a set $\{\ket{\Theta_{k,\omega_k(n)}}\mid n\in\N^+\}$,
and use $\gamma_k$ to denote an infinite sequence over $\Gamma_k$ such that
$\gamma_k(n):=\ket{\Theta_{k,\omega_k(n)}}$ for every $n\in\N^+$.
Note that $\Gamma_k$ is an alphabet.
Thus, since $\omega_k$ is Martin-L\"of $\overline{P}_k$-random
and
the mapping $\overline{\Omega}_k\ni m\mapsto\ket{\Theta_{k,m}}$ is an injection,
it is easy to see that $\gamma_k$ is a Martin-L\"of $Q_k$-random infinite sequence over $\Gamma_k$,
where $Q_k$ is a finite probability space on $\Gamma_k$ such that
\begin{equation}\label{Qkgkn=ovlPkokn}
  Q_k(\ket{\Theta_{k,m}})=\overline{P}_k(m)
\end{equation}
for every $m\in\overline{\Omega}_k$.
Thus, according to (i) of Definition~\ref{def-mixed-state},
the infinite sequence $\gamma_k$ is a mixed state of the system $\mathcal{S}_k$.
Due to \eqref{dmt-def-u-ak},
the mixed state $\gamma_k$ can be
phrased
as
\emph{the infinite sequence of states of the system $\mathcal{S}_k$,
resulting from an infinite repetition of independent measurements of the observables $A_1,\dots,A_k$
over infinite copies of the composite system consisting of the systems $\mathcal{S}_1,\dots,\mathcal{S}_K$
prepared in the identical state $\ket{\Psi_1}\otimes\dots\otimes\ket{\Psi_K}$}.

Always keep in mind the equation \eqref{ovlO=ovlO1tdtovlOK}
in what follows.
Recall that $\omega$ is Martin-L\"of $\overline{P}$-random and
$\omega=\omega_1\times\dots\times\omega_K$.
Thus, since the mapping $\overline{\Omega}_k\ni m\mapsto\ket{\Theta_{k,m}}$ is an injection,
it is easy to see that
$\gamma_1\times\dots\times\gamma_K$ is Martin-L\"of $Q$-random,
where $Q$ is a finite probability space on $\Gamma_1\times\dots\times\Gamma_K$ such that
$$Q(\ket{\Theta_{1,m_1}},\dots,\ket{\Theta_{K,m_K}})=\overline{P}(m_1,\dots,m_K)$$
for every $(m_1,\dots,m_K)\in
\overline{\Omega}$.
It follows from
\eqref{ovlP=ovlP1tdtovlPK} and \eqref{Qkgkn=ovlPkokn} that
$$Q(\ket{\Theta_{1,m_1}},\dots,\ket{\Theta_{K,m_K}})
=(Q_1\times\dots\times Q_K)(\ket{\Theta_{1,m_1}},\dots,\ket{\Theta_{K,m_K}})$$
for every
$(m_1,\dots,m_K)\in
\overline{\Omega}_1\times\dots\times\overline{\Omega}_K$.
Thus, we have
$Q=Q_1\times\dots\times Q_K$, and therefore
$\gamma_1\times\dots\times\gamma_K$ is Martin-L\"of $Q_1\times\dots\times Q_K$-random.
Hence, since $\gamma_k$ is Martin-L\"of $Q_k$-random for every $k=1,\dots,K$,
according to Definition~\ref{independency-of-ensembles}
we have
that
$\gamma_1,\dots,\gamma_K$ are independent.

Now, let us apply Corollary~\ref{cor-COSMS-independence}.
It is easy to check that each of $\Gamma_1,\dots,\Gamma_K$ is a set of linearly independent vectors.
Thus, since the mixed states $\gamma_1,\dots,\gamma_K$ are independent,
it follows from the implication (ii) $\Rightarrow$ (i) of Corollary~\ref{cor-COSMS-independence} that
\emph{the infinite sequence $\gamma_1\otimes\dots\otimes\gamma_n$
over $\Gamma_1\otimes\dots\otimes\Gamma_K$ is a mixed state of
the composite system consisting of the systems $\mathcal{S}^1,\dots,\mathcal{S}^K$ and
its density matrix is $\rho_1\otimes\dots\otimes\rho_K$,
where $\rho_1,\dots,\rho_K$ are the density matrices of $\gamma_1,\dots,\gamma_K$, respectively}.
Hence,
the second statement of Postulate~\ref{COSMS} is confirmed
in this natural and simple scenario,
based on Postulate~\ref{POT}, the principle of typicality,
together with Postulates~\ref{state_space}, \ref{composition}, and \ref{evolution}.

Note that $\braket{\Theta_{k,m}}{\Theta_{k,m'}}=\delta_{m,m'}$
for every $k=1,\dots,K$ and every $m,m'\in\overline{\Omega}_k$.
Thus, each of the mixed states $\gamma_1,\dots,\gamma_K$ in the above scenario
is
an infinite sequence over \emph{mutually orthogonal} pure states.
According to (i) of Definition~\ref{def-mixed-state}, however,
a mixed state is commonly an infinite sequence over mutually \emph{non-orthogonal} pure states.
In a similar manner as we did in Section~\ref{POT-mixed-states2},
by introducing ancilla systems $\mathcal{S}^{\mathrm{a}}_1,\dots,\mathcal{S}^{\mathrm{a}}_K$
in addition to the original systems $\mathcal{S}_1,\dots,\mathcal{S}_K$,
and
performing independent measurements of the observables $A_1,\dots,A_k$
of $\mathcal{S}^{\mathrm{a}}_1,\dots,\mathcal{S}^{\mathrm{a}}_K$,
instead of $\mathcal{S}_1,\dots,\mathcal{S}_K$,
over infinite copies of the composite system consisting of
$\mathcal{S}^{\mathrm{a}}_1,\dots,\mathcal{S}^{\mathrm{a}}_K$ and
$\mathcal{S}_1,\dots,\mathcal{S}_K$,
where for every $k=1,\dots,K$
we prepare each copy of the composite system consisting of
$\mathcal{S}_k$ and $\mathcal{S}^{\mathrm{a}}_k$ in an appropriate entangled state,
we make each of the mixed state $\gamma_1,\dots,\gamma_K$ into a general form.
We can then show that the second statement of Postulate~\ref{COSMS} is confirmed
in this generalized setting, as well,
based on Postulate~\ref{POT}
together with Postulates~\ref{state_space}, \ref{composition}, and \ref{evolution}.

\section{Clarification of the postulate for the ``probabilistic mixture'' of mixed states}
\label{CPFPM}

In this section we
investigate
the validity of
Postulate~\ref{density-matrices-probability} below
of the conventional quantum mechanics,
in terms of our framework based on Postulate~\ref{POT}, the principle of typicality, and
Definition~\ref{def-mixed-state}.
Postulate~\ref{density-matrices-probability}
is
the last
part of Postulate 1 described in Nielsen and Chuang \cite[Section 2.4.2]{NC00},
and treats the ``probabilistic mixture'' of mixed states.

\begin{postulate}\label{density-matrices-probability}
If a quantum system is in the state $\rho_k$ with probability $p_k$,
then the density operator for the system is
\begin{equation}\label{sumkpkrhok}
  \sum_{k} p_k\rho_k.
\end{equation}
\qed
\end{postulate}

First of all,
Postulate~\ref{density-matrices-probability} seems very vague in its original form.
What does it mean that \emph{a quantum system is in the state $\rho_k$ with probability $p_k$}?
What does the word ``probability'' mean
here?
What does the ``probabilistic mixture'' of mixed states mean?
In what follows,
we give a certain precise meaning to Postulate~\ref{density-matrices-probability}
by
means of
giving an appropriate scenario in which Postulate~\ref{density-matrices-probability} clearly holds,
based on the principle of typicality.
Namely, we derive Postulate~\ref{density-matrices-probability}
from the principle of typicality
together with Postulates~\ref{state_space}, \ref{composition}, and \ref{evolution}
in a certain natural and simple scenario regarding the setting of measurements.

For deriving Postulate~\ref{density-matrices-probability} based on the principle of typicality,
we consider an infinite repetition of the following two successive measurements,
in a similar manner as we did in the preceding sections:
The first measurement is described by observable $A$
and is performed
over an identical system $\mathcal{S}_A$ prepared in an identical initial state $\ket{\Psi_A}$.
Then,
the
second
measurement is described by one of observables $B_1,\dots,B_K$
and is performed
over an identical system $\mathcal{S}_B$,
different from $\mathcal{S}_A$,
prepared in an identical initial state $\ket{\Psi_B}$.
In the second measurement, an observable $B_k$ is chosen among the observables $B_1,\dots,B_K$
depending on the outcome $k$ of the first measurement of $A$,
and then the measurement of $B_k$ is performed over $\mathcal{S}_B$.
Put briefly, in order to realize the
statement
of Postulate~\ref{density-matrices-probability} based on the principle of typicality,
\emph{the first measurement is performed to get a value $k$ ``with probability'' $p_k$, and
according to the value $k$ an observable $B_k$
is chosen 
and its measurement is performed to generate a mixed state with the density matrix $\rho_k$}.
As a whole, \emph{we
generate
a mixed state with the density matrix $\sum_{k=1}^K p_k\rho_k$},
as in Postulate~\ref{density-matrices-probability}.
In this scenario, for simplicity,
the mixed state with each density matrix $\rho_k$
is
an infinite sequence over \emph{mutually orthogonal} pure states.%
\footnote{In a similar manner as we did in Section~\ref{POT-mixed-states2},
based on the introduction of an ancilla system into the original systems
$\mathcal{S}_A$ and $\mathcal{S}_B$,
we can make all mixed states with the density matrices $\rho_1,\dots,\rho_K$
an arbitrary mixed state which is
an
infinite sequence over \emph{not necessarily mutually orthogonal} pure states.}
The precise description of the setting of measurements is given
in what follows.

\subsection{Repeated once of measurements}
\label{dmpROM}

First, the repeated once of the infinite repetition of
measurements consists of
the measurement of the observable $A$ and
its subsequent measurement of one of the observables $B_1,\dots,B_K$,
and is described in detail as follows.

Let $\mathcal{S}_A$ be an arbitrary quantum system with state space $\mathcal{H}_A$
of finite dimension $K$,
and let $\mathcal{S}_B$ be another arbitrary quantum system with state space $\mathcal{H}_B$
of finite dimension $L$.
Consider arbitrary measurement over $\mathcal{S}_A$ described by observable $A$
and arbitrary $K$ measurements over $\mathcal{S}_B$ described by observables $B_1,\dots,B_K$,
respectively.
Let $\Omega$ be the spectrum of $A$.
For simplicity, we assume that $\Omega=\{1,\dots,K\}$.
Then, the observable $A$ is non-degenerate, and has a spectral decomposition of the form 
\begin{equation*}
  A=\sum_{k=1}^K k\ket{\psi_k}\bra{\psi_k},
\end{equation*}
where $\{\ket{\psi_1},\dots,\ket{\psi_K}\}$ is an orthonormal basis of
the state space $\mathcal{H}_A$ of $\mathcal{S}_A$.
On the other hand, for each
$k\in\Omega$,
let $\Theta_k$ be the spectrum of $B_k$ and let
\begin{equation}\label{dmpspectral-decompositionB}
  B_k=\sum_{\ell=1}^L g_k(\ell)\ket{\phi_{k,\ell}}\bra{\phi_{k,\ell}}
\end{equation}
be a spectral decomposition of the observable $B_k$,
where $\{\ket{\phi_{k,1}},\dots,\ket{\phi_{k,L}}\}$ is an orthonormal basis of
the state space $\mathcal{H}_B$ of $\mathcal{S}_B$ and
$g_k\colon\{1,\dots,L\}\to\Theta_k$ is a surjection.
Let
$$\Theta:=\bigcup_{k=1}^K\Theta_k.$$
According to Postulates~\ref{state_space}, \ref{composition}, and \ref{evolution},
the measurement processes of the observables $A$ and $B_1,\dots,B_K$ are described
by the following unitary operators $U_A$ and $U_{B_1},\dots,U_{B_K}$, respectively:
\begin{equation} \label{dmpsingle_measurementA}
  U_A\ket{\psi_k}\otimes\ket{\Phi_A^{\mathrm{init}}}=\ket{\psi_k}\otimes\ket{\Phi_A[k]},
\end{equation}
for every
$k\in\Omega$,
and
\begin{equation}  \label{dmpsingle_measurementB}
  U_{B_k}\ket{\phi_{k,\ell}}\otimes\ket{\Phi_B^{\mathrm{init}}}=\ket{\phi_{k,\ell}}\otimes\ket{\Phi_B[g_k(\ell)]}
\end{equation}
for every
$k\in\Omega$
and $\ell=1,\dots,L$.
The vector $\ket{\Phi_A^{\mathrm{init}}}$ is the initial state of the apparatus measuring $A$,
and $\ket{\Phi_A[k]}$ is a final state of
one
for each $k\in\Omega$, with $\braket{\Phi_A[k]}{\Phi_A[k']}=\delta_{k,k'}$.
On the other hand,
the vector $\ket{\Phi_B^{\mathrm{init}}}$ is the initial state of the apparatus measuring one
of $B_1,\dots,B_K$, depending on $k$,
and $\ket{\Phi_B[l]}$ is a final state of
one
for each $l\in\Theta$, with $\braket{\Phi_B[l]}{\Phi_B[l']}=\delta_{l,l'}$.
For every $k\in\Omega$, the state $\ket{\Phi_A[k]}$ indicates that
\emph{the apparatus measuring the observable $A$ of the system $\mathcal{S}_A$
records the value $k$ of $A$}, and
for every $l\in\Theta$, the state $\ket{\Phi_B[l]}$ indicates that
\emph{the apparatus measuring one of the observables $B_1,\dots,B_K$ of the system $\mathcal{S}_B$,
depending on $k$,
records the value $l$}.

For each $k\in\Omega$, let $E_k$ be the projector onto the eigenspace of $A$ with eigenvalue $k$,
and for each
$k\in\Omega$
and $l\in\Theta_k$, let $F_{k,l}$ be the projector onto the eigenspace of $B_k$ with eigenvalue $l$.
Then, the equalities~\eqref{dmpsingle_measurementA} and \eqref{dmpsingle_measurementB} are rewritten,
respectively, as the forms that
\begin{equation*}%
  U_A\ket{\Psi'}\otimes\ket{\Phi_A^{\mathrm{init}}}
  =\sum_{k\in\Omega}(E_k\ket{\Psi'})\otimes\ket{\Phi_A[k]}
\end{equation*}
for every $\ket{\Psi'}\in\mathcal{H}_A$, and
\begin{equation*}%
  U_{B_k}\ket{\Psi}\otimes\ket{\Phi_B^{\mathrm{init}}}
  =\sum_{l\in\Theta_k}(F_{k,l}\ket{\Psi})\otimes\ket{\Phi_B[l]}
\end{equation*}
for every $k\in\Omega$ and $\ket{\Psi}\in\mathcal{H}_B$.

Moreover,
depending on the outcome $k$ of the measurement of $A$,
the specific observable $B_k$
is chosen
among
the observables $B_1,\dots,B_K$ of the system $\mathcal{S}_B$
and then its measurement is performed over $\mathcal{S}_B$ in the second measurement.
This choice is realized by the following unitary time-evolution $U_c$:
\begin{equation*}
  U_c\ket{\Psi}\otimes\ket{\Phi_A[k]}=(U_{B_k}\ket{\Psi})\otimes\ket{\Phi_A[k]}
\end{equation*}
for every $k\in\Omega$ and
every state $\ket{\Psi}$ of the composite system consisting of the system $\mathcal{S}_B$ and
the apparatus measuring one of the observables $B_1,\dots,B_K$.
Note that the unitarity of $U_c$ is confirmed by
the following theorem.%
\footnote{%
As we mentioned previously,
the state space of an apparatus \emph{commonly} has infinite dimension.
\emph{For simplicity}, however, we here prove Theorem~\ref{unitarity}
by assuming that both the dimension of $\mathcal{H}_1$, which corresponds to the state space of
the composite system consisting of the system $\mathcal{S}_B$ and
the apparatus measuring one of the observables $B_1,\dots,B_K$,
and the dimension of $\mathcal{H}_2$,
which corresponds to the state space of the apparatus measuring $A$,
are finite.
We can
prove
a similar theorem to Theorem~\ref{unitarity},
which confirms the unitarity of $U_c$,
in the case where
both $\mathcal{H}_1$ and $\mathcal{H}_2$ have infinite dimension.
\label{footnote-unitarity}}

\begin{theorem}\label{unitarity}
Let $\mathcal{H}_1$ and $\mathcal{H}_2$ be
complex Hilbert spaces of finite dimension.
Let $\{\ket{\Psi_1},\dots,\ket{\Psi_N}\}$ and $\{\ket{\Phi_1},\dots,\ket{\Phi_N}\}$
be arbitrary two orthonormal bases of $\mathcal{H}_2$,
and let $U_1,\dots,U_N$ be arbitrary $N$ unitary operators on $\mathcal{H}_1$.
Then $U:=U_1\otimes\ket{\Psi_1}\bra{\Phi_1} + \dots + U_N\otimes\ket{\Psi_N}\bra{\Phi_N}$ is
a unitary operator on $\mathcal{H}_1\otimes \mathcal{H}_2$,
and $U(\ket{\Theta}\otimes\ket{\Phi_k})=(U_k\ket{\Theta})\otimes\ket{\Psi_k}$
for every $\ket{\Theta}\in\mathcal{H}_1$ and every $k=1,\dots,N$.
\end{theorem}

\begin{proof}
First, we see that
\begin{align*}
U^\dag U
&=\left(\sum_{k=1}^N U_k^\dag\otimes\ket{\Phi_k}\bra{\Psi_k}\right)
\left(\sum_{l=1}^N U_l\otimes\ket{\Psi_l}\bra{\Phi_l}\right)
=\sum_{k=1}^N \sum_{l=1}^N U_k^\dag U_l\otimes\ket{\Phi_k}\braket{\Psi_k}{\Psi_l}\bra{\Phi_l} \\
&=\sum_{k=1}^N U_k^\dag U_k\otimes\ket{\Phi_k}\bra{\Phi_k}
=I_1\otimes\left(\sum_{k=1}^N \ket{\Phi_k}\bra{\Phi_k}\right)=I_1\otimes I_2 \\
&=I,
\end{align*}
where $I_1$, $I_2$, and $I$ are the identity operators
on
$\mathcal{H}_1$, $\mathcal{H}_2$, and
$\mathcal{H}_1\otimes\mathcal{H}_2$, respectively.
Similarly, we can show that $UU^\dag =I$.
This completes the proof.
\end{proof}

Then, the repeated once of the infinite repetition of
measurements
is realized by
the sequential application of $U_A$ and $U_c$ to the composite system
consisting of the system $\mathcal{S}_A$ in a state $\ket{\Psi_A}$,
the system $\mathcal{S}_B$ in a state $\ket{\Psi_B}$,
and the two apparatuses in the states $\ket{\Phi_A^{\mathrm{init}}}$ and $\ket{\Phi_B^{\mathrm{init}}}$
respectively.
This sequential application results in the following single unitary time-evolution $U$:
\begin{equation}\label{dmproirm}
  U\ket{\Psi_A}\otimes\ket{\Psi_B}\otimes\ket{\Phi_A^{\mathrm{init}}}\otimes\ket{\Phi_B^{\mathrm{init}}}
  =\sum_{k\in\Omega}\sum_{l\in\Theta_k}
  (E_k\ket{\Psi_A})\otimes(F_{k,l}\ket{\Psi_B})\otimes\ket{\Phi_A[k]}\otimes\ket{\Phi_B[l]}
\end{equation}

\subsection{\boldmath Infinite repetition of the two measurements of $A$ and one of $B_1,\dots,B_K$}

Similarly in the preceding sections,
we consider countably infinite copies of each of the systems $\mathcal{S}_A$ and $\mathcal{S}_B$,
and each of the two apparatuses measuring the observable $A$ and
one of the observables $B_1,\dots,B_K$.
Let $\ket{\Psi_A}$ and $\ket{\Psi_B}$ are arbitrary states of the systems
$\mathcal{S}_A$ and $\mathcal{S}_B$, respectively.
We prepare each of copies of the composite system consisting of
the systems
$\mathcal{S}_A$ and $\mathcal{S}_B$
and the two apparatuses
in the identical state
$\ket{\Psi_A}\otimes\ket{\Psi_B}\otimes\ket{\Phi_A^{\mathrm{init}}}\otimes\ket{\Phi_B^{\mathrm{init}}}$,
and then perform \emph{one by one}
the two successive measurements of the observable $A$ and one of the observables $B_1,\dots,B_K$
over each of the copies of the systems $\mathcal{S}_A$ and $\mathcal{S}_B$, respectively,
by interacting each of the copies of the systems $\mathcal{S}_A$ and $\mathcal{S}_B$ with
the corresponding copy of the two apparatuses
according to the unitary time-evolution \eqref{dmproirm}.
For each $n\in\N^+$,
let $\mathcal{H}_n$ be the state space of the total system consisting of
the first $n$ copies of the two systems $\mathcal{S}_A$ and $\mathcal{S}_B$ and
the two apparatuses measuring $A$ and one of $B_1,\dots,B_K$.
These
successive interactions
between the copies of the two systems $\mathcal{S}_A$ and $\mathcal{S}_B$
and the copies of the two apparatuses measuring $A$ and one of $B_1,\dots,B_K$,
as measurement processes,
proceed in the following manner.

Let $n$ be an arbitrary positive integer.
We focus the $n$th measurement of the observable $A$ and its subsequent measurement of one of
the observables $B_1,\dots,B_K$.
Immediately before the measurement of $A$ over the $n$th copy of the system $\mathcal{S}_A$,
the state of the total system,
which consists of the first $n$ copies of the two systems $\mathcal{S}_A$ and $\mathcal{S}_B$ and
the two apparatuses measuring $A$ and one of $B_1,\dots,B_K$, is
\begin{align*}
  \sum_{k_1,\dots,k_{n-1}\in\Omega}&\sum_{l_1\in\Theta_{k_1}}\dots\sum_{l_{n-1}\in\Theta_{k_{n-1}}}
  (E_{k_1}\ket{\Psi_A})\otimes(F_{k_1,l_1}\ket{\Psi_B})\otimes\dots
  (E_{k_{n-1}}\ket{\Psi_A})\otimes(F_{k_{n-1},l_{n-1}}\ket{\Psi_B}) \\
  &\otimes\ket{\Psi_A}\otimes\ket{\Psi_B} \\
  &\otimes\ket{\Phi_A[k_1]}\otimes\ket{\Phi_B[l_1]}\otimes\dots
  \otimes\ket{\Phi_A[k_{n-1}]}\otimes\ket{\Phi_B[l_{n-1}]}
  \otimes\ket{\Phi_A^{\mathrm{init}}}\otimes\ket{\Phi_B^{\mathrm{init}}}
\end{align*}
in $\mathcal{H}_n$.
Then, immediately after the measurement of one of $B_1,\dots,B_K$
over the $n$th copy of the system $\mathcal{S}_B$, the total system results in the state
\begin{align}
  \sum_{k_1,\dots,k_{n}\in\Omega}\sum_{l_1\in\Theta_{k_1}}\dots\sum_{l_{n}\in\Theta_{k_{n}}}&
  (E_{k_1}\ket{\Psi_A})\otimes(F_{k_1,l_1}\ket{\Psi_B})\otimes\dots\otimes
  (E_{k_{n}}\ket{\Psi_A})\otimes(F_{k_{n},l_{n}}\ket{\Psi_B}) \nonumber \\
  &\otimes\ket{\Phi_A[k_1]}\otimes\ket{\Phi_B[l_1]}\otimes\dots
  \otimes\ket{\Phi_A[k_{n}]}\otimes\ket{\Phi_B[l_{n}]} \label{dmptotal_systemAB} \\
  & \nonumber \\
  =\sum_{k_1,\dots,k_{n}\in\Omega}\sum_{l_1\in\Theta_{k_1}}\dots\sum_{l_{n}\in\Theta_{k_{n}}}&
  (E_{k_1}\ket{\Psi_A})\otimes(F_{k_1,l_1}\ket{\Psi_B})\otimes\dots\otimes
  (E_{k_{n}}\ket{\Psi_A})\otimes(F_{k_{n},l_{n}}\ket{\Psi_B}) \nonumber \\
  &\otimes\ket{\Phi[(k_1,l_1)\dots(k_{n},l_{n})]} \label{dmptotal_systemAB2}
\end{align}
in $\mathcal{H}_n$, according to \eqref{dmproirm}.
Here
the state $\ket{\Phi[(k_1,l_1)\dots (k_n,l_n)]}$ denotes
$\ket{\Phi_A[k_1]}\otimes\ket{\Phi_B[l_1]}\otimes\dots
\otimes\ket{\Phi_A[k_{n}]}\otimes\ket{\Phi_B[l_{n}]}$, and indicates that
\emph{the first $n$ copies of the two apparatuses measuring $A$ and one of $B_1,\dots,B_K$ record
the values $(k_1,l_1)\dots (k_n,l_n)$ as measurement
results}.

Then, for applying Postulate~\ref{POT}, the principle of typicality,
we have to identify worlds and 
the
measure representation for the prefixes of
worlds
in this situation.
Note that the unitary time-evolution \eqref{dmproirm}
resulting from the sequential application of $U_A$ and $U_c$ is rewritten as the form that: 
\begin{equation*}%
  U\ket{\Psi}\otimes\ket{\Phi_A^{\mathrm{init}}}\otimes\ket{\Phi_B^{\mathrm{init}}}
  =\sum_{k\in\Omega}\sum_{l\in\Theta_k}
  (E_k\otimes F_{k,l})\ket{\Psi}\otimes\ket{\Phi[(k,l)]}
\end{equation*}
for every $\ket{\Psi}\in\mathcal{H}_1\otimes\mathcal{H}_2$.
It is then easy to check that a collection
$\{E_k\otimes F_{k,l}\}_{(k,l)\in\Lambda}$
forms measurement operators,
where $\Lambda:=\{(k,l)\mid k\in\Omega\;\&\;l\in\Theta_k\}$.
Thus, the successive measurements of $A$ and one of $B_1,\dots,B_k$,
which is chosen depending on the outcome of the measurement of $A$,
can be regarded
as the \emph{single measurement} which is described by the measurement operators
$\{E_k\otimes F_{k,l}\}_{(k,l)\in\Lambda}$.
Hence, we can apply Definition~\ref{pmrpwst}
to this scenario of the setting of measurements.
Therefore, according to Definition~\ref{pmrpwst},
we see that a \emph{world} is an infinite sequence over $\Lambda$ and
the
\emph{measure representation for the prefixes of
worlds}
is given by a function $r\colon\Lambda^*\to[0,1]$ with
$$r((k_1,l_1)\dotsc (k_n,l_n))
:=\prod_{i=1}^n \bra{\Psi_A}E_{k_i}\ket{\Psi_A}\bra{\Psi_B}F_{k_i,l_i}\ket{\Psi_B},$$
which is the square of the norm of each state
$$(E_{k_1}\ket{\Psi_A})\otimes(F_{k_1,l_1}\ket{\Psi_B})\otimes\dots\otimes
(E_{k_{n}}\ket{\Psi_A})\otimes(F_{k_{n},l_{n}}\ket{\Psi_B})\otimes\ket{\Phi[(k_1,l_1)\dots(k_{n},l_{n})]}$$
in the superposition~\eqref{dmptotal_systemAB2}.
It follows from \eqref{mr} that
the probability measure induced by the probability measure representation $r$ is
a Bernoulli measure $\lambda_P$ on $\Lambda^\infty$,
where $P$ is a finite probability space on $\Lambda$ such that
\begin{equation}\label{dmpPm=ipMS}
  P(k,l)=\bra{\Psi_A}E_{k}\ket{\Psi_A}\bra{\Psi_B}F_{k,l}\ket{\Psi_B}
\end{equation}
for every $(k,l)\in\Lambda$.

\subsection{Application of the principle of typicality}
\label{dmpAPOT}

Now, we
apply Postulate~\ref{POT} to the setting developed above.

Let $\omega$ be our world in the infinite repetition of the measurements of
$A$ and one of $B_1,\dots,B_K$
in the above setting.
Then $\omega$ is an infinite sequence over $\Lambda$.
Since the Bernoulli measure $\lambda_P$ on $\Lambda^\infty$ is
the probability measure induced by the
measure representation $r$
for the prefixes of
worlds
in the above setting,
it follows from Postulate~\ref{POT} that \emph{$\omega$ is Martin-L\"of random
with respect to the measure $\lambda_P$ on $\Lambda^\infty$}.
Therefore, \emph{$\omega$ is Martin-L\"of $P$-random,
where $P$ is the finite probability space on $\Lambda$ satisfying \eqref{dmpPm=ipMS}}.

Let $\alpha$ and $\beta$ be infinite sequences over $\Omega$ and $\Theta$, respectively,
such that $(\alpha(n),\beta(n))=\omega(n)$ for every $n\in\N^+$.
The infinite sequence $\alpha$ is
\emph{the infinite sequence of records of the values of the observable $A$
in the corresponding apparatuses measuring $A$ in our world}.
Put briefly, the infinite sequence $\alpha$ is \emph{the infinite sequence
of outcomes of the infinitely repeated measurements of the observable $A$
over the infinite copies of $\mathcal{S}_A$ in our world}.
On the other hand, the infinite sequence $\beta$ is
\emph{the infinite sequence of records of the values of one of the observables $B_1,\dots,B_K$
in the corresponding apparatuses measuring it in our world}.
Put briefly, the infinite sequence $\beta$ is \emph{the infinite sequence
of outcomes of the infinitely repeated measurements of one of the observables $B_1,\dots,B_K$
over the infinite copies of $\mathcal{S}_B$ in our world}.

Using Theorem~\ref{nonzero-prob-appears}, we note that $P(\omega(n))>0$ for every $n\in\N^+$.
Thus, it follows from \eqref{dmpPm=ipMS} that
\begin{equation}\label{dmpnon-zero-inner-productMS}
  \bra{\Psi_A}E_{\alpha(n)}\ket{\Psi_A}\bra{\Psi_B}F_{\alpha(n),\beta(n)}\ket{\Psi_B}>0
\end{equation}
for every $n\in\N^+$.

On the other hand,
using Theorem~\ref{contraction} repeatedly,
we have that $\alpha$ is Martin-L\"of $Q$-random, where
$Q$ is a finite probability space on $\Omega$
such that $$Q(k):=\sum_{l\in\Theta_k}P(k,l)$$ for every $k\in\Omega$.
Since
$\sum_{l\in\Theta_k} F_{k,l}=I$ for every $k\in\Omega$ and $\ket{\Psi_B}$ is a unit vector,
it follows from \eqref{dmpPm=ipMS} that
\begin{equation}\label{dmpPm=ipMSQ}
  Q(k)=\bra{\Psi_A}E_{k}\ket{\Psi_A}
\end{equation}
for every $k\in\Omega$,
as expected from the point of view of the conventional quantum mechanics.



\subsection{\boldmath Mixed state resulting from the measurements of one of $B_1,\dots,B_K$}
\label{Msrfmooob1dbk}

We calculate the mixed state resulting from the measurements of one of the observables $B_1,\dots,B_K$.
This mixed state is a mixed state
generating the density matrix \eqref{sumkpkrhok} in Postulate~\ref{density-matrices-probability}.

Let $n$ be an arbitrary positive integer.
In the superposition~\eqref{dmptotal_systemAB} of the total system consisting of
the first $n$ copies of the two systems $\mathcal{S}_A$ and $\mathcal{S}_B$ and
the two apparatuses measuring $A$ and one of $B_1,\dots,B_K$
immediately after the measurement of one of $B_1,\dots,B_K$
over the $n$th copy of the system $\mathcal{S}_B$, consider the specific state
\begin{equation}\label{dmpspecific_stateMSAB}
\begin{split}
  &(E_{k_1}\ket{\Psi_A})\otimes(F_{k_1,l_1}\ket{\Psi_B})\otimes\dots
  (E_{k_{n}}\ket{\Psi_A})\otimes(F_{k_{n},l_{n}}\ket{\Psi_B}) \\
  &\otimes\ket{\Phi_A[k_1]}\otimes\ket{\Phi_B[l_1]}\otimes\dots
  \otimes\ket{\Phi_A[k_{n}]}\otimes\ket{\Phi_B[l_{n}]}
\end{split}
\end{equation}
such that $(k_{i},l_{i})=\omega(i)$ for every $i=1,\dots,n$.
Due to \eqref{dmpnon-zero-inner-productMS},
we note here that the vector \eqref{dmpspecific_stateMSAB} is non-zero and therefore can be
a state vector certainly, when it is normalized.
Since $\omega$ is our world, the state \eqref{dmpspecific_stateMSAB} is
\emph{the state of the total system,
consisting of the first $n$ copies of the two systems $\mathcal{S}_A$ and $\mathcal{S}_B$
and the two apparatuses measuring $A$ and one of $B_1,\dots,B_K$,
that we have perceived after the $n$th measurement of one of $B_1,\dots,B_K$}.
Therefore, \emph{we perceive}
$$\frac{F_{k_n,l_n}\ket{\Psi_B}}{\sqrt{\bra{\Psi_B}F_{k_n,l_n}\ket{\Psi_B}}}$$
as the state of the $n$th copy of the system $\mathcal{S}_B$ at this time. 
Note that $m_n=\alpha(n)$ and $l_n=\beta(n)$, and $n$ is arbitrary.
Thus, for every positive integer $n$, the state of the $n$th copy of the system $\mathcal{S}_B$
immediately after the measurement of one of $B_1,\dots,B_K$ over it \emph{in our world} is given by
\begin{equation}\label{dmpmrMSB}
  \ket{\Upsilon^B_n}
  :=\frac{F_{\alpha(n),\beta(n)}\ket{\Psi_B}}{\sqrt{\bra{\Psi_B}F_{\alpha(n),\beta(n)}\ket{\Psi_B}}}.
\end{equation}

Let $\Gamma^B:=\{\ket{\Upsilon^B_n}\mid n\in\N^+\}$,
and let $\gamma_B$ be an infinite sequence over $\Gamma^B$ such that
$\gamma_B(n):=\ket{\Upsilon^B_n}$ for every $n\in\N^+$.
Note that $\Gamma^B$ is an alphabet.
For each $\ket{\psi}\in\Gamma^B$, we define $S(\ket{\psi})$ as the set
$$\{(\alpha(n),\beta(n))\mid\ket{\psi}=\ket{\Upsilon^B_n}\}.$$
Since $\omega$ is Martin-L\"of $P$-random,
it follows from Theorem~\ref{nonzero-prob-appears} and Theorem~\ref{contraction}
that the infinite sequence $\gamma_B$ is Martin-L\"of $P'$-random,
where $P'$ is a finite probability space on $\Gamma^B$ such that
\begin{equation}\label{dmpP'U=Pa}
  P'(\ket{\psi})=\sum_{x\in S(\ket{\psi})} P(x)
\end{equation}
for every $\ket{\psi}\in\Gamma^B$.
Thus, according to (i) of Definition~\ref{def-mixed-state},
the infinite sequence $\gamma_B$ is a mixed state of $\mathcal{S}_B$.
The mixed state $\gamma_B$ can be
interpreted
as
\emph{the infinite sequence of states of $\mathcal{S}_B$
resulting from the infinitely repeated measurements of one of the observables $B_1,\dots,B_K$
over the infinite copies of $\mathcal{S}_B$}.

Then, according to (ii) of Definition~\ref{def-mixed-state},
the density matrix $\rho_B$ of $\gamma_B$ is given by
\begin{equation*}
  \rho_B:=\sum_{\ket{\psi}\in\Gamma^B}P'(\ket{\psi})\ket{\psi}\bra{\psi}.
\end{equation*}
Note from \eqref{dmpP'U=Pa} and \eqref{dmpPm=ipMS} that
$$P'(\ket{\psi})=\sum_{(k,l)\in S(\ket{\psi})}\bra{\Psi_A}E_{k}\ket{\Psi_A}\bra{\Psi_B}F_{k,l}\ket{\Psi_B}$$
for every $\ket{\psi}\in\Gamma^B$.
Thus, since $\omega$ is Martin-L\"of $P$-random,
using Theorem~\ref{nonzero-prob-appears} and \eqref{dmpmrMSB} we have that
\begin{equation*}
  \rho_B=\sum_{(k,l):P(k,l)>0}\bra{\Psi_A}E_{k}\ket{\Psi_A}F_{k,l}\ket{\Psi_B}\bra{\Psi_B}F_{k,l},
\end{equation*}
where the sum is over all $(k,l)\in\Lambda$ such that $P(k,l)>0$.
It follows from \eqref{dmpPm=ipMS} that
\begin{equation*}
  \rho_B=\sum_{(k,l)\in\Lambda}\bra{\Psi_A}E_{k}\ket{\Psi_A}F_{k,l}\ket{\Psi_B}\bra{\Psi_B}F_{k,l}.
\end{equation*}
Hence, using \eqref{dmpPm=ipMSQ} we have that
\begin{equation}\label{dmprB=s=ss}
  \rho_B=\sum_{k\in\Omega}Q(k)\sum_{l\in\Theta_k}F_{k,l}\ket{\Psi_B}\bra{\Psi_B}F_{k,l}.
\end{equation}

\subsection{\boldmath Mixed state conditioned by a specific outcome of the measurements of $A$}
\label{dmpMSCSOMB}

We calculate
the
mixed state
of the system $\mathcal{S}_B$
resulting from the measurements of one of the observables $B_1,\dots,B_K$,
conditioned by a specific outcome $k$ of the measurement of $A$, i.e.,
the
mixed state
of the system $\mathcal{S}_B$
resulting only from the measurements of a specific observable $B_k$.
This mixed state is
a mixed state
which generates
each density matrix $\rho_k$ in Postulate~\ref{density-matrices-probability}.

Let $\overline{\Omega}:=\{\alpha(n)\mid n\in\N^+\}$,
which is the set of all possible measurement outcomes of $A$ in our world.
Let $k_0$ be an arbitrary element of
$\overline{\Omega}$,
and let $C:=\{(k_0,l)\mid l\in\Theta_{k_0}\}$.
Since $\omega$ is Martin-L\"of $P$-random,
it follows from Theorem~\ref{nonzero-prob-appears} that $P(C)>0$.
Let $\delta:=\cond{C}{\omega}$.
Recall that $\cond{C}{\omega}$ is defined as an infinite sequence in $C^\infty$ obtained from
$\omega$ by eliminating all elements of $\Lambda\setminus C$ occurring in $\omega$.
In other words, $\delta$ is the subsequence of $\omega$
such that the outcomes of measurements of $A$ equal to $k_0$.
We
assume that for each $\tau\in\N^+$,
the $\tau$th element of $\delta$ is originally the $n_{\tau}$th element of $\omega$ before the elimination.
It follows that $\delta(\tau)=(k_0,\beta(n_{\tau}))$ for every $\tau\in\N^+$.
Since $\omega$ is Martin-L\"of $P$-random again,
using Theorem~\ref{conditional_probability} we have that
$\delta$ is Martin-L\"of $P_C$-random for the finite probability space $P_C$ on $C$.

Let $\zeta$ be an infinite sequence over $\Theta_{k_0}$ such that
$\zeta(\tau)=\beta(n_{\tau})$ for every $\tau\in\N^+$.
The
sequence $\zeta$ can be phrased as \emph{the infinite sequence of outcomes of the
measurements of the observable $B_{k_0}$
over the infinite copies of
the system
$\mathcal{S}_B$ in our world,
where any outcomes of the measurements of  all $B_1,\dots,B_K$ other than $B_{k_0}$ are ignored}.
Since $\delta$ is Martin-L\"of $P_C$-random,
it is easy to show that $\zeta$ is Martin-L\"of $P_{k_0}$-random,
where $P_{k_0}$ is a finite probability space on $\Theta_{k_0}$ such that
$P_{k_0}(l)=P_C(k_0,l)$ for every $l\in\Theta_{k_0}$.
We see that
\begin{equation}\label{dmpPl0m=tr}
  P_{k_0}(l)=P_C(k_0,l)=\frac{P(k_0,l)}{\sum_{l'\in\Theta_{k_0}}P(k_0,l')}
  =\bra{\Psi_B}F_{k_0,l}\ket{\Psi_B}
\end{equation}
for each $l\in\Theta_{k_0}$, where the last equality follows from \eqref{dmpPm=ipMS} and
$\sum_{l'\in\Theta_{k_0}}F_{k_0,l'}=I$.

Let $\mu_{k_0}$ is a subsequence sequence  $\gamma_B$ such that
$\mu_{k_0}(\tau)=\gamma_B(n_{\tau})$ for every $\tau\in\N^+$.
Then, it follows from \eqref{dmpmrMSB} that
\begin{equation}\label{dmpmu=k=FEP}
  \mu_{k_0}(\tau)=\ket{\Upsilon^B_{n_{\tau}}}
  =\frac{F_{k_0,\zeta(\tau)}\ket{\Psi_B}}{\sqrt{\bra{\Psi_B}F_{k_0,\zeta(\tau)}\ket{\Psi_B}}}
\end{equation}
for every $\tau\in\N^+$.
Let $\Xi^B:=\{\mu_{k_0}(\tau)\mid \tau\in\N^+\}$. Note that $\Xi^B$ is an alphabet.
For each $\ket{\psi}\in\Xi^B$, we define $T(\ket{\psi})$ as the set
$\{\zeta(\tau)\mid\ket{\psi}=\mu_{k_0}(\tau)\}$.
Since $\zeta$ is Martin-L\"of $P_{k_0}$-random,
it follows from
Theorem~\ref{nonzero-prob-appears} and Theorem~\ref{contraction}
that the infinite sequence $\mu_{k_0}$ is Martin-L\"of $P_{k_0}'$-random,
where $P_{k_0}'$ is a finite probability space on $\Xi^B$ such that
\begin{equation}\label{dmpPl0'mk=Pl0zk}
  P_{k_0}'(\ket{\psi})=\sum_{l\in T(\ket{\psi})}P_{k_0}(l)
\end{equation}
for every $\ket{\psi}\in\Xi^B$.
Thus, according to (i) of Definition~\ref{def-mixed-state},
the infinite sequence $\mu_{k_0}$ is a mixed state of $\mathcal{S}_B$.
The mixed state $\mu_{k_0}$
can be phrased as
\emph{the infinite sequence of states of the system $\mathcal{S}_B$,
resulting from the infinitely repeated measurements of the observable $B_{k_0}$
over the infinite copies of the system $\mathcal{S}_B$ in our world,
where any resulting states from the measurements of  all $B_1,\dots,B_K$ other than $B_{k_0}$
are ignored}.

Then, according to (ii) of Definition~\ref{def-mixed-state},
the density matrix $\rho_{k_0}$ of $\mu_{k_0}$ is given by
\begin{equation*}
  \rho_{k_0}:=\sum_{\ket{\psi}\in\Xi^B}P_{k_0}'(\ket{\psi})\ket{\psi}\bra{\psi}.
\end{equation*}
Thus, since $\zeta$ is Martin-L\"of $P_{k_0}$-random,
using Theorem~\ref{nonzero-prob-appears}, \eqref{dmpmu=k=FEP}, and \eqref{dmpPl0'mk=Pl0zk}
we have that
\begin{equation*}
  \rho_{k_0}
  =\sum_{l:P_{k_0}(l)>0}
  P_{k_0}(l)\frac{F_{k_0,l}\ket{\Psi_B}\bra{\Psi_B}F_{k_0,l}}{\bra{\Psi_B}F_{l_0,k}\ket{\Psi_B}},
\end{equation*}
where the sum is over all $l\in\Theta_{k_0}$ such that $P_{k_0}(l)>0$.
Note from \eqref{dmpPl0m=tr} that, for every $l\in\Theta_{k_0}$,
$P_{k_0}(l)=0$ if and only if $F_{k_0,l}\ket{\Psi_B}=0$.
Hence, using \eqref{dmpPl0m=tr} we finally have that
\begin{equation*}
  \rho_{k_0}=\sum_{l\in\Theta_{k_0}}F_{k_0,l}\ket{\Psi_B}\bra{\Psi_B}F_{k_0,l}.
\end{equation*}

Recall that $k_0$ is an arbitrary element of $\Omega$.
Hence, in summary, we see that for every $k\in\Omega$,
\begin{equation}\label{dmppost-measuremnt-state}
  \rho_{k}:=\sum_{l\in\Theta_{k}}F_{k,l}\ket{\Psi_B}\bra{\Psi_B}F_{k,l}
\end{equation}
is the density matrix of
\emph{the mixed state
of the system $\mathcal{S}_B$,
resulting from the infinitely repeated measurements of the observable $B_{k}$
over the infinite copies of the system $\mathcal{S}_B$
in our world,
where any resulting states from the measurements of  all $B_1,\dots,B_K$ other than $B_{k_0}$
are ignored}.
The result is just as expected from the aspect of the conventional quantum mechanics.

\subsection{Derivation of Postulate~\ref{density-matrices-probability}}

For deriving Postulate~\ref{density-matrices-probability},
we have developed so far a natural and simple scenario regarding the setting of measurements
through the preceding subsections.
At last, from \eqref{dmprB=s=ss} and \eqref{dmppost-measuremnt-state} we certainly have the equation
\begin{equation}\label{dmpMain}
  \rho_B=\sum_{k\in\Omega}Q(k)\rho_{k},
\end{equation}
which has the same form as the equation \eqref{sumkpkrhok} in
Postulate~\ref{density-matrices-probability}.
Hence, we have derived Postulate~\ref{density-matrices-probability} under our scenario, as desired.
The precise meaning of the quantities $\rho_B$, $Q$, and $\rho_k$
appearing in the equation \eqref{dmpMain} are described as
above.

In this manner, we have given a
precise meaning to Postulate~\ref{density-matrices-probability}
based on
the principle of typicality
together with Postulates~\ref{state_space}, \ref{composition}, and \ref{evolution},
through
a
natural and simple scenario regarding the setting of measurements.

\section{Application to the BB84 quantum key distribution protocol}
\label{BB84QKD}

In this section, we make an application of our framework based on the principle of typicality,
to the BB84 quantum key distribution protocol \cite{BB84}
\emph{in order to demonstrate how properly our framework works in practical problems
in quantum mechanics}.

The BB84 quantum key distribution (QKD) \cite{BB84} is a protocol
by which ``random'' Private classical bits can be shared between two parties via
quantum and classical communications over a public channel and
which has an ability to detect the presence of eavesdropping.
The original BB84 QKD protocol \cite{BB84} starts by preparing and then sending
a \emph{block} of qubits from Alice to Bob,
and the subsequent processes between Alice and Bob are done
on a block basis of qubits.
For the simplicity of the analysis,
we consider the following slight modification of the original BB84 protocol \cite{BB84},
Protocol~\ref{BB84wsm} below,
where Alice sends
single qubits
to Bob one by one.
See Nielsen and Chuang \cite[Section 2.6]{NC00}
for the detail of the analysis of the original BB84 protocol.

Let $\ket{0}$ and $\ket{1}$ be an orthonormal basis of the state space of a single qubit system.
Based on
them
we define four states $\ket{\Psi_{00}},\ket{\Psi_{10}},\ket{\Psi_{01}}$,
and $\ket{\Psi_{11}}$
of a single qubit
system
as follows.
\begin{align*}
  \ket{\Psi_{00}}:&=\ket{0}, \\
  \ket{\Psi_{10}}:&=\ket{1}, \\
  \ket{\Psi_{01}}:&=(\ket{0}+\ket{1})/\sqrt{2}, \\
  \ket{\Psi_{11}}:&=(\ket{0}-\ket{1})/\sqrt{2}.
\end{align*}

\begin{protocol}[The BB84 QKD protocol with slight modifications]\label{BB84wsm}
Initially, set
$\mathrm{flag}:=0$.
Repeat the following procedure
\emph{forever}.
{
\renewcommand{\labelenumi}{Step \arabic{enumi}:}
\setlength{\leftmargini}{40pt} 
\begin{enumerate}
\item Alice tosses two fair coins $A$ and $B$ to get outcomes $a$ and $b$ in $\{0,1\}$, respectively.
\item Alice prepares $\ket{\Psi_{ab}}$ and sends it to Bob.
\item Bob tosses a fair coin $C$ to get outcome $c\in\{0,1\}$.
\item Bob performs the measurement of the observable $\ket{\Psi_{1c}}\bra{\Psi_{1c}}$
  over the state $\ket{\Psi_{ab}}$ sent from Alice, to obtain outcome $m\in\{0,1\}$.
\item Bob tosses a biased coin $D$ to get outcome $d\in\{0,1\}$, where $\Prob\{D=1\}=p$.
\item If $\mathrm{flag}=1$, Alice and Bob discard all the bits obtained so far, and then go to Step 1.
\item Alice and Bob announce $b$ and $c$, respectively.
\item If $b\neq c$, Alice and Bob discard $a$ and $m$ and then go to Step~1.
\item If $d=0$, Alice and Bob keep $a$ and $m$, respectively, \emph{as a shared random bit},
  and then go to Step~1.
\item Alice and Bob announce $a$ and $m$, respectively.
\item If $a\neq m$,  Alice and Bob set $\mathrm{flag}:=1$.
\item Alice and Bob discard $a$ and $m$.\qed
\end{enumerate}
\renewcommand{\labelenumi}{(\roman{enumi})}
}
\end{protocol}

First, the
variable
$\mathrm{flag}$ in Protocol~\ref{BB84wsm} is
a ``global variable'' which  both Alice and Bob can access and
which cannot be altered by an eavesdropper.
On the one hand, only Steps 1--5 involve quantum mechanics in Protocol~\ref{BB84wsm}.
We repeat this part of the protocol infinitely often in any case,
regardless of the presence of eavesdropping.
The real $p$ with $0<p<1$ in Step 5 is a \emph{security parameter}
which
determines
the ability to detect the presence of eavesdropping.
As $p$ becomes larger, the ability of the detection of eavesdropping increases
but the efficiency for sharing
random
classical bits between Alice and Bob decreases.
On the other hand, Steps 6--12 of Protocol~\ref{BB84wsm} are just a classical procedure:
Steps 7--9 of Protocol~\ref{BB84wsm} are a
procedure
to complete
sharing a random bit between Alice and Bob.
Steps 10--12 of Protocol~\ref{BB84wsm} are needed to detect the eavesdropping when it is present.
If the presence of the eavesdropping is detected during executing these steps, in particular, at Step 11,
then the
global variable
$\mathrm{flag}$
is set to $1$ and
fixed on that value forever afterward.
Once the variable $\mathrm{flag}$ is set to $1$,
due to
Step 6 the subsequent sharing of random bits is aborted,
in addition to the rejection of the random bits shared so far between Alice and Bob.

In what follows, we analyze Protocol~\ref{BB84wsm}
in terms of our framework based on Postulate~\ref{POT}, the principle of typicality,
together with Postulates~\ref{state_space}, \ref{composition}, and \ref{evolution}.
To complete this,
\emph{we have to implement everything} in Steps 1--5 of Protocol~\ref{BB84wsm}
\emph{by unitary time-evolution}.

\subsection{Analysis of the protocol without eavesdropping}

First, we investigate Protocol~\ref{BB84wsm}
in the case where there is no eavesdropping during the execution of the protocol.
We analyze Steps 1--5 of Protocol~\ref{BB84wsm} in terms of
our framework.
We realize each of the coin flippings in Steps 1, 3, and 5
by a measurement of the observable $\ket{1}\bra{1}$ over an appropriate qubit state.
We then describe all the measurement processes during Steps 1--5
as a unitary interaction between a system and an apparatus,
as in the scenarios regarding the setting of measurements, considered in the preceding sections.
Thus, each of Steps 1--5 is implemented by a unitary time-evolution
in the following manner:

\paragraph{Unitary implementation of Step 1.}

To realize the coin tossing
by Alice
in Step 1 of Protocol~\ref{BB84wsm}
we make use of a measurement over a two qubit system.
Namely, to implement the Step 1
we
introduce
a two qubit system $\mathcal{Q}_1$ with state space $\mathcal{H}_1$,
and perform a measurement over the system $\mathcal{Q}_1$ described
by a unitary time-evolution:
$$U_1\ket{ab}\otimes\ket{\Phi_1^{\mathrm{init}}}=\ket{ab}\otimes\ket{\Phi_1[ab]}$$
for every $a,b\in\{0,1\}$,
where $\ket{ab}:=\ket{a}\otimes\ket{b}\in\mathcal{H}_1$.
The vector $\ket{\Phi_1^{\mathrm{init}}}$ is the initial state of an apparatus
$\mathcal{A}_1$
measuring $\mathcal{Q}_1$,
and $\ket{\Phi_1[ab]}$ is a final state of
the apparatus
$\mathcal{A}_1$
for each $a,b\in\{0,1\}$.%
\footnote{We assume, of course, the orthogonality of the final states $\ket{\Phi_1[ab]}$, i.e.,
the property that
$\braket{\Phi_1[ab]}{\Phi_1[a'b']}=\delta_{a,a'}\delta_{b,b'}$.
Furthermore,
we assume the orthogonality of the finial states for
each of all apparatuses which appear in the rest of this section.}
Prior to the measurement, the system $\mathcal{Q}_1$ is prepared in the state
$\ket{\Psi_{01}}\otimes\ket{\Psi_{01}}$.

\paragraph{Unitary implementation of Step 2.}

For the preparation of the state $\ket{\Psi_{ab}}$ by Alice in Step 2,
we
introduce
a single qubit system $\mathcal{Q}_2$ of state space $\mathcal{H}_2$.
Then the preparation is realized by the following unitary time-evolution
$U_2$
of the composite system consisting of
the system $\mathcal{Q}_2$ and the apparatus
$\mathcal{A}_1$:
$$U_2\ket{0}\otimes\ket{\Phi_1[ab]}=\ket{\Psi_{ab}}\otimes\ket{\Phi_1[ab]}$$
for every $a,b\in\{0,1\}$,
where the system $\mathcal{Q}_2$ is initially prepared in the state $\ket{0}$.
To be
precise, the unitary operator $U_2$ is defined by the equation:
$$U_2\ket{\Psi}\otimes\ket{\Phi_1[ab]}=(W_{ab}\ket{\Psi})\otimes\ket{\Phi_1[ab]}$$
for every $a,b\in\{0,1\}$ and every $\ket{\Psi}\in\mathcal{H}_2$.
Here, $W_{ab}$ is a unitary operator on $\mathcal{H}_2$ defined by
$$W_{ab}:=\ket{\Psi_{ab}}\bra{0}+\ket{\Psi_{\bar{a}b}}\bra{1}$$
for each $a,b\in\{0,1\}$, where $\bar{a}:=1-a$.
Note that the unitarity of $U_2$ is confirmed by Theorem~\ref{unitarity}.%
\footnote{The remark given in Footnote~\ref{footnote-unitarity} also applies here.
The state space of an apparatus commonly has infinite dimension.
Thus, to be precise, the unitarity of $U_2$ is confirmed
by a theorem which is obtained by an immediate
generalization
of Theorem~\ref{unitarity} over a Hilbert space of infinite dimension,
which corresponds to the state space of the apparatus $\mathcal{A}_1$.}

\paragraph{Unitary implementation of Step 3.}

Similarly to Step 1, the coin tossing by Bob in Step 3 is implemented
by
introducing
a single qubit system $\mathcal{Q}_3$ with state space $\mathcal{H}_3$,
and performing a measurement over the system $\mathcal{Q}_3$ described
by a unitary time-evolution:
$$U_3\ket{c}\otimes\ket{\Phi_3^{\mathrm{init}}}=\ket{c}\otimes\ket{\Phi_3[c]}$$
for every $c\in\{0,1\}$.
The vector $\ket{\Phi_3^{\mathrm{init}}}$ is the initial state of an apparatus
$\mathcal{A}_3$
measuring $\mathcal{Q}_3$,
and $\ket{\Phi_3[c]}$ is a final state of the apparatus
$\mathcal{A}_3$
for each $c\in\{0,1\}$.
Prior to the measurement, the system $\mathcal{Q}_3$ is prepared in the state $\ket{\Psi_{01}}$.

\paragraph{Unitary implementation of Step 4.}

The switching of the two types of measurements in Step 4, depending on the outcome $c$ in Step 3,
is realized by
a unitary time-evolution:
\begin{equation}\label{U4ToP3c=VcToP3c}
  U_4\ket{\Theta}\otimes\ket{\Phi_3[c]}=(V_c\ket{\Theta})\otimes\ket{\Phi_3[c]}
\end{equation}
for every $c\in\{0,1\}$
and every state $\ket{\Theta}$ of the composite system
consisting of the system $\mathcal{Q}_2$ and the apparatus $\mathcal{A}_4$ explained below.
The operator $V_c$ appearing in \eqref{U4ToP3c=VcToP3c} describes
a unitary time-evolution of the composite system consisting of
the system $\mathcal{Q}_2$ and an apparatus $\mathcal{A}_4$ measuring $\mathcal{Q}_2$, and is
defined by the equation:
$$V_c\ket{\Psi_{mc}}\otimes\ket{\Phi_4^{\mathrm{init}}}=\ket{\Psi_{mc}}\otimes\ket{\Phi_4[m]}$$
for every $c,m\in\{0,1\}$.
The vector $\ket{\Phi_4^{\mathrm{init}}}$ is the initial state of the apparatus $\mathcal{A}_4$,
and $\ket{\Phi_4[m]}$ is a final state of the apparatus $\mathcal{A}_4$ for each $m\in\{0,1\}$.
Thus,
the operator $V_c$
describes the alternate measurement process of the qubit $\mathcal{Q}_2$ sent from Alice,
depending on the outcome $c$, on Bob's side.

\paragraph{Unitary implementation of Step 5.}

Finally, similarly to Steps 1 and 3, the biased coin tossing by Bob in Step 5 is implemented
by
introducing
a single qubit system $\mathcal{Q}_5$ with state space $\mathcal{H}_5$,
and performing a measurement over the system $\mathcal{Q}_5$ described
by a unitary time-evolution:
$$U_5\ket{d}\otimes\ket{\Phi_5^{\mathrm{init}}}=\ket{d}\otimes\ket{\Phi_5[d]}$$
for every $d\in\{0,1\}$.
The vector $\ket{\Phi_5^{\mathrm{init}}}$ is the initial state of an apparatus $\mathcal{A}_5$
measuring $\mathcal{Q}_5$,
and $\ket{\Phi_5[d]}$ is a final state of the apparatus $\mathcal{A}_5$ for each $d\in\{0,1\}$.
Prior to the measurement, the system $\mathcal{Q}_5$ is prepared in the state
$\sqrt{1-p}\ket{0}+\sqrt{p}\ket{1}$, instead of $\ket{\Psi_{01}}$.

\bigskip

\medskip

Now, the sequential application of $U_1,\dots,U_5$ to the composite system
consisting of the five qubit system $\mathcal{Q}_1, \mathcal{Q}_2, \mathcal{Q}_3$,
and $\mathcal{Q}_5$ and the apparatuses $\mathcal{A}_1,\mathcal{A}_3,\mathcal{A}_4$,
and $\mathcal{A}_5$
results in the following single unitary time-evolution $U$:
\begin{equation}\label{BB84withoutEve-roirm}
  U\ket{\Psi}\otimes\ket{\Phi^{\mathrm{init}}}
  =\sum_{(a,b,c,m,d)\in\{0,1\}^5}
  \left(\left(E^1_{ab}\otimes E^4_{mc}W_{ab} \otimes E^3_{c}\otimes E^5_d\right)\ket{\Psi}\right)\otimes
  \ket{\Phi[(a,b,c,m,d)]}
\end{equation}
for every $\ket{\Psi}\in\mathcal{H}_1\otimes\mathcal{H}_2\otimes\mathcal{H}_3\otimes\mathcal{H}_5$.
Here,
$E^1_{ab}:=\ket{ab}\bra{ab}, E^3_{c}:=\ket{c}\bra{c}, E^4_{mc}:=\ket{\Psi_{mc}}\bra{\Psi_{mc}}$, and
$E^5_{d}:=\ket{d}\bra{d}$,
and moreover
$\ket{\Phi^{\mathrm{init}}}$ denotes
$\ket{\Phi_1^{\mathrm{init}}}\otimes\ket{\Phi_3^{\mathrm{init}}}\otimes\ket{\Phi_4^{\mathrm{init}}}
\otimes\ket{\Phi_5^{\mathrm{init}}}$
and $\ket{\Phi[(a,b,c,m,d)]}$ denotes
$\ket{\Phi_1[ab]}\otimes\ket{\Phi_3[c]}\otimes\ket{\Phi_4[m]}\otimes\ket{\Phi_5[d]}$
for each $(a,b,c,m,d)\in\{0,1\}^5$.
Totally, prior to the application of $U$, the five qubit system
$\mathcal{Q}_1, \mathcal{Q}_2, \mathcal{Q}_3$, and $\mathcal{Q}_5$ is prepared in the state
$$\ket{\Psi^{\mathrm{init}}}:=\sum_{(a,b,c,d)\in\{0,1\}^4}\frac{1}{\sqrt{8}}A_d\ket{ab0cd},$$
where $A_0:=\sqrt{1-p}$ and $A_1:=\sqrt{p}$,
and $\ket{ab0cd}$ denotes the five qubit state
$\ket{a}\otimes\ket{b}\otimes\ket{0}\otimes\ket{c}\otimes\ket{d}$
for each $(a,b,c,d)\in\{0,1\}^4$.

In Protocol~\ref{BB84wsm}, only Steps 1--5 involve quantum mechanics, and
this part of the protocol is repeated infinitely often in any case.
The operator $U$ above describes the repeated once of
this
infinite repetition of the Steps 1-5.
We use $\Omega$ to denote the alphabet $\{0,1\}^5$.
It is
easy to check that a collection
\begin{equation}\label{BB84withoutevemo}
  \left\{E^1_{ab}\otimes E^4_{mc}W_{ab} \otimes E^3_{c}\otimes E^5_d\right\}_{(a,b,c,m,d)\in\Omega}
\end{equation}
forms measurement operators.
Thus,
the sequential application $U$ of $U_1,\dots,U_5$
can be regarded as the \emph{single measurement} which is described by
the measurement operators \eqref{BB84withoutevemo}.
Hence, we can apply Definition~\ref{pmrpwst}
to this scenario
of
the setting of measurements.
Therefore, according to Definition~\ref{pmrpwst},
we can see that a \emph{world} is an infinite sequence over
$\Omega$
and the probability measure induced by
the
\emph{measure representation for the prefixes of
worlds}
is
a Bernoulli measure $\lambda_P$ on
$\Omega^\infty$,
where $P$ is a finite probability space on
$\Omega$
such that
$P(a,b,c,m,d)$ is the square of the norm of
the
state
\begin{equation*}%
  \left(\left(E^1_{ab}\otimes E^4_{mc}W_{ab} \otimes E^3_{c}\otimes E^5_d\right)
  \ket{\Psi^{\mathrm{init}}}\right)\otimes\ket{\Phi[(a,b,c,m,d)]}
\end{equation*}
for every $(a,b,c,m,d)\in\Omega$.
Let us calculate the
explicit
form of $P(a,b,c,m,d)$.
The
direct
application of $U_1,\dots,U_5$ in this order to the state
$\ket{ab0cd}\otimes\ket{\Phi^{\mathrm{init}}}$
leads to the following:
\begin{enumerate}
\item
In the case of $b=c$, we have
$$U\ket{ab0cd}\otimes\ket{\Phi^{\mathrm{init}}}
=\ket{ab}\otimes\ket{\Psi_{ab}}\otimes\ket{cd}\otimes\ket{\Phi[(a,b,c,a,d)]}.$$
\item
In the case of $b\neq c$, we have
\begin{align*}
&U\ket{ab0cd}\otimes\ket{\Phi^{\mathrm{init}}} \\
&=\frac{1}{\sqrt{2}}\ket{ab}\otimes \ket{\Psi_{0c}}\otimes\ket{cd}\otimes\ket{\Phi[(a,b,c,0,d)]}
+\frac{(-1)^a}{\sqrt{2}}\ket{ab}\otimes\ket{\Psi_{1c}}\otimes\ket{cd}\otimes\ket{\Phi[(a,b,c,1,d)]}.
\end{align*}
\end{enumerate}
The comparison of this result with the equation \eqref{BB84withoutEve-roirm},
where $\ket{\Psi}$ is set to $\ket{\Psi^{\mathrm{init}}}$,
leads to the following:
\begin{enumerate}
\item In the case of $b=c$, we have that
  \begin{equation}\label{b=c-a=mPabcmd=Ad2-8}
    P(a,b,c,m,d)=\frac{A_d^2}{8}\delta_{a,m}.
  \end{equation}
\item In the case of $b\neq c$, we have that
  \begin{equation*}
    P(a,b,c,m,d)=\frac{A_d^2}{16}.
  \end{equation*}
\end{enumerate}

Now, let us apply Postulate~\ref{POT}, the principle of typicality, to the setting of measurements
developed above.
Let $\omega$ be our world in the infinite repetition of the measurements in the above setting.
Then $\omega$ is an infinite sequence over $\Omega$.
Since the Bernoulli measure $\lambda_P$ on $\Omega^\infty$ is
the probability measure induced by the
measure representation
for the prefixes of
worlds
in the above setting,
it follows from Postulate~\ref{POT} that
\emph{$\omega$ is Martin-L\"of $P$-random}.

Let $S:=\{(a,b,c,m,d)\in\Omega\mid b=c\;\&\; d=0\}$.
It is
the set of all records of
the apparatuses $\mathcal{A}_1,\mathcal{A}_3,\mathcal{A}_4$, and $\mathcal{A}_5$,
in a repeated once of the procedure in Protocol~\ref{BB84wsm}, where
\emph{Bob measured a qubit in the same basis as Alice prepared it
and moreover the detection of eavesdropping was not
attempted}.
It follows from \eqref{b=c-a=mPabcmd=Ad2-8}
that
\begin{equation}\label{PS=sumPf1-p2}
  P(S)=\sum_{(a,b,m)\in\{0,1\}^3} P(a,b,b,m,0)=\frac{1-p}{2}.
\end{equation}

\emph{Intuitively}, the real $P(S)$ can be interpreted as \emph{the probability that
Alice and Bob complete sharing a single random bit}.
The explicit form $(1-p)/2$ of $P(S)$ given in \eqref{PS=sumPf1-p2}
is just as expected from the point of view of the conventional quantum mechanics.
Recall here that the real $p$
is the security parameter of the protocol
which determines the ability to detect the presence of eavesdropping.
From \eqref{PS=sumPf1-p2} we can \emph{intuitively} see that,
as the security parameter $p$ becomes larger, the ``efficiency'' for sharing
random
bits between Alice and Bob decreases in Protocol~\ref{BB84wsm}.

Actually,
since $\omega$ is Martin-L\"of $P$-random, using Theorem~\ref{charaA} we can prove that
$\chara{S}{\omega}$ is Martin-L\"of $\charaps{P}{S}$-random.
Recall here that
$\chara{S}{\omega}$ is an infinite \emph{binary} sequence obtained from $\omega$
by replacing each element $\omega(n)$ occurring in $\omega$ by $1$ if $\omega(n)\in S$ and by $0$
otherwise.
In other words,
$\chara{S}{\omega}$
is \emph{an infinite
binary
sequence obtained from $\omega$
by replacing each records of the apparatuses
$\mathcal{A}_1,\mathcal{A}_3,\mathcal{A}_4$, and $\mathcal{A}_5$, occurring in $\omega$,
by $1$ if the records
indicate the completion of sharing of a single random bit between Alice and Bob, and by $0$ otherwise}.
Note
that $(\charaps{P}{S})(1)=P(S)=(1-p)/2$ due to \eqref{PS=sumPf1-p2}.
Therefore, it follows from Theorem~\ref{FI} that
\begin{equation}\label{limfNnn=f1-p2}
  \lim_{n\to\infty} \frac{N(n)}{n}=\frac{1-p}{2},
\end{equation}
where $N(n)$ is the number of the occurrences of $1$ in the first $n$ bits of
the infinite binary sequence $\chara{S}{\omega}$.
In other words,
$N(n)$ is the number of random bits shared between Alice and Bob,
as a consequence of
the first $n$ repetitions of the procedure in Protocol~\ref{BB84wsm}
\emph{in our world}.
Thus, the result \eqref{limfNnn=f1-p2} can be interpreted as that \emph{the
``efficiency''
that
Alice and Bob complete sharing a single random bit is equal to $(1-p)/2$ in our world}.

Since $0<p<1$, note from \eqref{PS=sumPf1-p2} that $P(S)>0$, in particular.
Let $\alpha:=\cond{S}{\omega}$.
It is an infinite sequence over $S$ obtained from $\omega$
by eliminating all elements of the form $(a,b,c,m,d)$ with $b\neq c$ or $d=1$.
Since $\omega$ is Martin-L\"of $P$-random,
using Theorem~\ref{conditional_probability} we have that
$\alpha$ is Martin-L\"of $P_S$-random for the finite probability space $P_S$ on $S$.
From \eqref{PS=sumPf1-p2} and \eqref{b=c-a=mPabcmd=Ad2-8}
we see that
\begin{equation}\label{BB84withoutEvePC}
  P_S(a,b,c,m,d)=\frac{P(a,b,b,m,0)}{P(S)}
  =\frac{\delta_{a,m}}{4}
\end{equation}
for each $(a,b,c,m,d)\in S$.
Therefore,
$P_S(a,b,c,m,d)=0$ for every $(a,b,c,m,d)\in S$ with $a\neq m$.
It follows from Theorem~\ref{one_probability} that
$\alpha$ consists only of elements of the form $(a,b,b,a,0)$.
This shows that
\emph{Alice and Bob certainly share an identical bit every time of the case of $b=c$ and $d=0$
in our world},
as expected.
Let $\beta$ be an infinite \emph{binary} sequence obtained from $\alpha$
by replacing each element $(a,b,b,a,0)$ occurring in $\alpha$ by $a$.
Using Theorem~\ref{contraction2} it is easy to show that $\beta$ is Martin-L\"of $U$-random,
where $U$ is a finite probability space on $\{0,1\}$ such that
\[
  U(a):=\sum_{b\in\{0,1\}}P_S(a,b,b,a,0)
\]
for every $a\in\{0,1\}$.
It follows from \eqref{BB84withoutEvePC} that $U(0)=U(1)=1/2$.
Thus, $\beta$ is \emph{Martin-L\"of random}, i.e.,
$\beta$ is Martin-L\"of random with respect to Lebesgue measure $\mathcal{L}$ on $\XI$,
where
\emph{Lebesgue measure $\mathcal{L}$ on $\XI$} satisfies that
$$\mathcal{L}\left(\osg{\sigma}\right)=2^{-\abs{\sigma}}$$
for every $\sigma\in\X$.
Recall that Martin-L\"of randomness is one of the major definitions of the notion of \emph{randomness}
for an infinite binary sequence.
Thus, the Martin-L\"of randomness of
$\beta$
means that \emph{Alice and Bob certainly share a ``random'' infinite binary sequence
in our world},
as desired and expected.

\subsection{Analysis of the protocol under the presence of eavesdropping}

Next, we investigate Protocol~\ref{BB84wsm}
in the case where there is an eavesdropping by Eve throughout the execution of the protocol.
We assume that Eve performs the following eavesdropping
between Step 2 and Step 3 of Protocol~\ref{BB84wsm}.

{
\renewcommand{\labelenumi}{Step E\arabic{enumi}:}
\setlength{\leftmargini}{40pt} 
\begin{enumerate}
\item Eve tosses a fair coin $E$ to get outcome $e\in\{0,1\}$.
\item Eve performs the measurement of the observable $\ket{\Psi_{1e}}\bra{\Psi_{1e}}$
  over the state $\ket{\Psi_{ab}}$ sent from Alice to Bob, and
  then
  obtains outcome $f\in\{0,1\}$.
\end{enumerate}
\renewcommand{\labelenumi}{(\roman{enumi})}
}

We analyze Steps 1--5 of Protocol~\ref{BB84wsm}
added with Steps E1 and E2 above,
in terms of our framework based on the principle of typicality.
As in the preceding subsection, we have to implement everything
in all these steps
by unitary time-evolution.
The unitary implementation of the Steps 1--5 is the same as given in the preceding subsection.
On the other hand,
Steps E1 and E2 are implemented by a unitary time-evolution
as in the same manner as Steps 3 and 4 of Protocol~\ref{BB84wsm}, respectively.
The detail is given as follows:


\paragraph{Unitary implementation of Step E1.}

The coin tossing by Eve in Step E1 is implemented
by
introducing
a single qubit system $\mathcal{Q}_{\mathrm{E1}}$ with state space $\mathcal{H}_{\mathrm{E1}}$,
and performing a measurement over the system $\mathcal{Q}_{\mathrm{E1}}$ described
by a unitary time-evolution:
$$U_{\mathrm{E1}}\ket{e}\otimes\ket{\Phi_{\mathrm{E1}}^{\mathrm{init}}}
=\ket{e}\otimes\ket{\Phi_{\mathrm{E1}}[e]}$$
for every $e\in\{0,1\}$.
The vector $\ket{\Phi_{\mathrm{E1}}^{\mathrm{init}}}$ is the initial state of an apparatus
$\mathcal{A}_{\mathrm{E1}}$
measuring $\mathcal{Q}_{\mathrm{E1}}$,
and $\ket{\Phi_{\mathrm{E1}}[e]}$ is a final state of the apparatus
$\mathcal{A}_{\mathrm{E1}}$
for each $e\in\{0,1\}$.
Prior to the measurement, the system $\mathcal{Q}_{\mathrm{E1}}$ is
prepared in the state $\ket{\Psi_{01}}$.

\paragraph{Unitary implementation of Step E2.}

The switching of the two types of measurements in Step E2, depending on the outcome $e$ in Step E1,
is realized by
a unitary time-evolution:
\begin{equation}\label{UEToPEc=VcToPEc}
  U_{\mathrm{E2}}\ket{\Theta}\otimes\ket{\Phi_{\mathrm{E1}}[e]}=
  \left(\widetilde{V_e}\ket{\Theta}\right)\otimes\ket{\Phi_{\mathrm{E1}}[e]}
\end{equation}
for every $e\in\{0,1\}$
and every state $\ket{\Theta}$ of the composite system
consisting of the system $\mathcal{Q}_2$ and the apparatus $\mathcal{A}_{\mathrm{E2}}$ explained below.
The operator $\widetilde{V_e}$ appearing in \eqref{UEToPEc=VcToPEc} describes
a unitary time-evolution of the composite system consisting of
the system $\mathcal{Q}_2$ and an apparatus $\mathcal{A}_{\mathrm{E2}}$ measuring $\mathcal{Q}_2$,
and is defined by the equation:
$$\widetilde{V_e}\ket{\Psi_{fe}}\otimes\ket{\Phi_{\mathrm{E2}}^{\mathrm{init}}}=
\ket{\Psi_{fe}}\otimes\ket{\Phi_{\mathrm{E2}}[f]}$$
for every $e,f\in\{0,1\}$.
The vector $\ket{\Phi_{\mathrm{E2}}^{\mathrm{init}}}$ is the initial state of
the apparatus $\mathcal{A}_{\mathrm{E2}}$, and $\ket{\Phi_{\mathrm{E2}}[f]}$ is
a final state of the apparatus $\mathcal{A}_{\mathrm{E2}}$ for each $f\in\{0,1\}$.

\bigskip

\medskip

Now, the sequential application of $U_1,U_2,U_{\mathrm{E1}},U_{\mathrm{E2}},U_3,U_4,U_5$
to the composite system consisting of the six qubit system
$\mathcal{Q}_1, \mathcal{Q}_2,  \mathcal{Q}_{\mathrm{E1}}, \mathcal{Q}_3$,
and $\mathcal{Q}_5$ and the apparatuses
$\mathcal{A}_1,\mathcal{A}_{\mathrm{E1}},\mathcal{A}_{\mathrm{E2}},\mathcal{A}_3,\mathcal{A}_4$,
and $\mathcal{A}_5$
results in the following single unitary time-evolution $U_{\mathrm{T}}$:
\begin{equation}\label{BB84withEve-roirm}
\begin{split}
  U_{\mathrm{T}}\ket{\Psi}\otimes\ket{\Phi_{\mathrm{T}}^{\mathrm{init}}}
  =&\sum_{(a,b,e,f,c,m,d)\in\{0,1\}^7}
  \left(\left(E^1_{ab}\otimes E^4_{mc}E^{\mathrm{E2}}_{fe}W_{ab} \otimes E^{\mathrm{E1}}_{e}\otimes
  E^3_{c}\otimes E^5_d\right)\ket{\Psi}\right) \\
  &\hspace*{30mm}\otimes\ket{\Phi_{\mathrm{T}}[(a,b,e,f,c,m,d)]}
\end{split}
\end{equation}
for every
$\ket{\Psi}\in\mathcal{H}_1\otimes\mathcal{H}_2\otimes\mathcal{H}_{\mathrm{E1}}\otimes
\mathcal{H}_3\otimes\mathcal{H}_5$.
Here,
$E^{\mathrm{E1}}_{e}:=\ket{e}\bra{e}$ and $E^{\mathrm{E2}}_{fe}:=\ket{\Psi_{fe}}\bra{\Psi_{fe}}$,
and moreover
$\ket{\Phi_{\mathrm{T}}^{\mathrm{init}}}$ denotes
$$\ket{\Phi_1^{\mathrm{init}}}\otimes\ket{\Phi_{\mathrm{E1}}^{\mathrm{init}}}\otimes
\ket{\Phi_{\mathrm{E2}}^{\mathrm{init}}}\otimes\ket{\Phi_3^{\mathrm{init}}}\otimes
\ket{\Phi_4^{\mathrm{init}}}\otimes\ket{\Phi_5^{\mathrm{init}}}$$
and $\ket{\Phi_{\mathrm{T}}[(a,b,e,f,c,m,d)]}$ denotes
$$\ket{\Phi_1[ab]}\otimes\ket{\Phi_{\mathrm{E1}}[e]}\otimes\ket{\Phi_{\mathrm{E2}}[f]}\otimes
\ket{\Phi_3[c]}\otimes\ket{\Phi_4[m]}\otimes\ket{\Phi_5[d]}$$
for each $(a,b,e,f,c,m,d)\in\{0,1\}^7$.
Note that $E^1_{ab}, W_{ab},E^3_{c}, E^4_{mc}$, and $E^5_{d}$ are the same as before.
Totally, prior to the application of $U_{\mathrm{T}}$, the six qubit system
$\mathcal{Q}_1, \mathcal{Q}_2, \mathcal{Q}_{\mathrm{E1}}, \mathcal{Q}_3$, and $\mathcal{Q}_5$ is
prepared in the state
$$\ket{\Psi_{\mathrm{T}}^{\mathrm{init}}}:=\sum_{(a,b,e,c,d)\in\{0,1\}^5}\frac{1}{4}A_d\ket{ab0ecd},$$
where
$A$'s are the same
as before,
and $\ket{ab0ecd}$ denotes the six qubit state
$\ket{a}\otimes\ket{b}\otimes\ket{0}\otimes\ket{e}\otimes\ket{c}\otimes\ket{d}$
for each $(a,b,e,c,d)\in\{0,1\}^5$.

In the execution of Protocol~\ref{BB84wsm} with the eavesdropping by Eve,
Steps 1--5 of Protocol~\ref{BB84wsm} added with the Steps E1 and E2
only involve quantum mechanics, and this part of the protocol is repeated infinitely often in any case.
The operator $U_{\mathrm{T}}$ above describes the repeated once of
the
infinite repetition.
We use $\Omega_{\mathrm{T}}$ to denote the alphabet $\{0,1\}^7$.
It is
easy to check that a collection
\begin{equation}\label{BB84withevemo}
  \left\{E^1_{ab}\otimes E^4_{mc}E^{\mathrm{E2}}_{fe}W_{ab} \otimes E^{\mathrm{E1}}_{e}\otimes
  E^3_{c}\otimes E^5_d\right\}_{(a,b,e,f,c,m,d)\in\Omega_{\mathrm{T}}}
\end{equation}
forms measurement operators.
Thus,
the sequential application $U_{\mathrm{T}}$ of $U_1,U_2,U_{\mathrm{E1}},U_{\mathrm{E2}},U_3,U_4,U_5$
can be regarded as the \emph{single measurement} which is described by
the measurement operators \eqref{BB84withevemo}.
Hence, we can apply Definition~\ref{pmrpwst}
to this scenario
of
the setting of measurements.
Therefore, according to Definition~\ref{pmrpwst},
we can see that a \emph{world} is an infinite sequence over
$\Omega_{\mathrm{T}}$
and the probability measure induced by
the
\emph{measure representation for the prefixes of
worlds}
is
a Bernoulli measure $\lambda_{R}$ on
$\Omega_{\mathrm{T}}^\infty$,
where $R$ is a finite probability space on
$\Omega_{\mathrm{T}}$
such that
$R(a,b,e,f,c,m,d)$ is the square of the norm of
the
state
\begin{equation*}%
  \left(\left(E^1_{ab}\otimes E^4_{mc}E^{\mathrm{E2}}_{fe}W_{ab} \otimes E^{\mathrm{E1}}_{e}\otimes
  E^3_{c}\otimes E^5_d\right)
  \ket{\Psi_{\mathrm{T}}^{\mathrm{init}}}\right)\otimes\ket{\Phi_{\mathrm{T}}[(a,b,e,f,c,m,d)]}
\end{equation*}
for every $(a,b,e,f,c,m,d)\in\Omega_{\mathrm{T}}$.
Let us calculate the
explicit
form of $R(a,b,e,f,c,m,d)$.
The direct application of $U_1,U_2,U_{\mathrm{E1}},U_{\mathrm{E2}},U_3,U_4,U_5$ in this order to the state
$\ket{ab0ecd}\otimes\ket{\Phi_{\mathrm{T}}^{\mathrm{init}}}$
leads to the following:
\begin{enumerate}
\item
In the case of $b=c=e$, we have
\begin{align*}
U_{\mathrm{T}}\ket{ab0ecd}\otimes\ket{\Phi_{\mathrm{T}}^{\mathrm{init}}}
=\ket{ab}\otimes\ket{\Psi_{ab}}\otimes\ket{ecd}\otimes\ket{\Phi_{\mathrm{T}}[(a,b,e,a,c,a,d)]}.
\end{align*}
\item In the case of $e=b\neq c$, we have
\begin{align*}
U_{\mathrm{T}}\ket{ab0ecd}\otimes\ket{\Phi_{\mathrm{T}}^{\mathrm{init}}}
&=\frac{1}{\sqrt{2}}\ket{ab}\otimes \ket{\Psi_{0c}}\otimes\ket{ecd}\otimes
\ket{\Phi_{\mathrm{T}}[(a,b,e,a,c,0,d)]} \\
&+\frac{(-1)^a}{\sqrt{2}}\ket{ab}\otimes \ket{\Psi_{1c}}\otimes\ket{ecd}\otimes
\ket{\Phi_{\mathrm{T}}[(a,b,e,a,c,1,d)]}.
\end{align*}
\item
In the case of $b\neq c=e$, we have
\begin{align*}
U_{\mathrm{T}}\ket{ab0ecd}\otimes\ket{\Phi_{\mathrm{T}}^{\mathrm{init}}}
&=\frac{1}{\sqrt{2}}\ket{ab}\otimes\ket{\Psi_{0e}}\otimes\ket{ecd}\otimes
\ket{\Phi_{\mathrm{T}}[(a,b,e,0,c,0,d)]} \\
&+\frac{(-1)^a}{\sqrt{2}}\ket{ab}\otimes\ket{\Psi_{1e}}\otimes\ket{ecd}\otimes
\ket{\Phi_{\mathrm{T}}[(a,b,e,1,c,1,d)]}.
\end{align*}
\item
In the case of $b=c\neq e$, we have
\begin{align*}
U_{\mathrm{T}}\ket{ab0ecd}\otimes\ket{\Phi_{\mathrm{T}}^{\mathrm{init}}}
&=\frac{1}{2}\ket{ab}\otimes\ket{\Psi_{0c}}\otimes\ket{ecd}\otimes
\ket{\Phi_{\mathrm{T}}[(a,b,e,0,c,0,d)]} \\
&+\frac{1}{2}\ket{ab}\otimes\ket{\Psi_{1c}}\otimes\ket{ecd}\otimes
\ket{\Phi_{\mathrm{T}}[(a,b,e,0,c,1,d)]} \\
&+\frac{(-1)^a}{2}\ket{ab}\otimes\ket{\Psi_{0c}}\otimes\ket{ecd}\otimes
\ket{\Phi_{\mathrm{T}}[(a,b,e,1,c,0,d)]} \\
&-\frac{(-1)^a}{2}\ket{ab}\otimes\ket{\Psi_{1c}}\otimes\ket{ecd}\otimes
\ket{\Phi_{\mathrm{T}}[(a,b,e,1,c,1,d)]}.
\end{align*}
\end{enumerate}
Here, $\ket{ecd}$ denotes the three qubit state $\ket{e}\otimes\ket{c}\otimes\ket{d}$.
The comparison of
the results (i)--(iv) above
with the equation \eqref{BB84withEve-roirm},
where $\ket{\Psi}$ is set to $\ket{\Psi_{\mathrm{T}}^{\mathrm{init}}}$,
leads to
the following:
\begin{enumerate}
\item In the case of $b=c=e$, we have that
  \begin{equation}\label{b=c=e-f=a=mPabefcmd=Ad2-16}
    R(a,b,e,f,c,m,d)=\frac{A_d^2}{16}\delta_{a,f}\delta_{a,m}.
  \end{equation}
\item In the case of $b=c\neq e$, we have that
  \begin{equation}\label{b=cneqePabefcmd=Ad2-64}
    R(a,b,e,f,c,m,d)=\frac{A_d^2}{64}.
  \end{equation}
\end{enumerate}

Now, let us apply Postulate~\ref{POT}, the principle of typicality, to the setting of measurements
developed above.
Let $\omega_{\mathrm{T}}$ be our world in the infinite repetition of the measurements
in the above setting.
Then $\omega_{\mathrm{T}}$ is an infinite sequence over $\Omega_{\mathrm{T}}$.
Since the Bernoulli measure $\lambda_{R}$ on $\Omega_{\mathrm{T}}^\infty$ is
the probability measure induced by the
measure representation
for the prefixes of
worlds
in the above setting,
it follows from Postulate~\ref{POT} that
\emph{$\omega_{\mathrm{T}}$ is Martin-L\"of $R$-random}.

Let $D_{\mathrm{T}}:=\{(a,b,e,f,c,m,d)\in\Omega_{\mathrm{T}}\mid b=c\;\&\; d=1\}$.
It is
the set of all records of the apparatuses
$\mathcal{A}_1,\mathcal{A}_{\mathrm{E1}},\mathcal{A}_{\mathrm{E2}},\mathcal{A}_3,\mathcal{A}_4$,
and $\mathcal{A}_5$,
in a repeated once of the procedure in Protocol~\ref{BB84wsm} with the eavesdropping by Eve, where
\emph{Bob measured a qubit in the same basis as Alice prepared it
and moreover the detection of eavesdropping was attempted}.
It follows from \eqref{b=c=e-f=a=mPabefcmd=Ad2-16}
and \eqref{b=cneqePabefcmd=Ad2-64} that
\begin{equation}\label{PEveSEve=sumPEvefp2}
  R(D_{\mathrm{T}})=
  \sum_{(a,b,e,f,m)\in\{0,1\}^5} R(a,b,e,f,b,m,1)=\frac{p}{2}.
\end{equation}
\emph{Intuitively},
the real $R(D_{\mathrm{T}})$ can be interpreted as \emph{the probability that
Alice and Bob proceed
beyond
Step 10 in Protocol~\ref{BB84wsm} to check the presence of
eavesdropping}.
The explicit form $p/2$ of $R(D_{\mathrm{T}})$ given in \eqref{PEveSEve=sumPEvefp2}
is just as expected from the point of view of the conventional quantum mechanics.

Since $0<p<1$, note from \eqref{PEveSEve=sumPEvefp2} that
$R(D_{\mathrm{T}})>0$, in particular.
Let $\delta_{\mathrm{T}}:=\cond{D_{\mathrm{T}}}{\omega_{\mathrm{T}}}$.
It is an infinite sequence over $D_{\mathrm{T}}$ obtained from $\omega_{\mathrm{T}}$
by eliminating all elements of the form $(a,b,e,f,c,m,d)$ with $b\neq c$ or $d=0$.
Since $\omega_{\mathrm{T}}$ is Martin-L\"of $R$-random,
using Theorem~\ref{conditional_probability} we have that
$\delta_{\mathrm{T}}$ is Martin-L\"of $R_{D_{\mathrm{T}}}$-random
for the finite probability space $R_{D_{\mathrm{T}}}$ on $D_{\mathrm{T}}$.
Note that
\begin{equation}\label{BB84withEveRD}
  R_{D_{\mathrm{T}}}(a,b,e,f,c,m,d)=\frac{R(a,b,e,f,b,m,1)}{R(D_{\mathrm{T}})}
\end{equation}
for every $(a,b,e,f,c,m,d)\in D_{\mathrm{T}}$.
Thus, for each $(a,b,e,f,c,m,d)\in D_{\mathrm{T}}$,
using \eqref{PEveSEve=sumPEvefp2}, \eqref{b=c=e-f=a=mPabefcmd=Ad2-16},
and \eqref{b=cneqePabefcmd=Ad2-64} we see the following:
\begin{enumerate}
\item In the case of $b=c=e$, we have that
  \begin{equation}\label{b=c=e-f=a=mRDTabefcmd=1-8dd}
    R_{D_{\mathrm{T}}}(a,b,e,f,c,m,d)=\frac{1}{8}\delta_{a,f}\delta_{a,m}.
  \end{equation}
\item In the case of $b=c\neq e$, we have that
  \begin{equation}\label{b=cneqeRDTabefcmd=1-32}
    R_{D_{\mathrm{T}}}(a,b,e,f,c,m,d)=\frac{1}{32}.
  \end{equation}
\end{enumerate}

We consider a subset $E_{\mathrm{T}}$ of $D_{\mathrm{T}}$ defined by
$$E_{\mathrm{T}}:=\{(a,b,e,f,c,m,d)\in D_{\mathrm{T}}\mid a\neq m\}.$$
It is
the set of all records of the apparatuses
$\mathcal{A}_1,\mathcal{A}_{\mathrm{E1}},\mathcal{A}_{\mathrm{E2}},\mathcal{A}_3,\mathcal{A}_4$,
and $\mathcal{A}_5$,
in a repeated once of the procedure in Protocol~\ref{BB84wsm} with the eavesdropping by Eve, where
\emph{Alice and Bob proceeded beyond
Step 10 in Protocol~\ref{BB84wsm} and
moreover
they actually
detected the presence of eavesdropping
by means of the criterion $a\neq m$}.
It follows from \eqref{b=c=e-f=a=mRDTabefcmd=1-8dd} and \eqref{b=cneqeRDTabefcmd=1-32} that
\begin{equation}\label{RDTETf1-4}
  R_{D_{\mathrm{T}}}(E_{\mathrm{T}})=\frac{1}{4}.
\end{equation}
On the other hand, since $\delta_{\mathrm{T}}$ is Martin-L\"of $R_{D_{\mathrm{T}}}$-random, it follows
from Theorem~\ref{charaA} that
$\chara{E_{\mathrm{T}}}{\delta_{\mathrm{T}}}$ is
Martin-L\"of $\charaps{R_{D_{\mathrm{T}}}}{E_{\mathrm{T}}}$-random.
Here, the infinite binary sequence $\chara{E_{\mathrm{T}}}{\delta_{\mathrm{T}}}$ can be phrased as
\emph{an infinite sequence obtained from $\delta_{\mathrm{T}}$
by replacing each records of the apparatuses
$\mathcal{A}_1,\mathcal{A}_{\mathrm{E1}},\mathcal{A}_{\mathrm{E2}},\mathcal{A}_3,\mathcal{A}_4$,
and $\mathcal{A}_5$,
occurring in
$\delta_{\mathrm{T}}$,
by $1$ if the records
indicate the detection of the eavesdropping and by $0$ otherwise}.
Note
that
$(\charaps{R_{D_{\mathrm{T}}}}{E_{\mathrm{T}}})(1)=R_{D_{\mathrm{T}}}(E_{\mathrm{T}})=1/4$
due to \eqref{RDTETf1-4}.
Therefore, using Theorem~\ref{FI} we can show that
\begin{equation}\label{limfNnn=f1-4}
  \lim_{n\to\infty} \frac{M(n)}{n}=\frac{1}{4},
\end{equation}
where $M(n)$ is the number of the occurrences of $1$ in the first $n$ bits of
the infinite binary sequence $\chara{E_{\mathrm{T}}}{\delta_{\mathrm{T}}}$.
In other words,
$M(n)$ is the number of the success of Alice and Bob in the detection of the eavesdropping
in the first $n$ repetitions of the procedure in Protocol~\ref{BB84wsm}
\emph{in our world},
conditioned that Alice and Bob proceed beyond
Step 10 in Protocol~\ref{BB84wsm} to check the presence of eavesdropping.
Thus, the result \eqref{limfNnn=f1-4} can be interpreted as that \emph{the ``frequency'' that
Alice and Bob actually detect the presence of eavesdropping is equal to $1/4$
in our world,
conditioned that they proceed beyond
Step 10 in Protocol~\ref{BB84wsm} to check the presence of eavesdropping}.

Finally, we note that, from the point of view of the conventional quantum mechanics,
the real $R_{D_{\mathrm{T}}}(E_{\mathrm{T}})$ can be interpreted as \emph{the probability that
Alice and Bob actually detect the
eavesdropping,
conditioned that they proceed beyond
Step 10 in Protocol~\ref{BB84wsm} to check the presence of eavesdropping}.
This probability \emph{should} be $1/4$ in the analysis of the original BB84 QKD protocol \cite{BB84}
from the aspect of the conventional quantum mechanics
(and actually we have $R_{D_{\mathrm{T}}}(E_{\mathrm{T}})=1/4$
in
\eqref{RDTETf1-4} in our framework).
Thus,
the limit frequency $1/4$ calculated in \eqref{limfNnn=f1-4}
is just as expected from the point of view of the conventional quantum mechanics.

\section{Concluding remarks}
\label{Concluding}

In this paper we have introduced
an operational refinement of the Born rule for pure states, Postulate~\ref{Tadaki-rule}, and
an operational refinement of the Born rule for mixed states, Postulate~\ref{Tadaki-rule2},
based on the notion of Martin-L\"of randomness with respect to Bernoulli measure,
for specifying the property of the results of quantum measurements \emph{in an operational way}.

Then, in the framework of MWI,
we have introduced
the \emph{principle of typicality}, Postulate~\ref{POT},
in order to overcome the deficiency of
the original MWI \cite{E57}.
The principle of typicality is a postulate naturally formed by applying
the basic idea of Martin-L\"of randomness
into
the framework of MWI.
Based on the principle of typicality,
in this paper we have made the \emph{whole} arguments by Everett~\cite{E57} \emph{clear}
and \emph{feasible}.
Actually, we have shown in the framework of MWI that
the refined rule
of the Born rule for pure states (i.e., Postulate~\ref{Tadaki-rule}),
the refined rule of the Born rule for mixed states (i.e., Postulate~\ref{Tadaki-rule2}), and
other postulates of the conventional quantum mechanics regarding mixed states and density matrices
(i.e., Postulates~\ref{COSMS} and \ref{density-matrices-probability})
can \emph{all} be derived,
in a natural setting of measurements for each of them, from
the principle of typicality
together with Postulates~\ref{state_space}, \ref{composition}, and~\ref{evolution}
\emph{in a unified manner}.

Based on the results above,
we conjecture that
the principle of typicality together with
Postulates~\ref{state_space}, \ref{composition}, and~\ref{evolution}
\emph{forms}
quantum mechanics.
Thus, the principle of typicality
is thought
to be a \emph{unifying principle} which \emph{refines}
the Born rule for pure states (i.e., Postulate~\ref{Born-rule}),
the Born rule for mixed states (i.e., Postulate~\ref{Born-rule2}),
and other postulates of the conventional quantum mechanics regarding mixed states and density matrices
(i.e., Postulates~\ref{COSMS} and \ref{density-matrices-probability})
\emph{in one lump}.

In this paper, for simplicity,
we have considered only the case of finite-dimensional quantum systems and measurements over them.
As the next step of the research,
it is natural to consider the case of infinite-dimensional quantum systems,
and measurements over them where the set of possible measurement outcomes is
\emph{countably infinite}.
Actually,
in this case
we can also develop
a framework for an operational refinement of the Born rule and the principle of typicality 
in \emph{almost the same manner} as the finite case developed in this paper.
A full paper which describes the detail of
the
results is in preparation.

\section*{Acknowledgments}

This work was partially supported by JSPS KAKENHI Grant Numbers
24540142, 15K04981.
This work was partially done while the author was visiting the Institute for Mathematical Sciences,
National University of Singapore in 2014.


\end{document}